\documentclass[review]{elsarticle}

\usepackage{hyperref}
\usepackage{amsmath,amssymb,amsfonts}
\usepackage{amsthm}
\usepackage{mathrsfs}
\usepackage{algorithm}
\usepackage{algorithmic}
\usepackage{graphicx}
\usepackage{textcomp}
\newtheorem{theorem}{Theorem}
\newtheorem{definition}{Definition}

\newtheorem{lemma}{Lemma}
\newtheorem{example}{Example}
\newtheorem{fact}{Fact}
\newtheorem{proposition}{Proposition}
\newtheorem{assumption}{Assumption}
\newtheorem{remark}{Remark}

\usepackage{mathtools}
\usepackage{enumitem}
\mathtoolsset{showonlyrefs}
\usepackage{tikz}

\usetikzlibrary{shapes}
\usetikzlibrary{calc}
\usetikzlibrary{decorations.pathreplacing}
\usetikzlibrary{arrows,positioning} 
\usetikzlibrary{pgfplots.groupplots}
\usetikzlibrary{decorations.text}

\tikzstyle{link}=[line width=2pt, ->,>=latex]
\tikzstyle{junc}=[draw,circle,inner sep=1pt,minimum width=8pt]
\tikzstyle{onramp}=[line width=2pt, dashed,->,>=latex]

\DeclareMathOperator*{\argmin}{arg\,min}
\DeclareMathOperator*{\argmax}{arg\,max}

\makeatletter
\DeclareRobustCommand\widecheck[1]{{\mathpalette\@widecheck{#1}}}
\def\@widecheck#1#2{%
    \setbox\z@\hbox{\m@th$#1#2$}%
    \setbox\tw@\hbox{\m@th$#1%
       \widehat{%
          \vrule\@width\z@\@height\ht\z@
          \vrule\@height\z@\@width\wd\z@}$}%
    \dp\tw@-\ht\z@
    \@tempdima\ht\z@ \advance\@tempdima2\ht\tw@ \divide\@tempdima\thr@@
    \setbox\tw@\hbox{%
       \raise\@tempdima\hbox{\scalebox{1}[-1]{\lower\@tempdima\box
\tw@}}}%
    {\ooalign{\box\tw@ \cr \box\z@}}}
\makeatother

\makeatletter
\def\blfootnote{\gdef\@thefnmark{}\@footnotetext}
\makeatother

\journal{the Nonlinear Analysis: Hybrid Systems Journal}

\begin{document}

\begin{frontmatter}

\title{Abstraction-based Synthesis for Stochastic Systems with Omega-Regular Objectives}

\address[mymainaddress]{School of Electrical and Computer Engineering}
\address[mysecondaryaddress]{School of Electrical and Computer Engineering and School of Civil and Environmental Engineering}

\author[mymainaddress]{Maxence Dutreix\corref{mycorrespondingauthor}}
\ead{maxdutreix@gatech.edu}
\author[mymainaddress]{Jeongmin Huh}
\ead{jhuh32@gatech.edu}
\author[mysecondaryaddress]{Samuel Coogan}
\ead{sam.coogan@gatech.edu}
\address{The Georgia Institute of Technology, Atlanta, GA, USA}


%\author[mysecondaryaddress]{Maxence Dutreix, Jeongmin Huh}
\cortext[mycorrespondingauthor]{Corresponding author}


\begin{abstract}
\blfootnote{This project was supported in part by the NSF under project \#1749357.}This paper studies the synthesis of controllers for discrete-time, continuous state stochastic systems subject to omega-regular specifications using finite-state abstractions. Omega-regular properties allow specifying complex behaviors and encompass, for example, linear temporal logic. First, we present a synthesis algorithm for minimizing or maximizing the probability that a discrete-time switched stochastic system with a finite number of modes satisfies an omega-regular property. Our approach relies on a finite-state abstraction of the underlying dynamics in the form of a Bounded-parameter Markov Decision Process arising from a finite partition of the system's domain. Such Markovian abstractions allow for a range of probabilities of transition between states for each selected action representing a mode of the original system. Our method is built upon an analysis of the Cartesian product between the abstraction and a Deterministic Rabin Automaton encoding the specification of interest or its complement. Specifically, we show that synthesis can be decomposed into a qualitative problem, where the so-called greatest permanent winning components of the product automaton are created, and a quantitative problem, which requires maximizing the probability of reaching this component in the worst-case instantiation of the transition intervals. Additionally, we propose a quantitative metric for measuring the quality of the designed controller with respect to the continuous abstracted states and devise a specification-guided domain partition refinement heuristic with the objective of reaching a user-defined optimality target. Next, we present a method for computing control policies for stochastic systems with a continuous set of available inputs. In this case, the system is assumed to be affine in input and disturbance, and we derive a technique for solving the qualitative and quantitative problems in the resulting finite-state abstractions of such systems. For this, we introduce a new type of abstractions called Controlled Interval-valued Markov Chains. Specifically, we show that the greatest permanent winning component of such abstractions are found by appropriately partitioning the continuous input space in order to generate a bounded-parameter Markov decision process that accounts for all possible qualitative transitions between the finite set of states. Then, the problem of maximizing the probability of reaching these components is cast as a (possibly non-convex) optimization problem over the continuous set of available inputs. A metric of quality for the synthesized controller and a partition refinement scheme are described for this framework as well. Finally, we present a detailed case study.
\end{abstract}


\end{frontmatter}

%\linenumbers

\section{Introduction}

The need for systems that are both complex and reliable is more critical than ever. Not only are the models describing these systems becoming increasingly complicated, but the tasks they are expected to perform also continue to grow in complexity. For example, the operating specification may combine an invariance and a reachability condition and require that \emph{the system will always return to a good state while always avoiding a bad state.} Such specifications can be formally and unambiguously represented as, for instance, a \textit{Linear Temporal Logic} (LTL)  \cite{pnueli1977temporal} specification, among other classes of symbolic languages. In this paper, we consider the class of \textit{$\omega$-regular properties} \cite{thomas1990automata}, a superset of LTL.

Recent research efforts in formal verification and synthesis have focused on the development of robust controllers to ensure that systems requirements are unequivocally met for broad classes of specifications and dynamics \cite{tabuada2006linear} \cite{kloetzer2008fully} \cite{yordanov2011temporal} \cite{rungger2013specification} \cite{rungger2016scots}  \cite{liu2016finite} \cite{sadraddini2018formal}. A general approach is to obtain a (non)deterministic finite abstraction of the continuous-state system, encode the specification as an appropriate transition system called an automaton, compute a product construction between the system abstraction and the automaton, and then synthesize a controller by solving graph-based problems on the product \cite{belta2017formal} \cite{baier2008principles}. The controller obtained from the finite abstraction is then mapped onto the original abstracted states. However, this basic recipe does not immediately work for stochastic systems because the random disturbances acting upon such systems add a quantitative component to the transitions between states in the form of transition probabilities, preventing the use of standard transition systems as finite abstractions for this framework. Typically, this limitation is overcome by using probabilistic finite transition systems as abstractions for stochastic systems \cite{bujorianu2005bisimulation} \cite{abate2008markov} \cite{abate2011approximate} \cite{kwiatkowska2011prism} \cite{lahijanian2015formal}. Even though general synthesis procedures for such abstractions inherit ideas from approaches proposed in non-stochastic settings, the mathematical machinery required is quite different.

%While the application of formal methods in purely deterministic or nondeterministic settings resulted in the implementation of dependable control tools \cite{holzmann1997model} \cite{larsen1997uppaal}, efficient techniques suitable for systems with random dynamics remain to be explored.

Indeed, for stochastic systems, satisfaction of a specification may never be fully guaranteed due to randomness. Therefore, the synthesis problem requires finding a control policy which maximizes or minimizes the probability of occurrence of some desired behavior from a given initial condition. In this work, we consider the problem of synthesizing a control policy for a discrete-time, continuous-state stochastic system subject to an $\omega$-regular specification. \textcolor{black}{Although the existence of optimal policies for this problem is not known, we seek to devise policies which are satisfactory with respect to a reasonable metric of quality}. First, we consider the case when the control action is selected from a finite set of modes that the system can switch between at each time step. Then we consider the case when the control action is selected from a continuous set of possible inputs.

Recent literature demonstrated the effectiveness of \textit{Bounded-Parameter Markov Decision Processes} (BMDP) as a tool for the synthesis of control policies in stochastic systems \cite{lahijanian2015formal} \cite{givan2000bounded}. Indeed, BMDPs are naturally amenable to finite-state abstractions of switched stochastic systems constructed from a finite partition of the continuous system domain. As each discrete state abstracts the behavior of an uncountably infinite number of underlying continuous states, the probabilities of transition between states are specified as intervals for each mode of the BMDP, rather than just a single number as in standard Markov Decision Processes. Solving for an optimal switching policy in the BMDP abstraction results in a near-optimal policy for the objective of maximizing or minimizing the probability of satisfying the specification with respect to the original abstracted states. The \textcolor{black}{quality} of this policy with respect to the original system states naturally depends on the quality and fineness of the continuous domain partition from which the abstraction is constructed.

In \cite{lahijanian2015formal}, the authors present an algorithm for computing switching policies that either minimize or maximize the probability of satisfying Probabilistic Computation Tree Logic (PCTL) specifications in a BMDP. The theory developed in \cite{lahijanian2015formal} has been applied to linear systems with additive Gaussian noise subject to cosafe LTL specifications and was shown to be computationally efficient \cite{2019arXiv190101576C}. This BMDP-based technique was also recently implemented in the comprehensive verification and synthesis toolbox StocHy \cite{cauchi2019stochy}. However, PCTL and cosafe LTL are strictly less expressive than the $\omega$-regular logic and cannot articulate certain important liveness and persistence properties, such as the infinite repetition of some event \cite{rozier2011linear}. A similar problem was solved in \cite{wolff2012robust} for LTL specifications, but the proposed solution makes simplifying assumptions on the connectivity properties of the system's abstraction which drastically reduces its scope of applicability. The synthesis of control strategies for interval Markov decision processes with multi-objectives that include $\omega$-regular properties was discussed in \cite{hahn2019interval}; however, the qualitative structure of the transition system is again assumed to be invariant, which alleviates key difficulties associated with the problem.

In this paper, we implement a procedure for computing switching policies in finite-mode discrete-time stochastic systems with the objective of minimizing or maximizing the probability of occurrence of any $\omega$-regular property. We first create a partition of the continuous domain from which a BMDP abstraction of the system is generated. We then consider the Cartesian product between the BMDP abstraction and a Deterministic Rabin Automaton (DRA) representing the $\omega$-regular property of interest \textcolor{black}{for the maximization problem, or the complement of the property for the minimization problem}. We prove that any such product BMDP induces a largest set of so-called \textit{Permanent Winning Component} for a subset of all possible switching policies, and show that the probability maximization and minimization problems reduce to a reachability maximization task on these sets of states in the product BMDP. Note that our approach does not necessitate any assumption on the connectivity structure of the BMDP unlike in \cite{wolff2012robust} and \cite{hahn2019interval}. Furthermore, we introduce a quantitative measure capturing the \textcolor{black}{quality} of the switching policy designed in the BMDP abstraction when mapped onto the continuous abstracted states with respect to the objective of minimizing or maximizing the probability of fulfilling some specification in the original system. Finally, we propose a partition refinement technique inspired by our method in \cite{dutreix2020specification}, which considered only the verification problem without inputs, in order to reach a desired level of optimality for the computed policy with respect to the continuous system states and progressively discard control actions which are guaranteed to be suboptimal. \textcolor{black}{While no formal proof of the convergence of this technique is provided in this article, such refinement-based heuristics have shown to work remarkably well in practice and offer advantages in terms of scalability.}

Expanding on the theory for finite-mode systems, we address the problem of synthesizing controllers for stochastic systems with $\omega$-regular objectives from a continuous set of available inputs using finite-state abstractions. Related works discussed the synthesis of controllers for continuous input stochastic systems subject to subsets of $\omega$-regular properties, such as Büchi objectives \cite{2019arXiv191012137M}, using abstraction-based methods. Here, we specifically study the class of stochastic systems which are affine-in-disturbance and affine-in-input. We introduce \textit{Controlled Interval-valued Markov Chains} (CIMC), which serve as abstractions for continuous input systems. We present an algorithm for constructing the largest permanent winning components in the product between a CIMC and a DRA. Then, we show that the reachability maximization step on these components can be formulated as an optimization program. The \textcolor{black}{quality} of the designed policy with respect to the original abstracted system and state-space refinement are discussed as well in this framework.

In brief, the novel contributions of \textcolor{black}{this article over existing works, and in particular over our work on the verification of stochastic systems in \cite{dutreix2020specification},} are as follows: 
\textcolor{black}{
\begin{itemize}
\item We present a synthesis procedure for finite-mode discrete-time stochastic systems against $\omega$-regular specifications, implemented in Algorithm \ref{AlgFiniteMode}. Our approach employs BMDP abstractions constructed from a partition of the continuous domain of the system, and we devise an automaton-based synthesis algorithm for BMDPs against $\omega$-regular specifications from the results of Theorem \ref{TheoWorstCasePol} in conjunction with Algorithms \ref{AlgPerBSCCW} to \ref{AlgPerWin}. These algorithms perform a search of specific components of a BMDP which do not exist in abstractions without control actions along with the computation of policies generating these components, and therefore are more involved than the graph search algorithms found in \cite{dutreix2020specification}. The switching policy synthesized in the BMDP abstraction is then mapped onto the continuous abstracted states.
\item We introduce a quantitative measure of the quality of the policy computed from the BMDP abstraction with respect to the original abstracted system states. The results in \cite{dutreix2020specification} are not concerned with the computation of switching policies and therefore do not propound such a measure. This metric is determined from the facts highlighted in Theorem \ref{UpBoundMaxTheo}.
\item We develop a specification-guided refinement strategy on the partition of the system domain in Algorithm \ref{ScorAlgSyn} to enhance the quality of the switching policy in refined BMDP abstractions of the dynamics. While an algorithm is presented in \cite{dutreix2020specification} for verification that is similar in spirit, major differences are found in the input of both algorithms, their termination criteria and the computations performed to select the states to be refined. This work additionally discusses some properties which are passed from coarser to refined abstractions, which is not done in \cite{dutreix2020specification}.
\item Next, we extend the techniques above to synthesize controllers for affine-in-disturbance, affine-in-input stochastic systems with a continuous set of permissible inputs. The control policy is computed by means of CIMC abstractions constructed from a partition of the system domain and mapped onto the continuous abstracted states as detailed in Algorithm \ref{ContInputSyn}. To this end, we present a synthesis procedure for CIMC abstractions arising from stochastic systems with the aforementioned structure that relies on Algorithms \ref{InpSelecAlg} and \ref{CompConsAlg} and requires solving (possibly non-convex) optimization problems.
\item For such systems with continuous input sets, we propose a refinement scheme for the domain partition to improve the quality of the computed controller with respect to the original abstracted states.
\end{itemize}
}

The paper is organized as follows: Section 2 introduces some preliminaries; Section 3 formulates the problem to be solved; Section 4 describes our controller synthesis strategy for finite-mode stochastic systems; Section 5 presents a controller synthesis algorithm for stochastic systems with a continuous set of inputs; Section 6 shows a case study; Section 7 concludes our work.

\section{Preliminaries}
\label{prelimNAHS}

A \textit{Deterministic Rabin Automaton (DRA)} \cite{baier2008principles} is a 5-tuple $\mathcal{A} = (S, \textcolor{black}{\Pi}, \delta, s_0, Acc)$ where:
\begin{itemize}
\item $S$ is a finite set of states,
\item $\textcolor{black}{\Pi}$ is an alphabet,
\item $\delta : S \times \textcolor{black}{\Pi} \rightarrow S$ is a transition function,
\item $s_0$ is an initial state,
\item $Acc \subseteq 2^{S} \times 2^{S}$. An element $(E_{i}, F_{i}) \in Acc$, with $E_{i}, F_{i} \, \textcolor{black}{\subseteq} \,  S$, is called a \textit{Rabin Pair}.
\end{itemize}

A DRA $\mathcal{A}$ reads an infinite string or \textit{word} over alphabet $\textcolor{black}{\Pi}$ as an input and transitions from state to state according to $\delta$. The resulting sequence of states or \textit{run} is an \textit{accepting} run if \textcolor{black}{some states of $F_i$ are visited infinitely often and all states of $E_i$ are visited finitely often for some $i$}. A word is said to be accepted by $\mathcal{A}$ if it produces an accepting run in $\mathcal{A}$. We call a set of words a \textit{property}. The property \textit{accepted} by $\mathcal{A}$ is the set of all words accepted by $\mathcal{A}$.\\

A property over an alphabet $\textcolor{black}{\Pi}$ is \textit{$\omega$-regular} if and only if it is accepted by a Rabin Automaton with alphabet $\textcolor{black}{\Pi}$ (for more detailed definitions of $\omega$-regular properties, see \cite[Section 4.3.1]{baier2008principles}). In particular, all properties defined by a \textit{Linear Temporal Logic} (LTL) formula are $\omega$-regular. For example, the property ``Eventually reach A", written in LTL as $\Diamond A$, has an equivalent $\omega$-regular expression representation $(\neg A)^{*}A(\textcolor{black}{\Pi})^{\omega}$, where $*$ and $\omega$ are respectively the finite and infinite repetition operators. See \cite{baier2008principles} for a detailed description of the syntax and semantics of LTL.\\

\textcolor{black}{A $S \times S$ matrix is a \textit{transition matrix} $M$ if $\Sigma_{j=1}^{S} M_{i,j} = 1$ for all $i = 1, 2, \ldots, S$, where $M_{i,j}$ is the $i$th row and $j$th column element of $M$.}

\section{Problem Formulation}

We first consider the discrete-time, continuous-state stochastic system
\begin{align}
x[k+1] = \mathcal{F}_{a}(x[k], w_a[k])
\label{eq1}
\end{align}
where $x[k] \in D \subset \mathbb{R}^{n}$ is the state of the system at time $k$, $a \in A$ where $A$ is a finite set of \textit{modes}, $w_{a}[k] \in W_{a} \subset \mathbb{R}^{p_a}$ is a random disturbance (which could be mode-dependent), $\mathcal{F}_{a} : D \times W_{a} \rightarrow D$ is a continuous map. Let $L: D \rightarrow \Sigma$ be a labeling function, where $\Sigma$ is a finite alphabet \textcolor{black}{and such that, for all $\sigma \in \Sigma$, the subset $D_{\sigma}\subseteq D$ of all states $x \in D$ satisfying $L(x) = \sigma$ can be written as a finite union of subsets of $D$, that is, $D_{\sigma} = \cup_{i = 1}^{N} J_{i}, \, J_{i} \subseteq D, \,  n \in \mathbb{N}$}. In Section \ref{Continputsec}, we extend this setup to allow for an infinite set of modes, \emph{i.e.}, a control input selected from a continuous set of inputs. An infinite random path $x[0] x[1] \ldots$ satisfying \eqref{eq1} generates the word $L(x[1]) L(x[2]) \ldots$ over $\Sigma$. At each time-step $k$, a mode $a \in A$ is chosen and the random disturbance $w_{a}[k]$ is sampled from a probability distribution with probability density function $f_{w_{a}}:\mathbb{R}^{p_{a}}\to \mathbb{R}_{\geq 0}$ satisfying $f_{w_{a}}(z)=0$ if $z\not\in W_{a}$. Then, a transition from state $x[k]$ to state $x[k+1]$ takes place according to the dynamics defined by mode $a$. \textcolor{black}{The set of all infinite paths of \eqref{eq1} is denoted by $Paths$.}
A finite sequence of states $\pi = x[0] x[1] \ldots x[n]$ produced by \eqref{eq1} is called a \textit{finite path}. The set of all finite paths of \eqref{eq1} is denoted by $Paths_{fin}$. A function $\mu : Paths_{fin} \rightarrow A$ assigning a mode to each finite path in \eqref{eq1}  is called a \textit{switching policy} and the set of all switching policies of \eqref{eq1} is denoted by $\mathcal{U}=\{ \mu \mid \mu: Paths_{fin} \rightarrow A\}$. For simplicity, we assume that all modes of $A$ are available at each state of $D$. \textcolor{black}{A policy $\mu \in \mathcal{U}$ induces a probability measure $Prob_{\mu}$ on the outcome space of infinite paths $Paths$ of \eqref{eq1}, where $Prob_{\mu}$ is defined by the stochastic transition kernel $T_{D} : D \times A \times \mathcal{R}(D) $ assigning a probability measure to any state $x \in D$ and mode $a \in A$ on the space $(D, \mathcal{R}(D))$, with $\mathcal{R}(D)$ denoting the Borel $\sigma$-algebra on $D$ and such that $Pr(x,a | \mathcal{K}) = \int_{\mathcal{A}} T_{D}(x,a | dx)$,  $\mathcal{K} \in \mathcal{R}(D)$  \cite{hernandez1996discrete}.}

We denote by $\Psi$ an arbitrary $\omega$-regular property over alphabet $\Sigma$ and write as $(p^x_{\Psi})_{\mu}$ the probability that a word generated by a random path starting in $x$ satisfies property $\Psi$ under policy $\mu$ (for a rigorous formalization of this probability, see, e.g., \cite{abate2011approximate}). Our objective is to determine switching policies $\widecheck{\mu}_{\Psi}$ and $\widehat{\mu}_{\Psi}$ that respectively minimize and maximize the probability of satisfying property $\Psi$ for any path in the system and, by extension, for any initialization to $x$ of the system.\\

\textbf{Problem 1}: \textit{Given a system of the form \eqref{eq1}, any initial state $x \in D$ and an $\omega$-regular property $\Psi$, find switching policies $\widecheck{\mu}_{\Psi} \in \mathcal{U}$ and $\widehat{\mu}_{\Psi} \in \mathcal{U}$ that respectively minimize and maximize the probability of satisfying $\Psi$ from $x$, \emph{i.e.},}
\begin{align}
\widecheck{\mu}_{\Psi} & = \argmin_{\mu \in \mathcal{U}} (p^x_{\Psi})_{\mu}\\
\widehat{\mu}_{\Psi}  & = \argmax_{\mu \in \mathcal{U}} (p^x_{\Psi})_{\mu} \ . 
\end{align}\mbox{}

For complex specifications and dynamics, devising these exact optimal policies is likely to be intractable or infeasible due to the uncountably infinite number of states of the system's domain. To determine a policy which is close to optimal, we consider an abstraction-based approach that consists in partitioning $D$ into a finite collection of states $P$ to construct a finite abstraction of the stochastic dynamics.\\

\begin{definition}[Partition]
A \textit{partition} $P$ of a domain $D \subset \mathbb{R}^{n}$ is a collection of discrete states $P = \lbrace Q_{j} \rbrace_{j=1}^{m}, \; Q_{j} \subset D,$ satisfying
\begin{itemize}
\item $\bigcup_{j=1}^{m} Q_{j} = D$,
\item $  \textbf{int}(Q_{j}) \cap \textbf{int}(Q_{\ell}) = \emptyset \;\; \forall j, \ell, \;  j \not = \ell \ ,$
\end{itemize}
where \textbf{int} denotes the interior. For any continuous state $x$ belonging to a state $Q_j$, we write $x \in Q_j$.
\end{definition}\mbox{}

\noindent For a partition $P$ of the domain $D$ of \eqref{eq1}, the likelihood of transitioning from a state $Q_j$ of $P$ to another state $Q_{\ell}$ generally varies with the continuous state abstracted by $Q_{j}$ from which the transition is actually taking place. Therefore, we cannot use partition $P$ to exactly abstract the system into a standard finite-mode Markovian model, such as an MDP. Instead, we propose producing a \textit{BMDP abstraction} of the system where, for each action of the BMDP abstracting the behavior of \eqref{eq1} under some mode, the transition probabilities between states are constrained within some bounds, as depicted in Figure \ref{FigBMDPAbs}.\\

\begin{definition}[Bounded-parameter Markov Decision Process]
A \emph{Bounded-parameter Markov Decision Process (BMDP)} \cite{givan2000bounded} is a 6-tuple $\mathcal{B} = (Q, Act, \widecheck{T}, \widehat{T}, \\ \textcolor{black}{q_0}, \Sigma, L)$ where:
\begin{itemize}
\setlength{\itemsep}{0pt}
\item $Q$ is a finite set of states,
\item $Act$ is a finite set of actions, and the set of actions available at state $Q_{j} \in Q$ is denoted by $A(Q_{j}) \subseteq Act$,
%\item $q_{0}$ is the initial state,
\item $\widecheck{T}: Q \times Act \times Q \rightarrow [0, 1] $ maps pairs of states and an action to a lower transition bound so that $\widecheck{T}_{Q_{j} \xrightarrow{a} Q_{\ell}} := \widecheck{T}(Q_{j}, a, Q_{\ell})$ denotes the lower bound of the transition probability from state $Q_{j}$ to state $Q_{\ell}$ under action $a \in A(Q_j)$, and 
\item $\widehat{T}: Q \times Act \times Q \rightarrow [0, 1] $ maps pairs of states and an action to an upper transition bound so that $\widehat{T}_{Q_{j} \xrightarrow{a} Q_{\ell}} := \widehat{T}(Q_{j}, a, Q_{\ell})$ denotes the upper bound of the transition probability from state $Q_{j}$ to state $Q_{\ell}$ under action $a \in A(Q_j)$,
\item \textcolor{black}{$q_{0} \subseteq Q$ is a set of initial states,}
\item $\Sigma$ is a finite set of atomic propositions,
\item $L : Q \rightarrow \textcolor{black}{2^{\Sigma}}$ is a labeling function from states to \textcolor{black}{the power set of} $\Sigma$,
\end{itemize}
and $\widecheck{T}$ and $\widehat{T}$ satisfy $\widecheck{T}(Q_j, a, Q_\ell)\leq \widehat{T}(Q_j, a, Q_\ell)$ for all $Q_j,Q_\ell\in Q$, all $a \in A(Q_j)$, and 
\begin{equation}
\label{eq:37}
\sum_{Q_\ell\in Q} \widecheck{T}(Q_j, a, Q_\ell)\leq 1\leq \sum_{Q_\ell\in Q}\widehat{T}(Q_j, a, Q_\ell)  
\end{equation}
 for all $Q_j\in Q$ and all $a \in A(Q_j)$.
 \end{definition}\mbox{}
 
 \begin{definition}[BMDP Abstraction]
Given the system \eqref{eq1} evolving on a domain $D \subset  \mathbb{R}^{n}$ and a partition $P=\{Q_j\}_{j=1}^m$ of $D$, a BMDP $\mathcal{B} = (Q, Act,  \widecheck{T}, \widehat{T}, \textcolor{black}{q_0}, \Sigma, L)$ is an \emph{abstraction} of \eqref{eq1} if:  
\begin{itemize}
\item $Q := P$, that is, the set of states of the BMDP is the partition $P$, 
%\item $Q_{i} \in P \Rightarrow  Q_{i} \in \mathscr{I}$ 
\item $Act := A$, that is, the set of actions of the BMDP are the modes of \eqref{eq1},
\item For all $Q_j,Q_\ell\in P$ and action $a \in Act$,
\begin{align}
\label{eq:3}   \widecheck{T}_{Q_{j} \xrightarrow{a} Q_{\ell}}  &\leq  \inf_{x\in Q_j} Pr( \mathcal{F}_{a} (x,w_{a}) \in Q_\ell ),\text{ and}\\
\label{eq:3-2}   \widehat{T}_{Q_{j} \xrightarrow{a} Q_{\ell}}  &\geq \sup_{x\in Q_j} Pr( \mathcal{F}_{a} (x,w_{a}) \in Q_\ell ),
\end{align}
where $Pr( \mathcal{F}_{a}(x,w_{a}) \in Q_\ell)$ for fixed $x$ denotes the probability that \eqref{eq1} transitions from $x$ to some state $x'=\mathcal{F}_{a}(x,w_{a})$ in $Q_{\ell}$ under mode $a$,
\item \textcolor{black}{$P = q_0$, i.e., the set of initial states of the BMDP is the partition $P$,}
\item For all $Q_j \in P$ and for any two states $x_i, x_{\ell} \in Q_j$, it holds that $L(Q_j) := L(x_i) = L(x_{\ell})$, that is, the partition conforms to the boundaries induced by the labeling function.
\end{itemize}
%The IMC abstraction $\mathcal{I}$ is said to be \emph{tight} if \eqref{eq:3} and \eqref{eq:3-2} hold with equality.
\end{definition}\mbox{}

For a given action, two continuous states belonging to the same discrete state of a BMDP abstraction $\mathcal{B}$ may, in general, give rise to different transition probabilities. This fact is encoded in $\mathcal{B}$ by the upper and lower transition probabilities.

\begin{figure}
\centering
\includegraphics[scale=0.35]{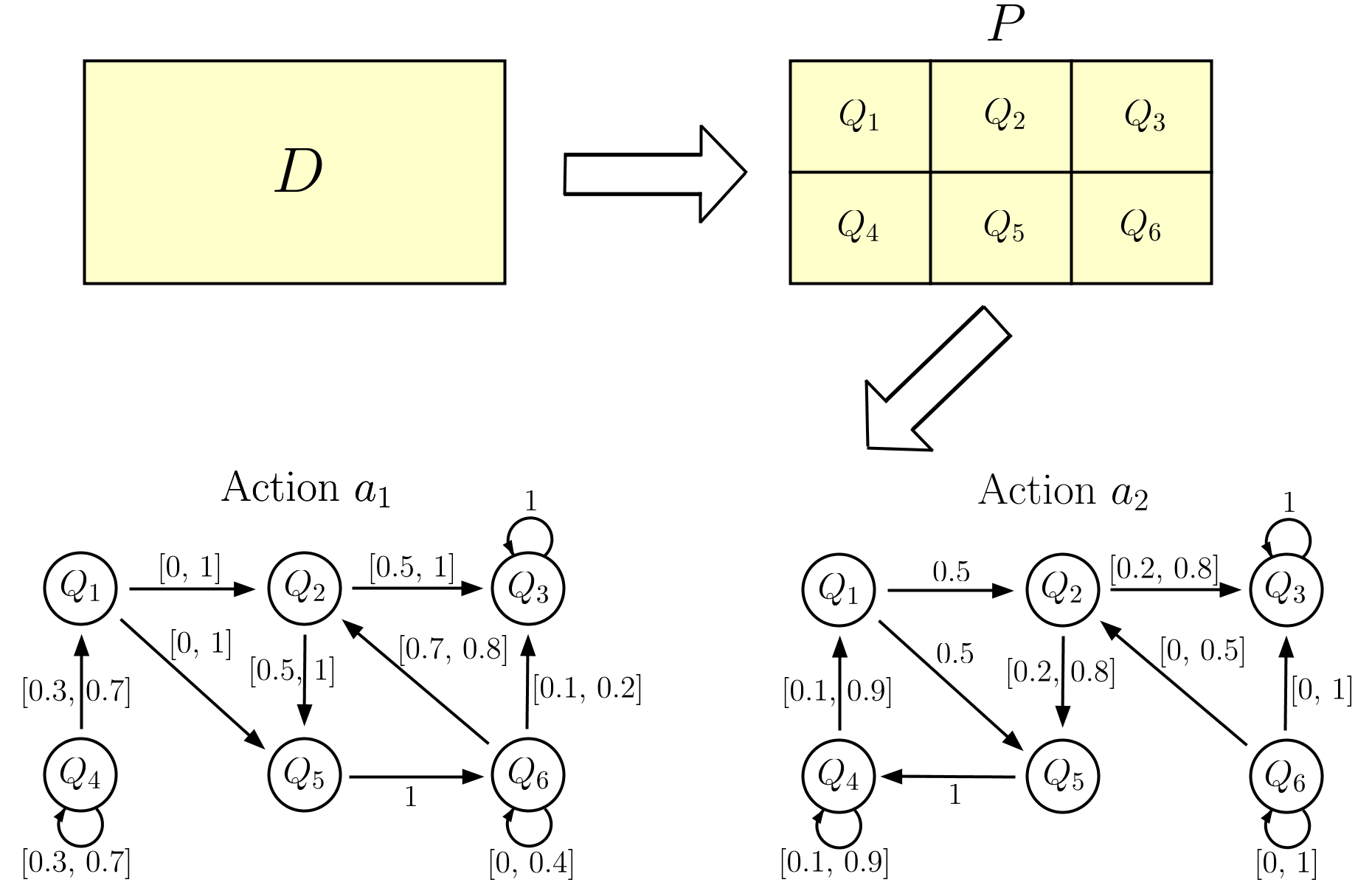}
\caption{A finite-state BMDP abstraction $\mathcal{B}$ of system \eqref{eq1} with domain $D$. A partition $P$ of $D$ is generated and bounds on the transition probabilities between states are estimated for two actions $a_{1}$ and $a_{2}$ of $\mathcal{B}$.}
\label{FigBMDPAbs}
\end{figure}

In this paper, we do not present algorithms for computing BMDP abstraction of \eqref{eq1}, which typically rely on overapproximating reachable sets; see \cite{dutreix2018} for such an approach. Thus, we assume that BMDP abstractions are available given a partition $P$ of $D$ for \eqref{eq1}. However, we will focus on the problem of refining $P$ in order to obtain better BMDP abstractions.
 
Furthermore, we make the assumption that any state in $Q$ of a BMDP can serve as an initial state. Denoting the set of all finite paths of a BMDP $\mathcal{B}$ by $(Paths_{fin})_{\mathcal{B}}$, a switching policy $\mu : (Paths_{fin})_{\mathcal{B}} \rightarrow Act$ for $\mathcal{B}$ is a function assigning an action to all finite paths in $\mathcal{B}$. The set of all switching policies of $\mathcal{B}$ is denoted by $\mathcal{U}_{\mathcal{B}} = \{ \mu \mid \mu : (Paths_{fin})_{\mathcal{B}} \rightarrow Act \}$. Under a switching policy $\mu$, the available actions in BMDP $\mathcal{B}$ reduce to a single possibility at each time step, namely, that prescribed by the switching policy $\mu$, inducing a (possibly countably infinite-state) \textit{Interval-valued Markov Chain} (IMC), defined formally next. As will be discussed further, only finite-memory policies need to be considered in this work, which induce finite-state IMCs.\\
 
\begin{definition}[Interval-valued Markov Chain]
An \emph{Interval-valued Markov Chain} (IMC) $\mathcal{I} = (Q, \widecheck{T}, \widehat{T}, \textcolor{black}{q_0}, \Sigma, L)$ is defined similarly to a BDMP with the difference that a single action (which is omitted in the defining tuple) is available.
\end{definition}\mbox{}

\noindent The IMC induced by policy $\mu$ in BMDP $\mathcal{B}$ is denoted by $\mathcal{B}[\mu]$.

The state of an IMC $\mathcal{I}$ evolves as follows: at each time step $k$, the environment non-deterministically chooses a transition matrix $T_k$ compatible with the transition bound functions $\widecheck{T}$ and $\widehat{T}$ of $\mathcal{I}$ and the next transition occurs according to $T_k$ \cite{sen2006model} \footnote{This is the Interval Markov Decision Process interpretation of IMCs.}. A mapping $\mathcal{\nu}$ from a finite path $\pi = q_0 \ldots q_k$ in $\mathcal{I}$ to a transition matrix $T_k$ is called an \textit{adversary}. The set of all adversaries of $\mathcal{I}$ is denoted by $\mathcal{\nu}_{\mathcal{I}}$. \textcolor{black}{A unique probability measure $Prob_{\mathcal{\nu}}$ is induced over the set of all infinite paths $Paths_{\mathcal{I}}$ of IMC $\mathcal{I}$ under adversary $\mathcal{\nu} \in \mathcal{\nu}_{\mathcal{I}}$ \cite[Def. 10.10]{baier2008principles}. By extension, a probability measure $Prob_{\mu, \mathcal{\nu}}$ is induced over the set of all infinite paths $Paths_{\mathcal{B}}$ of BMDP $\mathcal{B}$ under policy $\mu$ and adversary $\mathcal{\nu} \in \mathcal{\nu}_{\mathcal{B}[\mu]}$}.

The probability of satisfying $\omega$-regular property $\Psi$ starting from initial state $Q_j$ in IMC $\mathcal{I}$ under adversary $\mathcal{\nu}$ is denoted by $\mathcal{P}_{\mathcal{I}[\mathcal{\nu}]}(Q_j \models \Psi)$.
The greatest lower bound and least upper bound on the probability of satisfying property $\Psi$ starting from initial state $Q_j$ in IMC $\mathcal{I}$ are denoted by $\widecheck{\mathcal{P}}_{\mathcal{I}}(Q_j \models \Psi) = \inf_{\nu \in \nu_{\mathcal{I}}} \mathcal{P}_{\mathcal{I}[\mathcal{\nu}]}(Q_j \models \Psi)$ and $\widehat{\mathcal{P}}_{\mathcal{I}}(Q_j\models \Psi) = \sup_{\nu \in \nu_{\mathcal{I}}} \mathcal{P}_{\mathcal{I}[\mathcal{\nu}]}(Q_j \models \Psi)$ respectively. When these bounds are the same for all states in a set of states $C$ of $\mathcal{I}$, we write $\widecheck{\mathcal{P}}_{\mathcal{I}}(C\models \Psi)$ and $\widehat{\mathcal{P}}_{\mathcal{I}}(C\models \Psi)$. \\

\textcolor{black}{To design switching policies in BMDPs, it is crucial to note that a BMDP $\mathcal{B}$ subject to a switching policy $\mu$ reduces to an IMC $\mathcal{B}[\mu]$; therefore, finding the probability of satisfying a specification $\Psi$ from some initial state of $\mathcal{B}[\mu]$ amounts to solving a verification problem on an IMC. As discussed above, the probability of satisfying a specification $\Psi$ in an IMC is not uniquely defined and depends on the instantiation of a non-deterministic adversary. Consequently, the verification of the IMC $\mathcal{B}[\mu]$ induced by a policy $\mu$ in a BMDP $\mathcal{B}$ does not compute, in general, a fixed probability but an interval of satisfaction probabilities $(I_{j})_{\mu} = [(p^{j}_{min})_{\mu}, (p^{j}_{max})_{\mu}]$ for all initial states $Q_{j}$ of $\mathcal{B}[\mu]$. The meaning of this interval is that the probability of fulfilling $\Psi$ from state $Q_{j}$ in $\mathcal{B}[\mu]$ is contained in $(I_{j})_{\mu}$ for all possible adversaries of $\mathcal{B}[\mu]$, that is, $\mathcal{P}_{\mathcal{B}[\mu][\mathcal{\nu}]}(Q_j \models \Psi) \in (I_{j})_{\mu}, \; \forall \nu \in \nu_{\mathcal{B}[\mu]}$.} 

\textcolor{black}{Because a switching policy in a BMDP returns an interval of satisfaction for all its initial states, it may not seem obvious which quantities to minimize or maximize when synthesizing policies in BMDP abstractions of continuous state systems. Note that a policy $\mu$ for a BMDP abstraction $\mathcal{B}$ of \eqref{eq1} maps to a policy for \eqref{eq1} in the natural way, \emph{i.e.}, the control action prescribed by $\mu$ at a discrete state $Q_i$ of $\mathcal{B}$ is applied to all continuous states $x\in Q_i$ in \eqref{eq1}. By virtue of $\mathcal{B}$ being an abstraction of \eqref{eq1}, it then holds that the exact probability of satisfying $\Psi$ from any continuous initial state $x\in Q_j$ for \eqref{eq1} is contained within the bounds of the interval $(I_{j})_{\mu}$ induced by policy $\mu$ for initial state $Q_j$ in $\mathcal{B}$ \cite{lahijanian2015formal}.
Therefore, given a BMDP abstraction $\mathcal{B}$ of \eqref{eq1} generated from a partition $P$ of the domain $D$, our approach to Problem 1 is to find policies $\widehat{\mu}^{low}_{\Psi}$ and $\widecheck{\mu}^{up}_{\Psi}$ in $\mathcal{B}$ that respectively maximize the lower bound probability (for the maximization objective) and minimize the upper bound probability (for the minimization objective) of satisfying $\Psi$ for all initial states $Q_j$ of $\mathcal{B}$.}\\

\textbf{Subproblem 1.1}: \textit{Given a system of the form \eqref{eq1}, a partition $P$ of its domain $D$, a BMDP abstraction $\mathcal{B}$ of \eqref{eq1} arising from $P$, any initial state $Q_j \in Q$ of $\mathcal{B}$ and an $\omega$-regular property $\Psi$, compute switching policies $\widecheck{\mu}^{up}_{\Psi} \in \mathcal{U}_{\mathcal{B}}$ and  $\widehat{\mu}^{low}_{\Psi} \in \mathcal{U}_{\mathcal{B}}$ that respectively minimize the upper bound probability and maximize the lower bound probability of satisfying $\Psi$ in $\mathcal{B}$, \emph{i.e.},}
\begin{align}
\widecheck{\mu}^{up}_{\Psi} & = \argmin_{\mu \in \mathcal{U}_{\mathcal{B}}} \widehat{\mathcal{P}}_{\mathcal{B}[\mu]}(Q_j \models \Psi)\\
\widehat{\mu}^{low}_{\Psi}  & = \argmax_{\mu \in \mathcal{U}_{\mathcal{B}}} \widecheck{\mathcal{P}}_{\mathcal{B}[\mu]}(Q_j \models \Psi) \ . 
\end{align}\mbox{}

If $\mathcal{B}$ is a BMDP abstraction of \eqref{eq1}, then a unique control action is assigned to all continuous states abstracted by some $Q_{i}$ in $\mathcal{B}$. In this case, the \textcolor{black}{quality} of the policies $\widehat{\mu}^{low}_{\Psi}$ and $\widecheck{\mu}^{up}_{\Psi}$ heavily depends on the quality and fineness of the partition $P$ of the domain $D$. Indeed, because these policies only accommodate the extreme behaviors of all discrete states of $\mathcal{B}$, it is reasonable to assume that the computed policies may be suboptimal for a collection of continuous states abstracted by some $Q_{i}$. In this work, we address this problem by starting with a coarse partition of the system's domain; then, we iteratively and selectively refine this partition so as to target discrete states that are at a higher risk of containing suboptimally controlled continuous states or are responsible for considerable uncertainty in the control of other states. As finer partitions result in larger abstractions to be analyzed, it is crucial to avoid performing unnecessary refinement in order to alleviate the state-space explosion phenomenon. The procedure terminates once a precision threshold which will be defined in further sections has been reached.\\

\textbf{Subproblem 1.2}: \textit{Given a system of the form \eqref{eq1} with a BMDP abstraction $\mathcal{B}$ arising from a partition $P$ of the domain $D$ and an $\omega$-regular property $\Psi$, refine the partition $P$ of $D$ until the computed switching policy reaches a user-defined threshold of \textcolor{black}{quality} with respect to the objective of minimizing or maximizing the probability of satisfying $\Psi$ in \eqref{eq1}.}\\

After presenting solutions to Subproblem 1.1 and 1.2 in Section \ref{SecSynthFin}, we next investigate stochastic systems of the form 
\begin{align}
x[k+1] = \mathcal{F}(x[k], u[k], w[k])
\label{eq2}
\end{align}
\noindent where $x[k] \in D \subset \mathbb{R}^{n}$ is the state of the system at time $k$, $u[k] \in U$ where $U \subset \mathbb{R}^{m}$ is a continuous set of inputs, $w[k] \in W \subset \mathbb{R}^{p}$ is a random disturbance whose probability density function $f_{w}$ is assumed to be independent of $u$, $\mathcal{F} : D \times U \times W \rightarrow D$ is a continuous map. Here, a \textit{control policy} is a function $\mu : Paths_{fin} \rightarrow U$ assigning a control action to each finite path in \eqref{eq2}. The set of all control policies of \eqref{eq2} is denoted by $\mathcal{U}=\{ \mu \mid \mu: Paths_{fin} \rightarrow A\}$ as in the finite-mode system case.

The difficulty of establishing policies aiming to maximize or minimize the probability of satisfying a temporal property in \eqref{eq2} is highly dependent on the structure of the considered system. In this work, we restrict our attention to systems which are affine in input and disturbance, that is
\begin{align}
x[k+1] = \mathcal{F}(x[k]) + u[k] + w[k] \ .
\label{eq3}
\end{align}

As in the finite-mode case, we are interested in the design of a control policy that maximizes or minimizes the probability of satisfying an $\omega$-regular property $\Psi$.\\

\textbf{Problem 2}: \textit{Given a system of the form \eqref{eq3}, any initial state $x \in D$ and an $\omega$-regular property $\Psi$, find control policies $\widecheck{\mu}_{\Psi} \in \mathcal{U}$ and $\widehat{\mu}_{\Psi} \in \mathcal{U}$ that respectively minimize and maximize the probability of satisfying $\Psi$ from $x$.}\\

Solving this problem for an arbitrary property $\Psi$ again involves a partition $P$ of the domain $D$ from which a finite-state abstraction of the system is constructed and analyzed. In this work, we introduce new abstraction tools called \textit{Controlled Interval-valued Markov Chains} (CIMC) which differ from BMDPs in that the set of available actions is uncountably infinite. CIMCs are the abstractions of choice for systems of the form \eqref{eq3}.\\

\begin{definition}[Controlled Interval-valued Markov Chain]
A \emph{Controlled Interval-valued Markov Chain (CIMC)} is a 6-tuple $\mathcal{C} = (Q, U, \widecheck{T}, \widehat{T}, \textcolor{black}{q_{0}}, \Sigma, L)$ defined similarly to a BMDP with the difference that a continuous set of inputs $U \subseteq \mathbb{R}^{m}$ replaces the finite set of actions $Act$.
\end{definition}\mbox{}

\begin{definition}[Controlled Interval-valued Markov Chain Abstraction]
Given the system \eqref{eq3} evolving on a domain $D \subset  \mathbb{R}^{n}$ and a partition $P=\{Q_j\}_{j=1}^m$ of $D$, a CIMC $\mathcal{C} = (Q, U,  \widecheck{T}, \widehat{T}, \textcolor{black}{q_{0}}, \Sigma, L)$ is an \emph{abstraction} of \eqref{eq3} if it satisfies the same conditions as a BMDP abstraction with the difference that a continuous set of inputs $U \subseteq \mathbb{R}^{m}$ replaces the finite set of actions $Act$.
\end{definition}\mbox{}

Denoting the set of all finite paths in a CIMC $\mathcal{C}$ by $(Paths_{fin})_{\mathcal{C}}$, a control policy $\mu : (Paths_{fin})_{\mathcal{C}} \rightarrow U$ for $\mathcal{C}$ is a function assigning an input to all finite paths in $\mathcal{C}$. The set of all control policies of $\mathcal{C}$ is denoted by $\mathcal{U}_{\mathcal{C}} = \{ \mu \mid \mu : (Paths_{fin})_{\mathcal{C}} \rightarrow U \}$. A policy $\mu$ applied to a CIMC $\mathcal{C}$ induces an IMC denoted by $\mathcal{C}[\mu]$.

Computing an optimal policy in a CIMC abstraction translates to computing a near-optimal policy when the former is applied to the original abstracted system. Thus, for all possible finite paths in $\mathcal{C}$, the goal is to find the input in the uncountable set $U$ that yields the most favorable IMC abstraction with respect to the desired objective. Note that, unlike in a BMDP abstraction, this problem offers an infinite set of available inputs to select from, ruling out the possibility of using an exhaustive search.\\

\textbf{Subproblem 2.1}: \textit{Given a system of the form \eqref{eq3}, a partition $P$ of its domain $D$, a CIMC abstraction $\mathcal{C}$ of \eqref{eq3} arising from $P$, any initial state $Q_j \in Q$ of $\mathcal{C}$ and an $\omega$-regular property $\Psi$, compute the control policies $\widecheck{\mu}^{up}_{\Psi} \in \mathcal{U}_{\mathcal{C}}$ and  $\widehat{\mu}^{low}_{\Psi} \in \mathcal{U}_{\mathcal{C}}$ that respectively minimize the upper bound probability and maximize the lower bound probability of satisfying $\Psi$ in $\mathcal{C}$, \emph{i.e.},}
\begin{align}
\widecheck{\mu}^{up}_{\Psi} & = \argmin_{\mu \in \mathcal{U}_{\mathcal{C}}} \widehat{\mathcal{P}}_{\mathcal{C}[\mu]}(Q_j \models \Psi)\\
\widehat{\mu}^{low}_{\Psi}  & = \argmax_{\mu \in \mathcal{U}_{\mathcal{C}}} \widecheck{\mathcal{P}}_{\mathcal{C}[\mu]}(Q_j \models \Psi) \ . 
\end{align}\mbox{}

As our approach again relies on finite-state abstractions, finer partitions of the domain $D$ generally yield \textcolor{black}{higher-quality} control policies. Therefore, partition refinement for this case is discussed as well.\\

\textbf{Subproblem 2.2}: \textit{Given a system of the form \eqref{eq3} with a CIMC abstraction $\mathcal{C}$ arising from a partition $P$ of the domain $D$ and an $\omega$-regular property $\Psi$, refine the partition $P$ of $D$ until the computed control policy reaches a user-defined threshold of \textcolor{black}{quality} with respect to the objective of minimizing or maximizing the probability of satisfying $\Psi$ in \eqref{eq3}.}\mbox{}\\

\textcolor{black}{In the next section, we comprehensively detail our solution to the synthesis of switching policies for finite mode systems as formalized in Problem 1. Specifically, Subsections \ref{SecBMDPSynth} and \ref{SecCompSearch} focus on the computation of controllers for BMDP abstractions as stated in Subproblem 1.1, whereas Subsection \ref{SecRefFin} is concerned with Subproblem 1.2 and the refinement of BMDP abstractions for the synthesis of improved policies with respect to the abstracted system.}

\section{CONTROLLER SYNTHESIS FOR FINITE MODE SYSTEMS}
\label{SecSynthFin}

\subsection{BMDP CONTROLLER SYNTHESIS}
\label{SecBMDPSynth}

In this subsection, we present the theory for addressing Subproblem 1.1. We adopt an automaton-based approach for computing maximizing and minimizing switching policies in a BMDP $\mathcal{B}$ with respect to an $\omega$-regular property $\Psi$. As discussed in Section \ref{prelimNAHS}, for every such property, there exists a corresponding DRA representation $\mathcal{A}$. Similar to \cite[page 798]{baier2008principles} and \cite{dutreix2020specification} where the Cartesian product with a Markov Chain (MC) and an IMC are introduced, we define the product $\mathcal{B} \otimes \mathcal{A}$ between a BMDP and a DRA.\\

\begin{definition}[Product Bounded-Parameter Markov Decision Process]
\label{defProductBMDP}
Let $\mathcal{B} = (Q, Act, \widecheck{T}, \widehat{T}, \textcolor{black}{q_{0}}, \Sigma, L)$ be a BMDP and $\mathcal{A} = (S, 2^{\Sigma}, \delta, s_0, Acc)$ be a DRA. The \emph{product} $\mathcal{B} \otimes \mathcal{A} = (Q \times S, Act, \widecheck{T'}, \widehat{T'}, \textcolor{black}{q^{\otimes}_{0}}, Acc', L')$ is a BMDP where:
\begin{itemize}
\item $Q \times S$ is a set of states,
\item $Act$ is the same set of actions of $\mathcal{B}$, where $A(\left<Q_{j}, s_{i} \right>) = A(Q_{j})$ for all $Q_{j} \in Q$ and for all $s_{i} \in S$,
\item $\widecheck{T'}_{ \left<Q_{j},s\right> \xrightarrow{a} \left<Q_{\ell},s'\right>} = 
\begin{cases}
\textcolor{black}{\widecheck{T}_{Q_{j} \xrightarrow{a} Q_{\ell}}}, \;\; \text{if} \;\; s' = \delta(s, L(Q_{\ell}))\\ \;\;\; \;\;\;\;\;\;0, \;\;\;\;\;\; \text{otherwise}
\end{cases}$\mbox{}\\\\
\item $\widehat{T'}_{ \left<Q_{j},s\right> \xrightarrow{a} \left<Q_{\ell},s'\right>} = 
\begin{cases}
\textcolor{black}{\widehat{T}_{Q_{j} \xrightarrow{a} Q_{\ell}}}, \;\; \text{if} \;\; s' = \delta(s, L(Q_{\ell}))\\ \;\;\; \;\;\;\;\;\;0, \;\;\;\;\;\; \text{otherwise}
\end{cases}$\mbox{}\\
\item \textcolor{black}{$ q^{\otimes}_{0} = \{(Q_j,s_0):Q_j\in Q\}$ is a finite set of initial states,}
\item $Acc' = \lbrace E_{1}, E_{2}, \ldots, E_{k}, F_{1}, F_{2}, \ldots, F_{k} \rbrace$ is a set of atomic propositions, where $E_{i}$ and $F_{i}$ are the sets in the Rabin pairs of $Acc$,
\item $L': Q \times S \rightarrow 2^{Acc'}$ such that\textcolor{black}{, for all atomic proposition $H \in Acc' $, for all $Q_{j} \in Q$ and for all $s_{i} \in S$,  $H \in L'(\left<Q_{j},s_{i}\right>)$ if and only if $s_{i}$ belongs to the set in the Rabin pairs of $Acc$ corresponding to $H$}.
\end{itemize}
\end{definition}\mbox{}

In this product construction, the DRA $\mathcal{A}$ is used as a finite-memory instrument that monitors all transitions occurring in $\mathcal{B}$ and assesses whether the resulting path satisfies $\Psi$. Indeed, any random path $\pi = q_0 q_1 \ldots$ in $\mathcal{B}$ generates a unique path $\pi_{\otimes}^{\mathcal{A}} = \left< q_0, s_0 \right> \left< q_1, s_ j \right> \ldots $ in $\mathcal{B} \otimes \mathcal{A}$ which depends on the labels of the states of $\mathcal{B}$ as per Definition \ref{defProductBMDP}. It follows that a switching policy in $\mathcal{B}$ can be induced by inspecting the sequences of states generated in $\mathcal{B} \otimes \mathcal{A}$ and choosing control actions accordingly.\\

\begin{definition}[Generated Path in Product BMDP]
Consider a BMDP $\mathcal{B}$ with set of states $Q$ and labeling function $L$ and a DRA $\mathcal{A}$ with set of states $S$ and transition function $\delta$. A path $\pi_{\otimes}^{\mathcal{A}} = \left< q_{0}, s'_{0} \right>, \left<q_{1}, s'_{1} \right> \ldots, \ q_{i} \in Q, \ s'_{i} \in S,$ in the product BMDP $\mathcal{B} \otimes \mathcal{A}$ is said to be \emph{generated} by the path $\pi = q_{0}, q_{1} \ldots$ in $\mathcal{B}$ if it holds that $s'_{i+1} = \delta(s'_{i}, L(q_{i+1})), \forall i = 0, 1, 2, \ldots \ $. 
\end{definition}\mbox{}

\begin{definition}[Induced Switching Policy]
Consider a BMDP $\mathcal{B}$, a DRA $\mathcal{A}$ and a switching policy $\mu \in \mathcal{U}_{\mathcal{B}}$. Let $\pi \in (Paths_{fin})_{\mathcal{B}}$ be any finite path in $\mathcal{B}$. We denote by $\pi_{\otimes}^{\mathcal{A}}$ the path generated by $\pi$ in the product BMDP $\mathcal{B} \otimes \mathcal{A}$. The switching policy $\mu$ is said to be \emph{induced} by a switching policy $\mu_{\otimes}$ of $\mathcal{B} \otimes \mathcal{A}$ if, for all $\pi \in  (Paths_{fin})_{\mathcal{B}}$, it holds that $\mu(\pi) = \mu_{\otimes}( \pi_{\otimes}^{\mathcal{A}})$.
\end{definition}\mbox{}

For a fixed switching policy $\mu$ of $\mathcal{B}$, the probability of satisfying $\Psi$ in the induced IMC $\mathcal{B}[\mu]$ is equal to the probability of reaching a so-called \textcolor{black}{Accepting} \textit{Bottom Strongly Connected Component} (BSCC) in the product IMC $\mathcal{B}[\mu] \otimes \mathcal{A}$ \cite{dutreix2020specification} defined below. The probability of reaching an accepting BSCC in $\mathcal{B}[\mu] \otimes \mathcal{A}$ is not uniquely defined and depends on the assumed transition values within the probability intervals selected by a non-deterministic adversary $\nu \in \nu_{\mathcal{B}[\mu] \otimes \mathcal{A}}$ which induces a product MC $\mathcal{B}[\mu][\nu]_{\otimes}^{\mathcal{A}}$.\\

\begin{definition}[Product Interval-valued Markov Chain]
\label{ProdInterDef}
Let $\mathcal{I} = (Q, \widecheck{T}, \widehat{T}, \textcolor{black}{q_{0}}, \Sigma, L)$ be an IMC and $\mathcal{A} = (S, 2^{\Sigma}, \delta, s_0, Acc)$ be a DRA. The \emph{product}  $\mathcal{I} \otimes \mathcal{A}  = (Q, \widecheck{T'}, \widehat{T'}, \textcolor{black}{q^{\otimes}_{0}}, Acc', L')$ is an IMC defined similarly to a product BDMP with the difference that a single action (which is omitted in the defining tuple) is available.
\end{definition}\mbox{}

\begin{definition}[Markov Chain]
A \textit{Markov Chain} (MC) $\mathcal{M} = (Q, T, \textcolor{black}{q_{0}}, \Sigma, L)$ is defined similarly to an IMC with the difference that the transition probability function or transition matrix \textcolor{black}{of the Markov Chain} $T: Q \times Q \rightarrow [0, 1] $ satisfies $0 \leq T(Q_j, Q_\ell) \leq 1$ for all $Q_j, Q_{\ell}\in Q$ and $\sum_{Q_{\ell} \in Q} T(Q_j, Q_\ell) = 1$ for all $Q_j\in Q$.
\end{definition}
\noindent The probability of satisfying property $\Psi$ in Markov Chain $\mathcal{M}$ from initial state $Q_{j}$ is denoted by $P_{\mathcal{M}}(Q_{j} \models \Psi)$.\\

\begin{definition}[Induced Product Markov Chain]
A Product Markov Chain $\mathcal{I}[\nu]_{\otimes}^{\mathcal{A}} = (Q \times S, T, \textcolor{black}{q^{\otimes}_{0}}, Acc', L')$  is said to be \emph{induced} by an adversary $\nu$ of a product IMC $\mathcal{I} \otimes \mathcal{A}$ if they share the same $Q$ (for memoryless policies $\mu$), $\mathcal{A}$, $\textcolor{black}{q^{\otimes}_{0}}$, $L'$ and $Acc'$, and for all $q_{j}$, $q_{\ell} \in Q \times S$ and all action $a = \mu(q_{j})$, the transition probability function $T$ satisfies $\widecheck{T}_{q_{j} \xrightarrow{a} q_{\ell}} \leq T(q_{j}, q_{\ell}) \leq \widehat{T}_{q_{j} \xrightarrow{a} q_{\ell}}$.
\end{definition}\mbox{}

\begin{definition}[Bottom Strongly Connected Component]
\label{BSCCdef}
Given a Markov Chain $\mathcal{M}$ with states $Q$, a subset $B \subseteq Q$ is called a \emph{Bottom Strongly Connected Component} (BSCC) of $\mathcal{M}$ if
\begin{itemize}
\item $B$ is strongly connected: for each pair of states $(q,t)$ in $B$, there exists a path $q_{0}q_{1}\ldots q_n$ such that $T(q_i,q_{i+1}) > 0$, $i = 0,1, \ldots, n-1$, and $q_i \in B$ for $0 \leq i \leq n$ with $q_0 = q,$ $q_n = t$,
\item no proper superset of $B$ is strongly connected,
\item $\forall s \in B$, $\Sigma_{t \in B}T(s,t) = 1$.
\end{itemize}
\end{definition}\mbox{}

In words, every state in a BSCC $B$ is reachable from any state in $B$, and every state in $B$ only transitions to another state in $B$.\\

\begin{definition}[Accepting and Non-Accepting Bottom Strongly Connected Component]
\label{accepcon}
A Bottom Strongly Connected Component $B$ of a product Markov Chain $\mathcal{M}_{\otimes}^{\mathcal{A}}$ is said to be \emph{accepting} if:
\begin{align}
\exists i: & \Bigg( \; \exists \left<Q_{j},s_{\ell} \right> \in B \; : \; F_{i} \in L'(\left<Q_{j},s_{\ell} \right>) \; \Bigg) \nonumber  \wedge \Bigg( \;  \forall \left<Q_{j},s_{\ell} \right> \in B \; : \; E_{i} \not \in L'(\left<Q_{j},s_{\ell} \right>) \; \Bigg).
\end{align}
$\mathcal{M}_{\otimes}^{\mathcal{A}}$ is said to be \emph{non-accepting} if it is not accepting.
\end{definition}\mbox{}

\noindent A key observation is that, for any policy $\mu$ in $\mathcal{B}$ induced by a policy $\mu_{\otimes}$ in the product $\mathcal{B} \otimes \mathcal{A}$, the bounds on the probability of reaching an accepting BSCC from the initial states of $\mathcal{B}[\mu] \otimes \mathcal{A}$ are identical to the bounds on the probability of reaching an accepting BSCC from the initial states of $(\mathcal{B} \otimes \mathcal{A})[\mu_{\otimes}]$ according to Definitions \ref{defProductBMDP} and \ref{ProdInterDef} which ensure that the elements in the defining tuples of $\mathcal{B}[\mu] \otimes \mathcal{A}$ and $(\mathcal{B} \otimes \mathcal{A})[\mu_{\otimes}]$ are the same. Consequently, an analysis of the product $\mathcal{B} \otimes \mathcal{A}$ is sufficient for approaching the synthesis problem.

\textcolor{black}{Because $\omega$-regular properties are closed under complementation, it should be noted that the problem of minimizing the upper bound probability of satisfying property $\Psi$ in $\mathcal{B}$ can be converted to the problem of maximizing the lower bound probability of satisfying the complement property $\overline{\Psi}$ with corresponding DRA $\overline{\mathcal{A}}$. It follows that Subproblem 1.1 is solved by applying the same tools to both $\mathcal{B} \otimes \mathcal{A}$ and $\mathcal{B} \otimes \overline{\mathcal{A}}$.}\\

\textcolor{black}{\begin{fact}
Let $\mathcal{B}$ be a BMDP and $\Psi$ be an $\omega$-regular specification. We denote by $\overline{\Psi}$ the complement of property $\Psi$. For any initial state $Q_{j} \in Q$ of $\mathcal{B}$ and policy $\mu \in \mathcal{U}_{\mathcal{B}}$, it holds that 
\begin{align*}
\widehat{\mathcal{P}}_{\mathcal{B}[\mu]}(Q_j \models \Psi) & = 1 -  \widecheck{\mathcal{P}}_{\mathcal{B}[\mu]}(Q_j \models \overline{\Psi})\\
\widecheck{\mathcal{P}}_{\mathcal{B}[\mu]}(Q_j \models \Psi) & = 1 -  \widehat{\mathcal{P}}_{\mathcal{B}[\mu]}(Q_j \models \overline{\Psi}) \ .
\end{align*}
\label{FactComp}
\end{fact}}

\textcolor{black}{Therefore}, our objective consists in computing a policy that maximizes the lower bound probability of reaching an accepting BSCC from all initial states of the resulting product IMC $\mathcal{B}[\mu] \otimes \mathcal{A}$. \textcolor{black}{We introduce the class of \textit{memoryless} policies, which solely depend on the current state of the BMDP and will further prove optimal for our problem in the product BMDP $\mathcal{B} \otimes \mathcal{A}$, and the class of memoryless adversaries of an IMC.}\\

\begin{definition}[Memoryless Policy]
A policy $\mu \in \mathcal{U}_{\mathcal{B}}$ of a BMDP $\mathcal{B}$ is said to be \emph{memoryless} if, for all finite paths $\pi = q[0] q[1] \ldots q[k]$ of $\mathcal{B}$, it holds that $\mu(\pi) = \mu(q[k])$.
\end{definition}\mbox{}

\begin{definition}[Memoryless Adversary]
An adversary $\nu \in \mathcal{I}_{\nu}$ of an IMC $\mathcal{I}$ is said to be \emph{memoryless} if, for all finite paths $\pi = q[0] q[1] \ldots q[k]$ of $\mathcal{I}$, it holds that $\nu(\pi) = \nu(q[k])$.
\end{definition}\mbox{}

Before presenting a solution to Subproblem 1.1, we first recall some basic results established in \cite{dutreix2020specification} for the purpose of verification in IMCs which we then extend to compute switching policies in BMDPs. 

For a given policy $\mu$ of $\mathcal{B}$ and automaton $\mathcal{A}$, the sets of accepting and non-accepting BSCCs of the resulting product IMC $\mathcal{B}[\mu] \otimes \mathcal{A}$ depend on the assumed \textcolor{black}{probability} values for the transitions with zero lower bound and non-zero upper bound. \textcolor{black}{Specifically, whether a zero or a non-zero value is assigned to these transitions directly affects the qualitative structure of the product IMC, and therefore its sets of accepting and non-accepting BSCCs, as a zero probability implies that a transition can never occur between the corresponding states, while a non-zero probability indicates that a transition is possible. When a non-zero probability is assumed for such a transition, we describe the transition as being ``on", and we say that the transition is ``off" in the scenario that a probability of zero is assumed. Nonetheless, it is shown in} \cite{dutreix2020specification} that, for any product IMC, there exists a \textit{largest winning component} and a \textit{largest losing component} which can be created among all combinations of ``on" and ``off" transitions allowed by the transition bound functions of the product IMC. A winning component of a product MC is a set of states that reach an accepting BSCC with probability 1, while a losing component is a set of states that reach a non-accepting BSCC with probability 1.\\

\begin{definition}[Winning Component]
\cite{dutreix2020specification} A \emph{winning component} $WC$ of a product MC $\mathcal{M}_{\otimes}^{\mathcal{A}}$ is a set of states satisfying $\mathcal{P}(WC \models \Diamond R) = 1$, where $R$ is the set of states belonging to an accepting BSCC in $\mathcal{M}_{\otimes}^{\mathcal{A}}$.
\label{defwin}
\end{definition}\mbox{}

\begin{definition}[Losing Component]
\cite{dutreix2020specification} A \emph{losing component} $LC$ of a product MC $\mathcal{M}_{\otimes}^{\mathcal{A}}$ is a set of states satisfying $\mathcal{P}(LC \models \Diamond R) = 1$, where $R$ is the set of states belonging to a non-accepting BSCC in $\mathcal{M}_{\otimes}^{\mathcal{A}}$.
\end{definition}\mbox{}

\begin{definition}[Largest Winning/Losing Components]
\cite{dutreix2020specification} A state $\left< Q_j, s_i \right> \in Q \times S$ of a product IMC $\mathcal{I}$ is a member of the \emph{Largest Winning} (respectively, \emph{Losing}) \emph{Component} $(WC)_{L} \ \big( \text{respectively,} \; (LC)_{L}  \big)$ if there exists a product MC induced by $\mathcal{I}$ such that $\left< Q_j, s_i \right>$ is a winning (respectively, losing) component.
\end{definition}\mbox{}

Moreover, it was shown in \cite{dutreix2020specification} that the upper bound probability of satisfying $\Psi$ in the IMC $\mathcal{I}$ from state $Q_j$ is equal to the upper bound probability of reaching the largest winning component $(WC)_{L}$ of the product $\mathcal{I} \otimes \mathcal{A}$ from state $\left< Q_j, s_0 \right>$. Likewise, the lower bound probability of satisfying $\Psi$ is found by solving a reachability problem on the largest losing component $(LC)_{L}$. These results naturally apply to product IMCs $\mathcal{B}[\mu] \otimes \mathcal{A}$ constructed from an IMC $\mathcal{B}[\mu]$ induced by a policy $\mu$ of a BMDP $\mathcal{B}$.\\

\begin{fact}[\cite{dutreix2020specification}]
Let $\mathcal{B}[\mu]$ be an IMC induced by a switching policy $\mu$ of a BMDP $\mathcal{B}$ and $\mathcal{A}$ be a DRA corresponding to the $\omega$-regular property $\Psi$. Let $(WC)_{L}$ and $(LC)_{L}$ be the largest winning and losing components of $\mathcal{B}[\mu] \otimes \mathcal{A}$ respectively. It holds that, for all initial states $Q_{j}$ of $\mathcal{B}[\mu]$,
\begin{align}
\widehat{\mathcal{P}}_{\mathcal{B}[\mu]} (Q_j \models \Psi) = & \ \widehat{\mathcal{P}}_{\mathcal{B}[\mu] \otimes \mathcal{A}}(\left< Q_j, s_0 \right> \models \Diamond (WC)_{L})\\
\widecheck{\mathcal{P}}_{\mathcal{B}[\mu]} (Q_j \models \Psi) = & \ 1 -  \widehat{\mathcal{P}}_{\mathcal{B}[\mu] \otimes \mathcal{A}}(\left< Q_j, s_0 \right> \models \Diamond (LC)_{L}).
\end{align}
\end{fact}\mbox{}

\noindent The intuitive interpretation of this property is that any IMC $\mathcal{B}[\mu]$ has a ``best-case" adversary and a ``worst-case" adversary in the product $\mathcal{B}[\mu] \otimes \mathcal{A}$ that respectively maximizes and minimizes the probability of reaching an accepting BSCC  for all initial states of $\mathcal{B}[\mu] \otimes \mathcal{A}$ simultaneously, since reachability probabilities are maximized by memoryless adversaries. These probabilities are equal to the upper bound and lower bound probabilities of satisfying $\Psi$ from the initial states of $\mathcal{B}[\mu]$. In \textcolor{black}{an} induced product MC corresponding to the best-case scenario, the set of winning components is as large as it can possibly be;  \textcolor{black}{in an induced product MC corresponding to the worst-case scenario, the set of winning components is reduced to the smallest possible set of \textit{permanent} winning components.}\\

\begin{definition}[Permanent Winning Components]
\cite{dutreix2020specification} A state $\left< Q_j, s_i \right> \in Q \times S$ of a product IMC $\mathcal{I}\otimes \mathcal{A}$ is a member of the \emph{Permanent Winning} \emph{Component} $(WC)_{P}$ of $\mathcal{I} \otimes \mathcal{A}$ if $\left< Q_j, s_i \right>$ is a winning component for all product MCs induced by $\mathcal{I} \otimes \mathcal{A}$.
\end{definition}\mbox{}

\noindent We further introduce the notions of permanent accepting BSCC, \textcolor{black}{which is a subset of the permanent winning components of a product IMC}. These sets will prove useful in subsequent sections. \\

\begin{definition}[Permanent Accepting Bottom Strongly Connected Component]
\cite{dutreix2020specification} A state $\left< Q_j, s_i \right> \in Q \times S$ of a product IMC $\mathcal{I}\otimes \mathcal{A}$ is a member of the \emph{Permanent Accepting} \emph{BSCC} $(U^A)_{P}$ of $\mathcal{I} \otimes \mathcal{A}$ if $\left< Q_j, s_i \right>$ belongs to an accepting BSCC for all product MCs induced by $\mathcal{I} \otimes \mathcal{A}$.
\end{definition}\mbox{}

\textcolor{black}{Now that these fundamental results regarding verification in IMCs, which will play a key role in the remainder of this section, have been stated}, recall our \textcolor{black}{primary} objective which is to find switching policies $\widecheck{\mu}^{up}_{\Psi}$ and  $\widehat{\mu}^{low}_{\Psi}$ that respectively minimize the upper bound probability and maximize the lower bound probability of satisfying property $\Psi$ from initial state $Q_j$ in a BMDP $\mathcal{B}$. In light of the above facts, this amounts to enforcing the best possible worst-case scenario with respect to the probability of reaching an accepting BSCC in the product $\mathcal{B} \otimes \mathcal{A}$ for the maximization case, \textcolor{black}{or in the product $\mathcal{B} \otimes \overline{\mathcal{A}}$ for the minimization case}. To this end, we first state in the following lemma that there exist sets of \textcolor{black}{memoryless} switching policies of $\mathcal{B} \otimes \mathcal{A}$ resulting in the greatest possible set of permanent winning components in the corresponding induced product IMCs.\\

\begin{lemma}
\label{LemmaGreaPerm}
Let $\mathcal{B}$ be a BMDP and $\Psi$ be an $\omega$-regular property with corresponding DRA $\mathcal{A}$. \textcolor{black}{The set of policies of the product $\mathcal{B} \otimes \mathcal{A}$ is denoted by $\mathcal{U}_{\otimes}^{\mathcal{A}}$ and the set of memoryless policies of the product $\mathcal{B} \otimes \mathcal{A}$ is denoted by $(\mathcal{U}_{\otimes}^{\mathcal{A}})_{mem} \subseteq \mathcal{U}_{\otimes}^{\mathcal{A}}$}. There exists a set of \textcolor{black}{memoryless} switching policies $\mathcal{U}_{(WC)^{G}_{P}} \subseteq (\mathcal{U}_{\otimes}^{\mathcal{A}})_{mem}$ generating the set $(WC)_{P}^{G}$ in $\mathcal{B} \otimes \mathcal{A}$ such that, for all $\mu \in \mathcal{U}_{\otimes}^{\mathcal{A}}$, $(WC)_{P} \subseteq (WC)_{P}^{G}$ where $(WC)_{P}$ is the permanent winning component of $(\mathcal{B} \otimes \mathcal{A})[\mu]$, and, for all $\mu \in \mathcal{U}_{(WC)^{G}_{P}}$, the permanent winning component of $(\mathcal{B} \otimes \mathcal{A})[\mu]$ is $(WC)^{G}_{P}$.
\end{lemma}\mbox{}

\noindent A constructive proof of this lemma is provided in the Appendix. \textcolor{black}{The set $(WC)_{P}^{G}$ is called the \textit{Greatest Permanent Winning Component} of the product BMDP $\mathcal{B} \otimes \mathcal{A}$.}

From Lemma \ref{LemmaGreaPerm}, we infer that a maximizing policy with respect to $\Psi$ in BMDP $\mathcal{B}$ is induced by a policy $(\widehat{\mu}_{\Psi}^{low})_{\otimes}$ in the product BMDP $\mathcal{B} \otimes \mathcal{A}$ that effectively generates the set $(WC)_{P}^{G}$ and, for all states not in $(WC)_{P}^{G}$, maximizes the lower bound probability of reaching this set; \textcolor{black}{on the other hand, a minimizing policy with respect to $\Psi$ in $\mathcal{B}$ is induced by a policy $(\widecheck{\mu}_{\Psi}^{up})_{\otimes}$ achieving the same thing in $\mathcal{B} \otimes \overline{\mathcal{A}}$, with $\overline{\mathcal{A}}$ denoting a DRA for the complement property of $\Psi$.}  Because optimal switching policies for reachability objectives are memoryless \textcolor{black}{in BMDPs \cite{haddad2018interval}}, it follows that the policy $(\widehat{\mu}_{\Psi}^{low})_{\otimes}$ maximizing the lower bound probability of reaching an accepting BSCC in $\mathcal{B} \otimes \mathcal{A}$ is the same for all initial states of $\mathcal{B} \otimes \mathcal{A}$. \textcolor{black}{Likewise, the policy $(\widecheck{\mu}_{\Psi}^{up})_{\otimes}$ maximizing the lower bound probability of reaching an accepting BSCC in $\mathcal{B} \otimes \overline{\mathcal{A}}$ is the same for all initial states of $\mathcal{B} \otimes \overline{\mathcal{A}}$.}\\

\begin{theorem}
\label{TheoWorstCasePol}
Let $\mathcal{B}$ be a BMDP, $\Psi$ be an $\omega$-regular property with corresponding DRA $\mathcal{A}$, \textcolor{black}{and $\overline{\Psi}$ be the complement of $\Psi$ with corresponding DRA $\overline{\mathcal{A}}$}. Let $(WC)_{P}^{G}$ and \textcolor{black}{$(\overline{WC})_{P}^{G}$} be the greatest permanent winning component, respectively, of the product BMDP $\mathcal{B} \otimes \mathcal{A}$ and \textcolor{black}{$\mathcal{B} \otimes \overline{\mathcal{A}}$}, and $\mathcal{U}_{(WC)_{P}^{G}}$ and \textcolor{black}{$\mathcal{U}_{(\overline{WC})_{P}^{G}}$} be the \textcolor{black}{memoryless} policies generating these sets \textcolor{black}{in the corresponding product BMDP} as defined in Lemma \ref{LemmaGreaPerm}. A lower bound maximizing and upper bound minimizing switching policy $\widehat{\mu}^{low}_{\Psi}$ and $\widecheck{\mu}^{up}_{\Psi}$ in $\mathcal{B}$ with respect to $\Psi$ are respectively induced by switching policies $(\widehat{\mu}^{low}_{\Psi})_{\otimes}$ in $\mathcal{B} \otimes \mathcal{A}$ and $(\widecheck{\mu}^{up}_{\Psi})_{\otimes}$  \textcolor{black}{in $\mathcal{B} \otimes \overline{\mathcal{A}}$} such that
\begin{align}
\label{eqth1} (\widehat{\mu}^{low}_{\Psi})_{\otimes} = & \argmax_{\mu \in \mathcal{U}_{(WC)_{P}^{G}}} \widecheck{\mathcal{P}}_{(\mathcal{B} \otimes \mathcal{A})[\mu]} \big(\left<Q_j, s_0 \right> \models \Diamond (WC)_{P}^{G} \big)\\
\label{eqth2} (\widecheck{\mu}^{up}_{\Psi})_{\otimes} = & \argmax_{\mu \in \mathcal{U}_{\textcolor{black}{(\overline{WC})_{P}^{G}}}} \widecheck{\mathcal{P}}_{(\mathcal{B} \otimes \textcolor{black}{\overline{\mathcal{A}}})[\mu]} \big(\left<Q_j, s_0 \right> \models \Diamond \textcolor{black}{(\overline{WC})_{P}^{G}} \big)\;
\end{align}
for all initial states $Q_j$ of $\mathcal{B}$.
\end{theorem}
\begin{proof}
We first prove equation \eqref{eqth1}. For all states belonging to $(WC)_{P}^{G}$, the lower bound probability of reaching an accepting BSCC under the defined policy $(\widehat{\mu}^{low}_{\Psi})_{\otimes}$ is equal to 1, since $(\widehat{\mu}^{low}_{\Psi})_{\otimes} \in \mathcal{U}_{(WC)^{G}_{P}}$, and is therefore maximized. 

Next, in \cite[Theorem 1]{dutreix2020specification}, it was shown that a lower bound on the probability of reaching an accepting BSCC in a product IMC $\mathcal{I} \otimes \mathcal{A}$ is achieved in an induced product MC $(\mathcal{M}_{\otimes}^{\mathcal{A}})$ with the smallest possible set of winning components admissible by $\mathcal{I} \otimes \mathcal{A}$, which is the permanent winning component $(WC)_{P}$ of $\mathcal{I} \otimes \mathcal{A}$, for all states of $\mathcal{I} \otimes \mathcal{A}$. Furthermore, it was shown in \cite[Lemma 9]{dutreix2020specification} that the probability of reaching an accepting BSCC in an induced product MC $(\mathcal{M}_{\otimes}^{\mathcal{A}})$ increases for all states of $(\mathcal{M}_{\otimes}^{\mathcal{A}})$ as more states are added to the set of winning components of $(\mathcal{M}_{\otimes}^{\mathcal{A}})$ while keeping all other transition probabilities identical. \textcolor{black}{Assume the optimal policy $\mu^{*}$ does not belong to $\mathcal{U}_{(WC)^{G}_{P}}$ for some initial state $\left<Q_{j}, s_{0} \right>$ of $\mathcal{B} \otimes \mathcal{A}$ and denote by $(WC)^{*}_{P}$ the permanent winning component of $\mathcal{B} \otimes \mathcal{A}[\mu^{*}]$. As per the facts above, it follows that the probability of reaching an accepting BSCC from $\left<Q_{j}, s_{0} \right>$ in the worst-case MC of $\mathcal{B} \otimes \mathcal{A}[\mu^{*}]$ has to be less than the probability of reaching an accepting BSCC from $\left<Q_{j}, s_{0} \right>$ in the worst-case MC of $\mathcal{B} \otimes \mathcal{A}[(\mu^{*})']$, where $(\mu^{*})' \in \mathcal{U}_{(WC)^{G}_{P}}$ allows the states in $(WC)^{G}_{P} \setminus (WC)^{*}_{P}$ to be members of the permanent winning component and is the same as $\mu^{*}$ for all states outside of $(WC)^{G}_{P}$, which is a contradiction.} Therefore, for all states of $\mathcal{B} \otimes \mathcal{A}$ which are not in $(WC)^{G}_{P}$, a policy $\mu$ maximizing the lower bound probability of reaching a winning component has to belong to the set $\mathcal{U}_{(WC)^{G}_{P}}$ and generates the largest possible permanent winning component in $(\mathcal{B} \otimes \mathcal{A})[\mu]$. 

Due to the properties of reachability problems in BMDPs, whose optimal policies are memoryless \cite{haddad2018interval}, there exists a policy in $\mathcal{U}_{(WC)^{G}_{P}}$ maximizing the lower bound probability of reaching $(WC)^{G}_{P}$ simultaneously for all states which are not in $(WC)_{P}^{G}$, and, in particular, for all initial states $\left<Q_{j}, s_{0} \right>$ of $\mathcal{B} \otimes \mathcal{A}$ that do not belong to $(WC)_{P}^{G}$, concluding the proof of \eqref{eqth1}.  Symmetric arguments \textcolor{black}{combined with Fact \ref{FactComp} prove \eqref{eqth2}}.
\end{proof}\mbox{}
This theorem shows that the desired policies are computed by solving a lower bound reachability maximization problem on a fixed set of states, which can be accomplished using the value iteration scheme presented in \cite{lahijanian2015formal}. An algorithm for finding the sets $(WC)_{P}^{G}$ and $\textcolor{black}{(\overline{WC})_{P}^{G}}$ as well as their associated control actions are presented in the next subsection.

In this work, we also consider the policies $(\widehat{\mu}_{\Psi}^{up})_{\otimes}$ and $(\widecheck{\mu}_{\Psi}^{low})_{\otimes}$ that respectively maximize the upper bound and minimize the lower bound probability of reaching a winning component for all states in a product BMDP $\mathcal{B} \otimes \mathcal{A}$. While these policies are not mapped onto the original system states, they will prove useful for assessing the \textcolor{black}{quality} of $\widehat{\mu}^{low}_{\Psi}$ and $\widecheck{\mu}^{up}_{\Psi}$ in further sections. These are found by solving an upper bound reachability maximization problem on the \textit{Greatest Winning Component} $(WC)^{G}_L$ in $\mathcal{B} \otimes \mathcal{A}$ \textcolor{black}{(or $\mathcal{B} \otimes \overline{\mathcal{A}}$)}, whose existence is established in the lemma below.\\

\begin{lemma}
\label{LemmaGrea}
Let $\mathcal{B}$ be a BMDP and $\Psi$ be an $\omega$-regular property with corresponding DRA $\mathcal{A}$. \textcolor{black}{The set of policies of the product $\mathcal{B} \otimes \mathcal{A}$ is denoted by $\mathcal{U}_{\otimes}^{\mathcal{A}}$ and the set of memoryless policies of the product $\mathcal{B} \otimes \mathcal{A}$ is denoted by $(\mathcal{U}_{\otimes}^{\mathcal{A}})_{mem} \subseteq \mathcal{U}_{\otimes}^{\mathcal{A}}$}. There exists a set of \textcolor{black}{memoryless} switching policies $\mathcal{U}_{(WC)^{G}_{L}} \subseteq (\mathcal{U}_{\otimes}^{\mathcal{A}})_{mem}$ generating the set $(WC)_{L}^{G}$ in $\mathcal{B} \otimes \mathcal{A}$ such that, for all $\mu \in \mathcal{U}_{\otimes}^{\mathcal{A}}$,  $(WC)_{L} \subseteq (WC)_{L}^{G}$ where $(WC)_{L}$ is the largest winning component of $(\mathcal{B} \otimes \mathcal{A})[\mu]$, and, for all $\mu \in \mathcal{U}_{(WC)^{G}_{L}}$, the largest winning component of $(\mathcal{B} \otimes \mathcal{A})[\mu]$ is $(WC)^{G}_{L}$.
\end{lemma}
\begin{proof}
Lemma \ref{LemmaGrea} follows from a similar constructive argument as the one in the proof of Lemma \ref{LemmaGreaPerm} \textcolor{black}{where the lower bound probability operators are replaced with upper bound probability operators and vice versa}.
\end{proof}\mbox{}

The set $(WC)_{L}^{G}$ \textcolor{black}{is} called the \textit{Greatest Winning Component} of the product BMDP $\mathcal{B} \otimes \mathcal{A}$.\\

\begin{theorem}
\label{UpBoundMaxTheo}
Let $\mathcal{B}$ be a BMDP,  $\Psi$ be an $\omega$-regular property with corresponding DRA $\mathcal{A}$ \textcolor{black}{and $\overline{\Psi}$ be the complement of $\Psi$ with corresponding DRA $\overline{\mathcal{A}}$}. Let $(WC)_{L}^{G}$ and \textcolor{black}{$(\overline{WC})_{L}^{G}$} be the greatest winning component, respectively, of the product BMDP $\mathcal{B} \otimes \mathcal{A}$ \textcolor{black}{and $\mathcal{B} \otimes \overline{\mathcal{A}}$}, and $\mathcal{U}_{(WC)_{L}^{G}}$ and $\mathcal{U}_{\textcolor{black}{(\overline{WC})_{L}^{G}}}$ be the \textcolor{black}{memoryless} policies generating these sets \textcolor{black}{in the corresponding BMDP} as defined in Lemma \ref{LemmaGrea}. An upper bound maximizing and lower bound minimizing switching policy $\widehat{\mu}^{up}_{\Psi}$ and $\widecheck{\mu}^{low}_{\Psi}$ in $\mathcal{B}$ with respect to $\Psi$ are respectively induced by switching policies $(\widehat{\mu}^{up}_{\Psi})_{\otimes}$ \textcolor{black}{in $\mathcal{B} \otimes \mathcal{A}$} and $(\widecheck{\mu}^{low}_{\Psi})_{\otimes}$ in \textcolor{black}{$\mathcal{B} \otimes \overline{\mathcal{A}}$} such that
\begin{align}
\label{eqth21} (\widehat{\mu}^{up}_{\Psi})_{\otimes} = & \argmax_{\mu \in \mathcal{U}_{(WC)_{L}^{G}}} \widehat{\mathcal{P}}_{(\mathcal{B} \otimes \mathcal{A})[\mu] }\big(\left<Q_j, s_0 \right> \models \Diamond (WC)_{L}^{G} \big)\\
\label{eqth22}(\widecheck{\mu}^{low}_{\Psi})_{\otimes} = & \argmax_{\mu \in \mathcal{U}_{\textcolor{black}{(\overline{WC})_{L}^{G}}}} \widehat{\mathcal{P}}_{(\mathcal{B} \otimes \textcolor{black}{\overline{\mathcal{A}}})[\mu]} \big(\left<Q_j, s_0 \right> \models \Diamond \textcolor{black}{(\overline{WC})_{L}^{G}} \big)\;
\end{align}
for all initial states $Q_j$ of $\mathcal{B}$.
\end{theorem}
\begin{proof}
As shown in \cite[Theorem 1]{dutreix2020specification}, an upper bound on the probability of reaching an accepting BSCC in a product IMC $\mathcal{I} \otimes \mathcal{A}$ is achieved in an induced product MC $(\mathcal{M}_{\otimes}^{\mathcal{A}})$ with the largest possible set of winning components allowed by $\mathcal{I} \otimes \mathcal{A}$, which is the largest winning component $(WC)_{L}$ of $\mathcal{I} \otimes \mathcal{A}$, for all initial states of $\mathcal{I} \otimes \mathcal{A}$. Hence, the same arguments as in the proof of Theorem \ref{TheoWorstCasePol} proves \eqref{eqth21}. \textcolor{black}{Symmetric arguments combined with Fact \ref{FactComp} prove \eqref{eqth22}}.
\end{proof}\mbox{}

\noindent We remark that replacing $(WC)_{L}^{G}$ in \eqref{eqth21} by \textit{the greatest accepting BSCC} $(U)^{G}_{L} \subseteq (WC)_{L}^{G}$ \textcolor{black}{of $\mathcal{B} \otimes \mathcal{A}$} does not change the validity of \eqref{eqth21}. The set $(U)^{G}_{L}$ contains all states which belong to an accepting BSCC for at least one induced product MC under at least one policy in $\mathcal{B} \otimes \mathcal{A}$. The proof of the existence of a set of control policies generating this set is similar to the first part of the proof of Lemma \ref{LemmaGreaPerm}. This substitution can be done because, by definition, $\widehat{\mathcal{P}}_{(\mathcal{B} \otimes \mathcal{A})[(\widehat{\mu}^{up}_{\Psi})_{\otimes}]} \Big( (WC)^{G}_L \models \Diamond (U)_{L}^{G} \Big) = 1$, and leads to algorithmic simplifications as the full set $(WC)_{L}^{G}$ may not need to be computed explicitly. \textcolor{black}{A similar reasoning holds by replacing $(\overline{WC})_{L}^{G}$ with the greatest accepting BSCC of $\mathcal{B} \otimes \overline{\mathcal{A}}$} in \eqref{eqth22}. The components $(WC)_{L}^{G}$ and \textcolor{black}{$(\overline{WC})_{L}^{G}$} as well as the control actions generating these components are found via a graph search, as detailed in the next subsections.

\subsection{WINNING COMPONENTS SEARCH ALGORITHMS}
\label{SecCompSearch}

Now, we present graph-based algorithms for finding the greatest permanent winning component $(WC)_{P}^{G}$ of a product BMDP $\mathcal{B} \otimes \mathcal{A}$ defined in Lemma \ref{LemmaGreaPerm}. Furthermore, we show how to design a switching policy that effectively generates this greatest permanent component.

\textcolor{black}{The search is decomposed in two parts: first, we determine a superset of the \textit{greatest permanent accepting BSCC}, denoted by $(U)^{G}_P$, of $\mathcal{B} \otimes \mathcal{A}$ following Algorithm \ref{AlgPerBSCCW}. The set $(U)^{G}_P$ contains all states which belong to a permanent accepting for some control policy in $\mathcal{B} \otimes \mathcal{A}$, and all such states are a part of $(WC)_{P}^{G}$ as seen in the proof of Lemma \ref{LemmaGreaPerm}. We call the superset of $(U)^{G}_P$ returned by this algorithm an extended greatest permanent accepting BSCC, denoted by $(U_{+})^{G}_P$. This set additionally satisfies $(U)^{G}_P \subseteq (U_{+})^{G}_P \subseteq (WC)_{P}^{G}$. Although Algorithm \ref{AlgPerBSCCW} is driven by a search of the sets $(U)^{G}_P$, our implementation allows us to find additional members of $(WC)_{P}^{G}$ in some instances.}

\textcolor{black}{Then, by using an iterative technique which alternates between a graph search and a reachability maximization step in Algorithm \ref{AlgPerWin}, one can find the set of states which are not members of $(U_{+})^{G}_P$ but for which the lower bound probability of reaching an accepting BSCC is equal to 1 nonetheless for some control policy, and effectively create $(WC)_{P}^{G}$}. 

\subsubsection{GREATEST PERMANENT BSCC SEARCH ALGORITHMS}

We now detail an algorithm for finding an extended greatest permanent accepting BSCC \textcolor{black}{$(U_{+})^{G}_P$} of a product BMDP $\mathcal{B} \otimes \mathcal{A}$.

We introduce the following notations and terminology: a set of states in a product $\mathcal{B} \otimes \mathcal{A}$ is said to be accepting if it satisfies the acceptance condition in Definition \ref{accepcon} and is said to be non-accepting otherwise. A state $\left <Q_{\ell}, s_{j} \right>$ of $\mathcal{B} \otimes \mathcal{A}$ with labeling function $L'$ is said to be \textit{Rabin accepting with respect to the $i^{th}$ Rabin pair} of $\mathcal{A}$ if $F_{i} \in L'(\left<Q_{\ell}, s_{j} \right>)$; $\left <Q_{\ell}, s_{j} \right>$ is said to be \textit{Rabin non-accepting with respect to the $i^{th}$ Rabin pair} of $\mathcal{A}$ if $E_{i} \in L'(\left<Q_{\ell}, s_{j} \right>)$. A Rabin accepting state with respect to the $i^{th}$ pair is said to be \textit{unmatched} in a set of states $C$ if, for all $\left <Q_{\ell}, s_{j} \right> \in C$, $E_{i} \not \in L'(\left<Q_{\ell}, s_{j} \right>)$, and it is said to be \textit{matched} otherwise. $Act(C)$ is a set containing all sets of actions allowed for each state in a set $C$, that is, if $C = \{q_{0}, q_{1}, \ldots, q_{k}\}$, $q_{i} \in Q \times S$, then,  $Act(C) = \{ A(q_{0}), A(q_{1}), \ldots, A(q_{k}) \}$. $At_{P}(B, C, Act(C))$ is a function which outputs the set of states in $C$ which have a non-zero probability of transition to $B$ for at least one adversary under all actions in $Act(C)$. In addition, this function removes all actions from the sets in $Act(C)$ for which a transition to $B$ is possible under at least one adversary and returns the updated set of allowed actions for each state of $C$.

We provide a short description of the algorithm: Algorithm \ref{AlgPerBSCCW} first finds the largest possible set of Strongly Connected Components (SCC), denoted by $S$, that can be constructed in the product BMDP in line 4 and 5 assuming all actions are available, as the greatest permanent BSCCs are a subset of these by Definition \ref{BSCCdef}. Set $S$ is determined by applying a standard SCC search techniques on the graph $G$ defined in line 4.

Then, the algorithms iteratively remove the actions and states which prevent these SCCs from being a permanent BSCC, that is, actions and states which allow for a transition outside of the SCCs, as captured by line 9. Note that a state is discarded in set $C_{i}$ once its action set is empty. Then, new SCCs are computed with the remaining states and actions in line 12. If the algorithm finds an SCC $S_{k}$ which does not allow any transition outside of $S_{k}$ for any state and action available, then it is potentially a member of \textcolor{black}{$(U_{+})^{G}_P$}  (line 13). 

Next, the acceptance status of SCC $S_{k}$ is checked at line 14.  This is done by inspecting the states belonging to the SCC and comparing them with Definition \ref{accepcon}. If $S_{k}$ \textcolor{black}{is not accepting}, states which can revert the acceptance status of $S_{k}$ are removed and new SCCs are computed with the remaining states in line 23. Otherwise, the algorithm enters the if-statement in line 14 for a further analysis of $S_{k}$.

An additional condition for $S_{k}$ to be a part of \textcolor{black}{$(U_{+})^{G}_P$} is that no subset of states of $S_{k}$ can form a non-accepting BSCC under any scenario allowed by the transition intervals of the product BSCC.  Too make sure that no subset of $S_{k}$ can form a non-accepting BSCC, we choose control actions for the states in $S_{k}$ that maximize the lower bound probability of reaching the unmatched Rabin accepting states contained in $S_{k}$ in line 14 to 17. If this lower bound is zero for some subset of $S_{k}$, then these states could potentially form a non-accepting BSCC inside $S_{k}$ for some assignment of the probabilities under all available actions. The set of all such states is denoted by $A_{bad}$. If $A_{bad}$ is empty, the algorithm found a control policy that guarantees $S_{k}$ to be accepting for all possible adversaries of the induced product IMC, since no state of $S_{k}$ can form a BSCC which doesn't contain at least one of the unmatched accepting states, and $S_{k}$ is added to \textcolor{black}{$(U_{+})^{G}_P$} in line 18.  Otherwise, the SCCs which can be formed by the states in $A_{bad}$ and by the states in $S_{k} \setminus A_{bad}$ with the remaining actions are computed and added to $S$ in line 20.

\begin{algorithm}[H]
\caption{Find Extended Greatest Permanent Accepting BSCC} 
\algsetup{linenosize=\scriptsize}
\scriptsize
\begin{algorithmic}[1]
\STATE \textbf{Input}: Product BMDP $\mathcal{B} \otimes \mathcal{A}$
\STATE \textbf{Output}: Extended greatest permanent accepting BSCCs $\textcolor{black}{(U_{+})^{G}_P}$ with corresponding policy $(\widehat{\mu}_{\Psi}^{low})_{\otimes}$ for the states in this set

\STATE \textbf{Initialize}: $(U_{+})^{G}_P := \emptyset$
\STATE Initially allow all actions for all states. Construct $G := (V,E)$ with a vertex for each state in $\mathcal{B} \otimes \mathcal{A}$ $(V = Q \times S)$ and an edge between states $\left< Q_{i} , s_{j} \right>$ and $\langle Q_{i'} , s_{j'} \rangle$ if $\widehat{T}(\left< Q_{i} , s_{j} \right>, a, \langle Q_{i'} , s_{j'} \rangle) > 0$ for some $a \in A(\left< Q_{i} , s_{j} \right>)$
\STATE Find all SCCs of $G$ and list them in $S$
\FOR {$S_k \in S$}
\STATE $C_0 := \emptyset$,  $i := 0$
\REPEAT
\STATE $R_i := S_k \setminus \cup_{\ell = 0}^{i} C_{\ell}$; \; $Tr_i := V \setminus R_i$; \; $(C_{i+1}, Act(R_i)) = At_P(Tr_i, R_i, Act(R_i))$; \; $i = i +1$
\UNTIL $C_i = \emptyset$ and no action is removed from $Act(R_i)$
\IF {$i \not = 1$}
\STATE Find all SCCs of $R_i$ (with the remaining actions) and add them to $S$
\ELSE
\IF{$S_k$ is accepting}
\STATE Find the set $A$ of all unmatched Rabin accepting states of  $S_k$
\STATE For all states in $S_{k}$, maximize the lower bound probability of $\Diamond A$. Find the set of states $A_{bad}$ whose lower bound probability of reaching $A$ is zero after the maximization step
\IF{$A_{bad} = \emptyset$}
\STATE $(U_{+})^{G}_P := (U_{+})^{G}_P \cup S_k$ and save the actions computed in the maximization of $\Diamond A$ to $(\widehat{\mu}_{\Psi}^{low})_{\otimes}$ for all states of $S_{k}$
\ELSE
\STATE Compute the SCCs formed by the states in $A_{bad}$ and the states in $S_{k} \setminus A_{bad}$ with the remaining actions and add them to $S$
\ENDIF
\ELSE
\STATE If $S_{k}$ does not contain any Rabin accepting state, continue. Otherwise, for all Rabin accepting set of states $A_{i}$ with respect to pair $i$ in $S_{k}$, find the set $A_{i}^{non}$ of all states in $S_{k}$ which are non-accepting with respect to the same pair as $A_{i}$. Compute the SCCs formed by the states in $S_{k} \setminus A_{i}^{non}$ with the remaining actions and add them to $S$
\ENDIF
\ENDIF
\ENDFOR
\RETURN $(U_{+})^{G}_P$ , $(\widehat{\mu}_{\Psi}^{low})_{\otimes}$ for states in $(U_{+})^{G}_P$ 
\end{algorithmic}
\label{AlgPerBSCCW}
\end{algorithm}

We offer the following reasoning as a proof sketch for the correctness of the algorithm \textcolor{black}{, i.e, to show that the output $(U_+)^G_P$ of Algorithm \ref{AlgPerBSCCW} satisfies the chains of inequalities $(U)^G_P \subseteq (U_+)^G_P \subseteq (WC)^G_P$}: for a set of states $S_{k}$ to belong to a permanent BSCC of a given kind in a product IMC, \textcolor{black}{the following conditions must hold}: 1) its constituents are not allowed to transition outside of $S_{k}$ under any adversary, \textcolor{black}{2) its constituents have to be reachable from one another under all adversaries}, \textcolor{black}{3)} its constituents have to fulfill the requirements for accepting and non-accepting BSCCs defined in Definition \ref{accepcon}, \textcolor{black}{4)} no subset of $S_{k}$ is allowed to form a BSCC of the opposite acceptance status under any adversary.  \textcolor{black}{Condition 1)} is guaranteed by lines 7 to 10; \textcolor{black}{Condition 2) is not enforced and is the reason for outputting a superset of $(U)^G_P$. This is because, as long as the other 3 conditions are fulfilled, the states in the set $S_{k}$ will still be permanently winning,, although the transition bounds within $S_{k}$ might allow these sets to be winning via different scenarios that are not only a BSCC formed by all the states of $S_{k}$ (e.g. a subset of $S_{k}$ always transitioning to a another subset of $S_{k}$ forming a BSCC)}; \textcolor{black}{Condition 3)} is enforced by the if-statement in line 14  and the corresponding else-statements of lines 22 to 24; \textcolor{black}{Condition 4)} is imposed by the remainder of the main for-loop. Lastly, the algorithm iteratively removes the minimum number of actions and states causing a set $S_{k}$ to violate one of these conditions and analyze all of the remaining states, ensuring that the procedure does not skip any permanent component. Note that none of the removed states could form a permanent BSCC between each other under any policy. Indeed, if these states did not belong to a common SCC in $S$, this would be a contradiction. \textcolor{black}{Therefore, by virtue of this fact, Algorithm \ref{AlgPerBSCCW} does not ``miss" any permanent BSCCs and it must hold that $(U)^G_P \subseteq (U_+)^G_P$. Moreover, the previous discussion regarding Condition 2 ensures that all states in $(U_+)^G_P \setminus (U)^G_P$ are still permanently winning, guaranteeing that $(U_+)^G_P \subseteq (WC)^G_P$ and concluding the proof sketch.}

\textcolor{black}{This algorithm} can be adapted to determine an extended greatest accepting $(U_{+})^{G}_{L}$ by replacing all instances of the function $At_{P}(B, C, Act(C))$ with the function $At_{\textcolor{black}{pot}}(B, C,$ $ Act(C))$, where $At_{\textcolor{black}{pot}}(B, C, Act(C))$ returns the set of states of $C$ which have a non-zero probability of transition to $B$ for all adversaries under all allowed actions. This function also removes all actions from $Act(C)$ for which a non-zero probability of transition to $B$ exists under all adversaries of the induced IMC and returns the updated set of allowed actions. In addition, all mentions of the term ``lower bound" have to be replaced with ``upper bound". The extended set \textcolor{black}{is} such that $(U)^{G}_L \subseteq (U_{+})^{G}_L \subseteq (WC)_{L}^{G}$.

\subsubsection{GREATEST PERMANENT COMPONENTS SEARCH ALGORITHMS}

Next, we present an algorithm which constructs the greatest permanent winning components $(WC)^{G}_P$ in a product BMDP $\mathcal{B} \otimes \mathcal{A}$ once \textcolor{black}{an extended greatest permanent BSCC $(U_{+})_{P}^{G}$ has been found.}

In a product IMC $\mathcal{I} \otimes \mathcal{A}$, some states which are not in a permanent BSCC can still be a part of the permanent winning component of $\mathcal{I} \otimes \mathcal{A}$, as discussed in the second part of the proof of Lemma \ref{LemmaGreaPerm}. These states are those which belong to a set of states $C$ such that no transition outside the union of $C$ and the permanent BSCCs of $\mathcal{I} \otimes \mathcal{A}$ is possible for any adversary, and such that no subset of $C$ can form a \textcolor{black}{non-accepting} BSCC status under any adversary. We can further classify these states into \textit{permanent sink} states, which cannot be a part of a BSCC under any scenario but transition to another winning set of state with lower bound probability 1, and states which allow non-deterministic scenarios where the state is sometimes a sink state with respect to another permanent winning set of states and sometimes a part of a winning component that reaches a non permanent accepting BSCC with probability one. The examples below, \textcolor{black}{illustrated in Figure \ref{FigExample1}}, present situations where these scenarios can occur.

\begin{figure}
\hspace{0.3cm}
\includegraphics[scale=0.8]{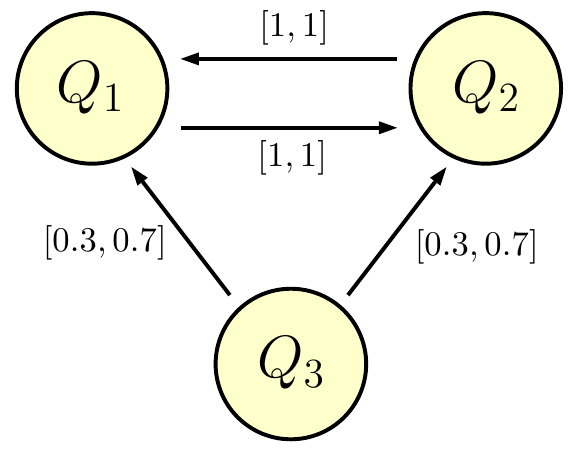}
\hspace{2cm}
\includegraphics[scale=0.8]{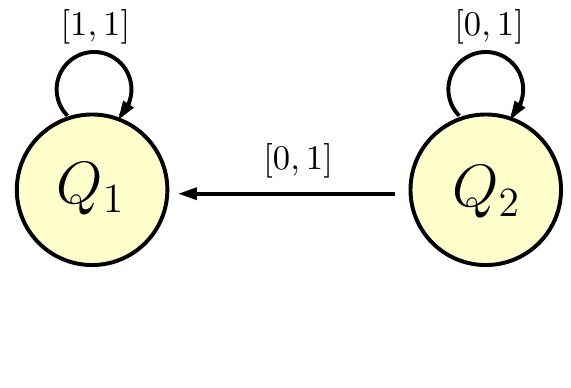}
\caption{\textcolor{black}{Depiction of the product IMCs in Example 1. On the left, state $Q_{3}$ transitions to the BSCC formed by $Q_{1}$ and $Q_{2}$ under all possible adversaries and is therefore a permanent sink state. On the right, state $Q_{2}$ is either a sink state with respect to state $Q_{1}$ or a BSCC itself for all realizations of the probability intervals.}}
\label{FigExample1}
\end{figure}

\begin{example}
Consider three states $Q_{1}$, $Q_{2}$ and $Q_{3}$ of a product IMC such that $Q_{1}$ and $Q_{2}$ form a permanent BSCC, with $\widecheck{T}(Q_{1}, Q_{2}) = \widecheck{T}(Q_{2}, Q_{1}) = 1$. Furthermore, $\widecheck{T}(Q_{3}, Q_{1}) = \widecheck{T}(Q_{3}, Q_{2}) = \textcolor{black}{0.3}$ \textcolor{black}{and $\widehat{T}(Q_{3}, Q_{1}) = \widehat{T}(Q_{3}, Q_{2}) = 0.7$}. Clearly, $Q_{3}$ is not a member of the BSCC encompassing $Q_{1}$ and $Q_{2}$; yet, $Q_{3}$ always transitions to either $Q_{1}$ or $Q_{2}$ with probability probability 1 and is therefore a permanent sink state.

Now, consider two states $Q_{1}$ and $Q_{2}$ such that $\widecheck{T}(Q_{1}, Q_{1}) = 1$, $\widecheck{T}(Q_{2}, Q_{1}) =  \widecheck{T}(Q_{2}, Q_{2}) = 0$ and $\widehat{T}(Q_{2}, Q_{1}) =  \widehat{T}(Q_{2}, Q_{2}) = 1$. While $Q_{1}$ is a permanent BSCC, $Q_{2}$ is neither a permanent sink state nor a permanent BSCC. However, all adversaries of the product IMC make $Q_{2}$ either a sink state with respect to $Q_{1}$ or a BSCC with itself.
\end{example}

Consequently, we describe a procedure in Algorithm \ref{AlgPerWin} that finds all states in a product $\mathcal{B} \otimes \mathcal{A}$ for which a control policy induces one of the aforementioned scenarios given extended greatest permanent BSCCs $ (U_{+})_{P}^{G}$.

We explain the main features of \textcolor{black}{this algorithm}: first, the greatest permanent \textcolor{black}{winning} component $(WC)^{G}_P$ \textcolor{black}{is} initialized to the extended greatest permanent \textcolor{black}{accepting} BSCCs in line 3. Then, in line 5, the lower bound probability of reaching \textcolor{black}{this component} is maximized in the product BMDP to reveal the states which can be rendered permanent sinks with respect to $(WC)^{G}_P$, as these states yield a lower bound of 1 of reaching the component. The sink states are added to \textcolor{black}{$(WC)^{G}_P$} in line 8.

Next, we define the greatest potential accepting BSCC $(U)_{\textcolor{black}{pot}}^{G}$ of a product BMDP, which are computed by taking the set difference between the greatest \textcolor{black}{winning} BSCC and the greatest permanent \textcolor{black}{winning} BSCC. States in $(U)_{\textcolor{black}{pot}}^{G}$ are those which could engender the second type of permanent components previously discussed. If $(U)_{\textcolor{black}{pot}}^{G}$ happened to contain a permanent sink state found in line 8, we compute the greatest accepting and non-accepting BSCC as well as their associated allowed actions with the remaining states in line 10 to update $(U)_{\textcolor{black}{pot}}^{G}$.

Then, in lines 12 to 17, for all BSCCs $S$ which can be created in $(U)_{\textcolor{black}{pot}}^{G}$, \textcolor{black}{we check whether there exists a policy such that no state of $S$ can transition outside of the union of $S$ and the current version of the greatest permanent winning component} for any instantiation of the resulting transition intervals. If such a policy does not exist, states and actions for which a transition outside of the aforementioned set is possible are removed from $S$ and the BSCCs which can be created inside the greatest BSCC of the remaining states are added to the list $N$ of BSCCs to inspect in line 19. On the other hand, if $S$ only contains valid states and corresponding actions, the algorithm enters the else-statement in line 20,  where we need to choose a policy for the states in $S$ which additionally does not allow the existence of a \textcolor{black}{non-accepting} BSCC within $S$ under any adversary. 

This step is done similarly as in Algorithm \ref{AlgPerBSCCW} by maximizing the lower bound probability of reaching the unmatched Rabin accepting states in $S$ and removing the states yielding a lower bound probability of $0$. If no such state is found, then we designed a policy that effectively makes $S$ either a set of sink states or an \textcolor{black}{accepting} BSCC for all adversaries, and the states of $S$ are added to the greatest permanent winning component $(WC)^{G}_P$. This process is described in line 21 to 28.

In the case that new states were added to $(WC)^{G}_P$ upon execution of the reachability maximization step and the graph search, which is checked in line 31 to 33, we return to the beginning of the while-loop and repeat this process with the augmented version of the \textcolor{black}{greatest permanent winning component}, as it could now allow previously discarded states to become permanently winning. Otherwise, the loop is exited and the algorithms return the true set $(WC)^{G}_P$ with \textcolor{black}{its} associated control actions.

A slight modification of Algorithm \ref{AlgPerWin} can be employed to compute the greatest set $(WC)^{G}_L$ defined in Lemma \ref{LemmaGrea}. However, in this paper, we solely use the greatest \textcolor{black}{accepting} BSCC \textcolor{black}{$(U_{+})_{L}^{G}$} as our target set for computing the upper bound maximizing and lower bound minimizing policies $(\widehat{\mu}_{\Psi}^{up})_{\otimes}$ and $(\widecheck{\mu}_{\Psi}^{low})_{\otimes}$, as explained in Subsection \ref{SecBMDPSynth}.

\begin{algorithm}[H]
\algsetup{linenosize=\scriptsize}
\scriptsize
\caption{Find Greatest Permanent Winning Components} 
\begin{algorithmic}[1]
\STATE \textbf{Input}: Product BMDP $\mathcal{B} \otimes \mathcal{A}$, extended greatest permanent accepting BSCC \textcolor{black}{$(U_{+})_{P}^{G}$}, extended greatest accepting BSCCs \textcolor{black}{$(U_{+})_{L}^{G}$}
\STATE \textbf{Output}: Greatest permanent winning component $(WC)^{G}_P$ with corresponding policy $(\widehat{\mu}_{\Psi}^{low})_{\otimes}$ for the states in this set
\STATE \textbf{Initialize}: $ (WC)^{G}_P :=  (U_{+})^{G}_P$, $(U)^{G}_{\textcolor{black}{pot}} := (U_{+})_{L}^{G} \setminus (U_{+})_{P}^{G}$, $(WC)^{G}_{P, prev} := (WC)^{G}_P$
\REPEAT
\STATE Maximize the lower bound probability of $ \Diamond (WC)^{G}_P$ for all states $\left< Q_{i}, s_{j} \right>$ in $\mathcal{B} \otimes \mathcal{A}$
\STATE Construct the set $L$ of all states with a lower bound equal to 1 that are not in $(WC)^{G}_P$
\FOR{$Q \in L$}
\STATE $(WC)^{G}_P := (WC)^{G}_P \cup Q$, save the action $(\widehat{\mu}_{\Psi}^{low})_{\otimes}(Q)$ computed during maximization step
\ENDFOR
\STATE Find the greatest accepting BSCC of $(U)^{G}_{\textcolor{black}{pot}} \setminus L$ using Algorithm \ref{AlgPerBSCCW} and set $(U)^{G}_{\textcolor{black}{pot}}$ to this new set of states
\STATE Construct the set $N$ of all accepting BSCCs constructed in $(U)^{G}_{\textcolor{black}{pot}}$ under some policy
\FOR{$S_{k} \in N$}
\STATE Construct $G := (V,E)$ with a vertex for each state in $\mathcal{B} \otimes \mathcal{A}$ $(V = Q \times S)$ and an edge between states $\left< Q_{i} , s_{j} \right>$ and $\langle Q_{i'} , s_{j'} \rangle$ if $\widehat{T}(\left< Q_{i} , s_{j} \right>, a, \langle Q_{i'} , s_{j'} \rangle) > 0$ for some $a \in A(\left< Q_{i} , s_{j} \right>)$
\STATE $C_{0} := \emptyset$, $i :=0$
\REPEAT
\STATE $R_i := S_{k} \setminus \cup_{\ell = 0}^{i} C_{\ell}$; \; $Tr_i := V \setminus (R_i \cup (WC)^{G}_P)$; \; $(C_{i+1}, Act(R_i)) := At_P(Tr_i, R_i, Act(R_i))$; \; $i := i +1$
\UNTIL $C_i = \emptyset$ and no action is removed from $Act(R_i)$
\IF {$i \not = 1$}
\STATE Find the greatest accepting BSCC of $R_i$ (with remaining actions) using Algorithm \ref{AlgPerBSCCW}, enumerate all accepting BSCCs constructed in this set under some policy, and add them to $N$
\ELSE 
\STATE Find the set $A$ of all unmatched Rabin accepting states of $S_k$
\STATE For all states in $S_k$, maximize the lower bound probability of $\Diamond A$. Find the set of states $A_{bad}$ whose lower bound probability of reaching $A$ is zero after the maximization step
\IF{$A_{bad} = \emptyset$}
\STATE $(WC)^{G}_P := (WC)^{G}_P \cup S_{k}$, save corresponding actions in $(\widehat{\mu}_{\Psi}^{low})_{\otimes}$ for the states in $S_{k}$\hspace{-0.2em}
\STATE $(U)^{G}_{\textcolor{black}{pot}} := (U)^{G}_{\textcolor{black}{pot}} \setminus S_{k}$
\ELSE 
\STATE Compute the greatest accepting BSCC of $A_{bad}$ and $S_{k} \setminus A_{bad}$ using Algorithm \ref{AlgPerBSCCW}, enumerate all accepting BSCCs constructed in this set under some policy, and add them to $N$
\ENDIF
\ENDIF
\ENDFOR
\STATE $Y := (WC)^{G}_P \setminus (WC)^{G}_{P, prev}$
\STATE $(WC)^{G}_{P, prev} := (WC)^{G}_P$
\UNTIL{$Y = \emptyset$}
\RETURN $(WC)^{G}_P$ , $(\widehat{\mu}_{\Psi}^{low})_{\otimes}$ for states in $(WC)^{G}_P$
\end{algorithmic}
\label{AlgPerWin}
\end{algorithm}

In summary, we develop a procedure for computing policies that either maximize the lower bound probability or minimize the upper bound probability of satisfying an arbitrary $\omega$-regular property in a BMDP. To this end, we show that these policies are induced by policies in the product between the BDMP and \textcolor{black}{a} DRA encoding the specification of interest \textcolor{black}{for the maximization objective}, \textcolor{black}{or a DRA encoding the complement of the specification for the minimization objective}. In Lemma \ref{LemmaGreaPerm}, we remarked that a product BMDP always possesses a greatest permanent winning component. In Algorithms \textcolor{black}{\ref{AlgPerBSCCW} and \ref{AlgPerWin}},  we devise graph-based techniques for determining \textcolor{black}{this component} as well as the corresponding control actions for the states composing them. Finally, we show in Theorem \ref{TheoWorstCasePol} that, for the remaining states in the product BMDPs, the optimal policies \textcolor{black}{are} found by carrying out a lower bound reachability maximization computation on the greatest permanent \textcolor{black}{winning} component.

\subsection{STATE SPACE REFINEMENT}
\label{SecRefFin}

\subsubsection{\textcolor{black}{QUALITY} OF COMPUTED POLICY}
\label{SubSubOpt}
In the previous subsections, we implemented a technique for computing an optimal switching policy in a BMDP subject to an $\omega$-regular specification. However, recall that, in the problem at hand, BMDPs are used as abstractions of the underlying system \eqref{eq1} with respect to a partition of the system's continuous domain. Therefore, as each state of the BMDP abstracts the behavior an infinite number of continuous states of \eqref{eq1}, the switching policy derived in the BMDP abstraction is likely to be suboptimal when mapped onto the original system.

Here, we provide a measure of the suboptimality of the control strategy computed in a BMDP abstraction with respect to the abstracted system. \textcolor{black}{While the discussion in this section focuses on optimality for the probability maximization problem with respect to specification $\Psi$, the same facts can straightforwardly be applied to the dual minimization problem by replacing the instances of $(\widehat{\mu}_{\Psi}^{low})_{\otimes}$, $(\widehat{\mu}^{up}_{\Psi})_{\otimes}$ and $\mathcal{B} \otimes \mathcal{A}$ with $(\widecheck{\mu}_{\Psi}^{up})_{\otimes}$, $(\widecheck{\mu}_{\Psi}^{low})_{\otimes}$ and $\mathcal{B} \otimes \overline{\mathcal{A}}$ respectively, where $\overline{\mathcal{A}}$ is a DRA representing the complement specification $\overline{\Psi}$.} 

\textcolor{black}{The value iteration algorithm used to design the policies $(\widehat{\mu}_{\Psi}^{low})_{\otimes}$ and $(\widehat{\mu}^{up}_{\Psi})_{\otimes}$ discussed in Theorem \ref{TheoWorstCasePol} and Theorem \ref{UpBoundMaxTheo} provides useful information amenable to a quantitative measure of the quality of the lower bound maximizing policy $(\widehat{\mu}_{\Psi}^{low})_{\otimes}$. In particular, for all states $\left< Q_j, s_{i} \right>$, the algorithm determines a lower bound on the maximum lower bound probability of reaching an accepting BSCC achievable from $\left< Q_j, s_{i} \right>$ over all memoryless policies of $\mathcal{B} \otimes \mathcal{A}$ choosing the lower bound maximizing action $a_{\ell,max} =  (\widehat{\mu}_{\Psi}^{low})_{\otimes}(\left< Q_j, s_{i} \right>)$ at state $\left< Q_j, s_{i} \right>$, and an upper bound on the maximum upper bound probability of reaching an accepting BSCC achievable from $\left< Q_j, s_{i} \right>$ over all memoryless policies of $\mathcal{B} \otimes \mathcal{A}$ choosing action $a_{\ell}$ at state $\left< Q_j, s_{i} \right>$ for all actions $a_{\ell} \in A(\left< Q_j, s_{i} \right>)$. Denoting these lower and upper bounds by $\widecheck{p}_{\ell}$ and $\widehat{p}_{\ell}$ respectively for action $a_{\ell}$, and the set of memoryless policies of $\mathcal{B} \otimes \mathcal{A}$ by $(\mathcal{U}_{\otimes}^{\mathcal{A}})_{mem}$, this is formally stated as
\begin{align}
\widecheck{p}_{\ell,max} \leq \hspace{-8mm} \max_{\substack{\mu \in (\mathcal{U}_{\otimes}^{\mathcal{A}})_{mem} \\ s.t. \\ \mu(\left< Q_j, s_{i} \right>) = a_{\ell,max}}} \hspace{-8mm} \widecheck{\mathcal{P}}_{(\mathcal{B} \otimes \mathcal{A})[\mu]}(\left< Q_j, s_{i} \right> \models \Diamond R) \ ,
\end{align}}

\noindent \textcolor{black}{where the subscript $\ell,max$ refers to the lower bound maximizing action and, for all actions $a_{\ell} \in A(\left< Q_j, s_{i} \right>)$,
\begin{align}
\widehat{p}_{\ell} \geq \hspace{-5mm} \max_{\substack{\mu \in (\mathcal{U}_{\otimes}^{\mathcal{A}})_{mem}\\ s.t. \\ \mu(\left< Q_j, s_{i} \right>) = a_{\ell}}} \hspace{-5mm} \widehat{\mathcal{P}}_{(\mathcal{B} \otimes \mathcal{A})[\mu]}(\left< Q_j, s_{i} \right> \models \Diamond R) \ ,
\end{align}}

\noindent \textcolor{black}{where $\Diamond R$ is a slight abuse of notation denoting the objective of reaching an accepting BSCC --- which is generally not a fixed set of states as discussed in previous sections --- in the product IMC $(\mathcal{B} \otimes \mathcal{A})[\mu]$.}

\textcolor{black}{Therefore, we introduce the \textit{suboptimality factor} $\epsilon_{\left< Q_j, s_{i} \right>}$ of state $\left< Q_j, s_{i} \right>$ with respect to the lower bound maximizing policy $(\widehat{\mu}_{\Psi}^{low})_{\otimes}$ in the product BMDP $\mathcal{B} \otimes \mathcal{A}$ which is defined as
\begin{align}
\label{SubFacMax}
 \epsilon_{\left< Q_j, s_{i} \right>}  =  \max_{\ell \not = \ell,max} \widehat{p}_{\ell}   -  \widecheck{\mathcal{P}}_{(\mathcal{B} \otimes \mathcal{A})[(\widehat{\mu}_{\Psi}^{low})_{\otimes}]}\left(\left< Q_j, s_{i} \right> \models \Diamond (WC)_{P}^{G} \right) \ .
\end{align}
\noindent The quantity $\epsilon_{\left< Q_j, s_{i} \right>}$ represents an upper bound on the maximum improvement in the probability of satisfying $\Psi$, using a memoryless policy with respect to the DRA states, any continuous state in $Q_j$ could achieve by choosing another fixed action from the one prescribed by $(\widehat{\mu}_{\Psi}^{low})_{\otimes}$ when the product state is $\left< Q_j, s_{i} \right>$, as the maximum satisfaction probability attainable when applying a different action is upper bounded by $ \max_{\ell \not = \ell,max} \widehat{p}_{\ell}$. Therefore, the smaller $\epsilon_{\left< Q_j, s_{i} \right>}$ is, the more certain we are that $(\widehat{\mu}_{\Psi}^{low})_{\otimes}$ is close to the best memoryless (in the product) policy for all states in $Q_j$ when the automaton state is $s_{i}$.}

\textcolor{black}{Furthermore, the bounds computed by the value iteration algorithm can additionally be used to show that certain actions are \textit{suboptimal} or \textit{optimal} at a given state of a product BMDP $\mathcal{B} \otimes \mathcal{A}$ and, by extension, that the modes represented by these actions are suboptimal or optimal for some continuous states of the abstracted system for policies that are memoryless in the product. By comparing these bounds for all actions in an action space of a given state of the product BMDP $\mathcal{B} \otimes \mathcal{A}$, some of these actions may appear to surely perform worse or better than others at that particular state, as illustrated in the example below.}

\begin{example}
Consider a state $\left< Q_{j}, s_{i} \right>$ of the product BMDP $\mathcal{B} \otimes \mathcal{A}$ with a set of actions $A(\left< Q_{j}, s_{i} \right>) = \{a_{1}, a_{2}, a_{3} \}$, and $(\widehat{\mu}_{\Psi}^{low})_{\otimes}(\left< Q_{j}, s_{i} \right>) = a_{1}$.  Suppose the probabilities of reaching an accepting BSCC from $\left< Q_{j}, s_{i} \right>$ under all 3 actions are described by the following intervals:
\begin{itemize}
\item $(I_{\left< Q_{j}, s_{i} \right>})_{a_{1}} = [0.5, 0.8], $
\item $(I_{\left< Q_{j}, s_{i} \right>})_{a_{2}} = [0.0, 0.7], $
\item $(I_{\left< Q_{j}, s_{i} \right>})_{a_{3}} = [0.0, 0.45], $
\end{itemize}
\noindent where the lower bounds correspond to a lower bound on the maximum lower bound probability of reaching an accepting BSCC from state $\left< Q_{j}, s_{i} \right>$ achievable over all memoryless policies of $\mathcal{B} \otimes \mathcal{A}$ choosing the corresponding action at state $\left< Q_{j}, s_{i} \right>$, and the upper bounds correspond to an upper bound on the maximum upper bound probability of reaching an accepting BSCC from state $\left< Q_{j}, s_{i} \right>$ achievable over all memoryless policies of $\mathcal{B} \otimes \mathcal{A}$ choosing the corresponding action at state $\left< Q_{j}, s_{i} \right>$.

Although action $a_{1}$ maximizes the lower bound probability of reaching an accepting BSCC at $0.5$, it appears that some continuous states of $Q_{j}$ could potentially produce a higher probability --- up to $0.7$ --- of reaching an accepting BSCC under action $a_{2}$, since a non-deterministic scenario of the product BMDP allows for this probability to occur under some policy choosing $a_{2}$. However, under no memoryless policy and adversary can action $a_{3}$ generate a higher probability of reaching an accepting BSCC than action $a_{1}$, since $0.45 < 0.5$, and can therefore be discarded. 

\textcolor{black}{In spite of action $a_{3}$ being removed}, the suboptimality factor of $\left< Q_{j}, s_{i} \right>$ with respect to $(\widehat{\mu}_{\Psi}^{low})_{\otimes}$ in this case is $ \epsilon_{\left< Q_j, s_{i} \right>}  = 0.7 - 0.5 = 0.2$, \textcolor{black}{as there still exists an action achieving a higher upper bound probability of reaching an accepting BSCC, namely $a_{2}$ with $0.7$, than the lower bound probability of reaching an accepting BSCC under the lower bound maximizing action, namely $a_{1}$ with $0.5$.}
\end{example}\mbox{}

\begin{definition}[Optimal/Suboptimal Action]
\label{OptActDef}
Consider a state $\left< Q_{j}, s_{i} \right>$ of a product BMDP $\mathcal{B} \otimes \mathcal{A}$ with a set of actions $A(\left< Q_{j}, s_{i} \right>)$. Let us denote by $\widecheck{p}_{\ell}$ a lower bound on the maximum lower bound probability of reaching an accepting BSCC from $\left< Q_{j}, s_{i} \right>$ achievable over all memoryless policies of $\mathcal{B} \otimes \mathcal{A}$ choosing action $a_{\ell} \in A(\left< Q_{j}, s_{i} \right>)$ at state $\left< Q_{j}, s_{i} \right>$, and by $\widehat{p}_{\ell}$ an upper bound on the maximum upper bound probability of reaching an accepting BSCC from $\left< Q_{j}, s_{i} \right>$ achievable over all memoryless policies of $\mathcal{B} \otimes \mathcal{A}$ choosing action $a_{\ell} $ at state $\left< Q_{j}, s_{i} \right>$. An action $a_{\ell}$ is said to be \emph{suboptimal} for state $\left< Q_{j}, s_{i} \right>$ with respect to $A(\left< Q_{j}, s_{i} \right>)$ if there exists an action $a_{k} \in A(\left< Q_{j}, s_{i} \right>)$, $k \not = \ell$, such that $\widehat{p}_{\ell} < \widecheck{p}_{k}$. An action $a_{\ell}$ is said to be \emph{optimal} for state $\left< Q_{j}, s_{i} \right>$ with respect to $A(\left< Q_{j}, s_{i} \right>)$ if, for all $a_{k} \in A(\left< Q_{j}, s_{i} \right>)$, $k \not = \ell$,  $\widecheck{p}_{\ell} \geq \widehat{p}_{k}$.
\end{definition}\mbox{}

\begin{definition}[Optimal/Suboptimal Mode]
Let $\pi = x[0]x[1]x[2] \ ... \ x[k]$ be any finite path of \eqref{eq1} such that the word $L(x[0])L(x[1])L(x[2]) \ ... \ L(x[k])$ produces a run $s[0]s[1]s[2] \ldots s[k]$ in automaton $\mathcal{A}$ corresponding to property $\Psi$, where $x[k] =: x \in D$ and $s[k] = s_{i} \in S$. Let us denote by $\widecheck{p}_{\ell}$ a lower bound on the maximum (respectively, minimum) probability of an infinite path with prefix $\pi$ to satisfy $\Psi$ in \eqref{eq1} over all policies of \eqref{eq1} choosing mode $a_{\ell} \in A$ for path $\pi$, and by $\widehat{p}_{\ell}$ an upper bound on the maximum (respectively, minimum) probability of an infinite path with prefix $\pi$ to satisfy $\Psi$ in \eqref{eq1} over all policies of \eqref{eq1} choosing mode $a_{\ell} \in A$ for path $\pi$. When the objective is to maximize (respectively, minimize) the probability of satisfying $\Psi$, a mode $a_{\ell}$ is said to be \emph{suboptimal} for state $x$ with respect to automaton state $s_{i}$ and the set of modes $A$ if there exists a mode $a_{k} \in A$, $k \not = \ell$, such that $\widehat{p}_{\ell} < \widecheck{p}_{k}$ (respectively, $\widecheck{p}_{\ell} > \widehat{p}_{k}$). A mode $a_{\ell}$ is said to be \emph{optimal} for state $x$ with respect to automaton state $s_{i}$ and the set of modes $A$ if, for all $a_{k} \in A$, $k \not = \ell$,  $\widehat{p}_{k} \leq \widecheck{p}_{\ell}$ (respectively, $\widehat{p}_{\ell} \leq \widecheck{p}_{k}$).
\end{definition}\mbox{}

\noindent If the set of actions $A(\left< Q_{j}, s_{i} \right>)$ of state $\left< Q_{j}, s_{i} \right>$ contains an optimal action, then the suboptimality factor $\epsilon_{\left< Q_{j}, s_{i} \right>}$ is set to 0.\\

\subsubsection{REFINEMENT PROCEDURE}
\label{SubSubRef}

Now that a quantitative measure for the \textcolor{black}{quality} of the computed switching policy has been introduced, our next objective is to design a domain partition refinement scheme to address Subproblem 1.2 and achieve a user-defined level of optimality. In order to mitigate the state-space explosion phenomenon, the refinement algorithm should specifically target the states causing the most uncertainty in the domain partition. 

We define the \textit{greatest suboptimality factor} $\epsilon_{max}$ as
\begin{align}
\epsilon_{max} = \max_{\left<Q_j , s_{i} \right> \in (Q \times S)} \; \epsilon_{\left<Q_j , s_{i} \right>}
\label{maxsub}
\end{align}
\noindent  which can be used as a natural precision criterion for a given domain partition $P$. A low factor $\epsilon_{max}$ ensures that no state in the original system is poorly controlled under the switching policy computed in the BMDP abstraction arising from $P$. Looser notions of optimality, such as \textit{the average suboptimality factor} or \textit{the fraction of states} below a fixed optimality threshold, are less sensitive to outliers and can alternatively be considered. We denote the desired suboptimality target by $\epsilon_{thr}$. Note that a target $\epsilon_{thr}$ equal to $0$ requires to find an optimal action for all states in $\mathcal{B} \otimes \mathcal{A}$ in the case of maximization \textcolor{black}{or in $\mathcal{B} \otimes \overline{\mathcal{A}}$ for the case of minimization}.

Formally, a partition $P'$ is a refinement of a coarser partition $P$ if all states in $P$ is equal to the union of a set of states in $P'$. In the general case, abstractions constructed from a refinement $P'$ of $P$ will exhibit a lesser degree of non-determinism than abstractions constructed from $P$, allowing for the computation of \textcolor{black}{higher-quality} controllers with respect to the abstracted system. \\

\begin{definition}[Partition Refinement]
A partition $P'$ is a refinement of a partition $P$ if, for all states $Q_{j} \in P$, there exists a set of states $\{ Q^{k}_{j'}\}_{k = 0}^{m_{j}}$ in $P'$ such that $Q_{j} = \cup_{k=0}^{m_j} Q^{k}_{j'}$.
\end{definition}\mbox{}

\noindent The proposed refinement procedure to achieve a target precision $\epsilon_{thr}$ is inspired by our technique in \cite{dutreix2020specification} where refinement was conducted for the purpose of verification in an IMC \textcolor{black}{and whose main features are extended to the synthesis problem at hand}. This \textcolor{black}{new} procedure is based on a heuristical scoring of the states in a partition $P$ which highlights the regions of the state-space causing the most uncertainty with respect to the specification of interest and the set of actions at hand. Specifically, this score aims to capture how differently a partition state behaves between the extreme cases induced by the two maximizing (or minimizing) policies previously discussed, as well as how much this state influences other states which are known to be suboptimaly controlled. 

Our scoring algorithm is presented in Algorithm \ref{ScorAlgSyn} and is summarized as follows: first, we take as input  a ``best-case" product MC $(\mathcal{M}_{\otimes}^{\mathcal{A}})_{u}$ and a ``worst-case" product MC $(\mathcal{M}_{\otimes}^{\mathcal{A}})_{l}$. For the case of maximization, the worst-case product MC $(\mathcal{M}_{\otimes}^{\mathcal{A}})_{l}$ is \textcolor{black}{a} worst-case product MC induced by the IMC $(\mathcal{B} \otimes \mathcal{A})[(\widehat{\mu}_{\Psi}^{low})_{\otimes}]$ with respect to the objective of reaching an accepting BSCC, while the best-case product MC $(\mathcal{M}_{\otimes}^{\mathcal{A}})_{u}$ is \textcolor{black}{a} best-case product MC induced by the IMC $(\mathcal{B} \otimes \mathcal{A})[(\widehat{\mu}_{\Psi}^{up})_{\otimes}]$. Similarly, for the case of minimization, the worst-case product MC $(\mathcal{M}_{\otimes}^{\mathcal{A}})_{l}$ is \textcolor{black}{a} worst-case product MC induced by the IMC \textcolor{black}{$(\mathcal{B} \otimes \overline{\mathcal{A}})[(\widecheck{\mu}_{\Psi}^{up})_{\otimes}]$} with respect to the objective of reaching an accepting BSCC, while the best-case product MC $(\mathcal{M}_{\otimes}^{\mathcal{A}})_{u}$ is \textcolor{black}{a} best-case product MC induced by the IMC \textcolor{black}{$(\mathcal{B} \otimes \overline{\mathcal{A}})[\widecheck{\mu}_{\Psi}^{low}]$}. Again, the aforementioned MCs are automatically constructed when applying the \textcolor{black}{reachability} value iteration algorithm used \textcolor{black}{in Algorithms \ref{AlgPerBSCCW} and  \ref{AlgPerWin}} and for designing the two maximizing (or minimizing) policies.  

Next, for all state $\left< Q_{j}, s_{i} \right>$ of the product BMDP $\mathcal{B} \otimes \mathcal{A}$ \textcolor{black}{(or $\mathcal{B} \otimes \overline{\mathcal{A}}$ )} whose suboptimality factor is greater than the target $\epsilon_{thr}$, we compute the probability $p_{\left< j, i \right> \rightarrow \left< j', i' \right>}$ of reaching any state $\langle Q_{j'}, s_{i'} \rangle$ from $\left< Q_{j}, s_{i} \right>$  in the MC $(\mathcal{M}_{\otimes}^{\mathcal{A}})_{u}$ on line 7 using the results in \cite{katoen2013model}. Then, for all states $\langle Q_{j'}, s_{i'} \rangle$ of the product BMDP that do not belong to a permanent component (as these do not require refinement), the quantity $p_{\left< j, i \right> \rightarrow \left< j', i' \right>}  \cdot ||T^{u}_{\left< j', i' \right>}  - T^{\ell}_{ \left< j', i' \right>} ||_{2}$ is added to the score $\sigma_{j'}$ of the partition state $Q_{j'}$ on \textcolor{black}{line 10}, where $T^{u}_{\left< j', i' \right>}$ and $T^{\ell}_{ \left< j', i' \right>} $ are the rows corresponding to state $\langle Q_{j'}, s_{i'} \rangle$ in the transition matrices of $(\mathcal{M}_{\otimes}^{\mathcal{A}})_{u}$ and $(\mathcal{M}_{\otimes}^{\mathcal{A}})_{l}$ respectively. The term $||T^{u}_{\left< j', i' \right>}  - T^{\ell}_{ \left< j', i' \right>} ||_{2}$ aims to capture how differently state $\langle Q_{j'}, s_{i'} \rangle$ behaves in the two extreme MCs, while $p_{\left< j, i \right> \rightarrow \left< j', i' \right>}$ is a term associated with how much state $\langle Q_{j'}, s_{i'} \rangle$ affects state $\langle Q_{j}, s_{i} \rangle$.  Finally, from line 10 to 13, we additionally increment the score of states which have the potential of changing the qualitative connectivity structure of the ``best" and ''worst" case scenarios. These states are those which belong to a BSCC that is present in one of the scenarios and not in the other and have the potential of confirming or invalidating the existence of these BSCCs, that is, states which have an outgoing transition with a zero lower bound and a non-zero upper bound for at least one available \textcolor{black}{(non-suboptimal)} control action.

Once a score is attributed to each state of $P$ via Algorithm \ref{ScorAlgSyn}, states with a score above a user-defined threshold are refined to generate a finer partition $P'$. A new switching policy is computed in a BMDP abstraction constructed from $P'$, and more refinement steps are subsequently applied if necessary. The procedure terminates once the optimality factor $\epsilon_{max}$ becomes less than the target $\epsilon_{thr}$.

\textcolor{black}{It should be noted that a product IMC generally does not induce a unique worst-case and best-case MC, but rather induces sets of possible worst-case and best-case MCs yielding the same probabilities of reaching an accepting BSCC from all states \cite{dutreix2020specification}. Therefore, the choice of inputs for Algorithm \ref{ScorAlgSyn} may not be unique. As previously discussed, we choose to input the MCs computed in the process of designing the control policies for the BMDP. Although selecting other MCs is possible, we claim that the design of Algorithm \ref{ScorAlgSyn} renders the effect of choosing other input MCs negligible in the long-term behavior of the refinement-based synthesis algorithm in all but pathological cases. The reasoning behind this claim is that a lot of discrepancies between different worst-case (or best-case) MCs occur in the transitions within permanent winning or losing components which belong to the set $G$ defined in Line 5 of Algorithm \ref{ScorAlgSyn}, as the existence of these components depend on the qualitative structure of the IMC and not on the exact transition values, and have no influence on the computation of the refinement scores. Other large discrepancies between different such MCs may be found in the potential BSCCs stored in set $\textcolor{black}{R}$ in Line 4. However, the relative difference captured by the term $||T^{u}_{\left< j', i' \right>}  - T^{\ell}_{ \left< j', i' \right>} ||_{2}$ at such states is likely to be similar regardless, causing only a minor variation in refinement scores for two different input MCs, and the set $\textcolor{black}{R}$ may quickly become empty after a few iterations of the refinement algorithm as the states causing this set to exist, namely states without zero lower bound and non-zero upper bound, are heavily targeted in Line 12 to 14. Finally, different transitions between two worst-case (or best-case) MCs can be found outside of the aforementioned sets, but this scenario is improbable for abstractions computed from dynamical systems with continuous state-spaces. Indeed, this would require for at least two states outside these sets to have the exact same probability of reaching an accepting or a non-accepting BSCC in the extremal assignments of the transition probabilities, which is highly unlikely when using transition bounds derived from integrals over continuous sets. Such a scenario may be encountered on coarse abstractions with very few states that are all bound to be refined no matter which best or worst-case MC is chosen, and finding such states with equal reachability probabilities in high-dimensional MCs would be implausible or an isolated event with little impact on the refinement algorithm.}

\begin{algorithm}[H]
\caption{Refinement Scoring Algorithm} 
\begin{algorithmic}[1]
\STATE \textbf{Input}: Product BMDP $\mathcal{B} \otimes \mathcal{A}$, best-case product MC $(\mathcal{M}_{\otimes}^{\mathcal{A}})_{u}$, worst-case product MC $(\mathcal{M}_{\otimes}^{\mathcal{A}})_{l}$, threshold suboptimality factor $\epsilon_{thr}$, suboptimality factors $\epsilon_{\left< Q_{j}, s_{i} \right>}$ for all states $\left< Q_{j}, s_{i} \right>$ of $\mathcal{B} \otimes \mathcal{A}$
\STATE \textbf{Output}: Refinement scores $\sigma = \left[\sigma_{0}, \sigma_{1}, \ldots, \sigma_{|Q|-1} \right]$ for all states of partition $P$
\STATE \textbf{Initialize}: $\sigma = \left[\sigma_{0}, \sigma_{1}, \ldots, \sigma_{|Q|-1} \right]$ where $\sigma_{i} = 0$
\STATE In $\textcolor{black}{R}$, list all states of $\mathcal{B} \otimes \mathcal{A}$ belonging to a BSCC that exists in $(\mathcal{M}_{\otimes}^{\mathcal{A}})_{u}$ and not in $(\mathcal{M}_{\otimes}^{\mathcal{A}})_{l}$, or vice-versa
\STATE In $G$, list all states of $\mathcal{B} \otimes \mathcal{A}$ with a probability of reaching an accepting BSCC of 0 in both $(\mathcal{M}_{\otimes}^{\mathcal{A}})_{u}$ and $(\mathcal{M}_{\otimes}^{\mathcal{A}})_{l}$ or of 1 in both $(\mathcal{M}_{\otimes}^{\mathcal{A}})_{u}$ and $(\mathcal{M}_{\otimes}^{\mathcal{A}})_{l}$
\FOR{$\left<Q_{j}, s_{i} \right> \in \mathcal{B} \otimes \mathcal{A}$}
\IF{$\epsilon_{\left< Q_{j}, s_{i} \right>} \geq \epsilon_{thr}$}
\STATE Compute the probability $p_{\left< j, i \right> \rightarrow \left< j', i' \right>}$ of reaching  $\langle Q_{j'}, s_{i'} \rangle$ from $\left<Q_{j}, s_{i} \right>$ in $(\mathcal{M}_{\otimes}^{\mathcal{A}})_{u}$, for all $\langle Q_{j'}, s_{i'} \rangle \in \mathcal{B} \otimes \mathcal{A}$, using the technique in \cite{katoen2013model}
\FOR{$\langle Q_{j'}, s_{i'} \rangle \in \mathcal{B} \otimes \mathcal{A}$ such that $\langle Q_{j'}, s_{i'} \rangle \not \in G$}
\STATE $\sigma_{j'} = \sigma_{j'} + p_{\left< j, i \right> \rightarrow \left< j', i' \right>}  \cdot ||T^{u}_{\left< j', i' \right>}  - T^{\ell}_{ \left< j', i' \right>} ||_{2}$, where $T^{u}_{\left< j', i' \right>}$ and $T^{\ell}_{ \left< j', i' \right>} $ are the rows corresponding to state $\langle Q_{j'}, s_{i'} \rangle$ in the transition matrices of $(\mathcal{M}_{\otimes}^{\mathcal{A}})_{u}$ and $(\mathcal{M}_{\otimes}^{\mathcal{A}})_{l}$ respectively
\IF{ $\langle Q_{j'}, s_{i'} \rangle \in \textcolor{black}{R}$}
\FOR{$\langle Q_{j''}, s_{i''} \rangle \in \mathcal{B} \otimes \mathcal{A}$ such that $\langle Q_{j'}, s_{i'} \rangle$ and $\langle Q_{j''}, s_{i''} \rangle$ belong to a common  BSCC in $(\mathcal{M}_{\otimes}^{\mathcal{A}})_{u}$ or $(\mathcal{M}_{\otimes}^{\mathcal{A}})_{l}$}
\IF{$\langle Q_{j''}, s_{i''} \rangle$ has an outgoing transition with a zero lower bound and a non-zero upper bound for at least one available \textcolor{black}{(non-suboptimal)} control action}
\STATE $\sigma_{j''} = \sigma_{j''} + p_{\left< j, i \right> \rightarrow \left< j', i' \right>}  \cdot ||T^{u}_{\left< j', i' \right>}  - T^{\ell}_{ \left< j', i' \right>} ||_{2}$
\ENDIF
\ENDFOR
\ENDIF
\ENDFOR
\ENDIF
\ENDFOR
\end{algorithmic}
\label{ScorAlgSyn}
\end{algorithm}

The fact that a partition $P'$ is a refinement of a partition $P$ allows us to make inferences about the properties of the states in $P'$ from the synthesis computations previously performed on the states in $P$. First, as discussed in the previous subsection, not all actions allowed in $P$ may need to be considered in the refined partition $P'$ when computing a new switching policy. Indeed, given a partition $Q_{j} = \cup_{k=0}^{m_j} Q^{k}_{j'}$ of a state $Q_{j} \in P$, it follows that a certainly suboptimal action with respect to the action set of a product state $\left<Q_{j}, s_{i} \right>$ will also be suboptimal with respect to all $\left<Q^{k}_{j'} , s_{i} \right>$ and can be eliminated in the synthesis procedure applied to $P'$.\\

\begin{proposition}
\label{PropSubop}
Let $\mathcal{B}$ be a BMDP abstraction constructed from a partition $P$ of the domain $D$ of \eqref{eq1}, $\mathcal{A}$ be a DRA corresponding to specification $\Psi$, and $P'$ be a refinement of $P$. Let $\{ Q^{k}_{j'}\}_{k = 0}^{m_{j}} \subseteq P'$, be a partition of state $Q_{j} \in P$. If action $a \in A(\left<Q_{j}, s_{i} \right>)$ is suboptimal for state $\left<Q_{j}, s_{i} \right>$ with respect to $A(\left<Q_{j}, s_{i} \right>)$ in the product BMDP $\mathcal{B} \otimes \mathcal{A}$, then the mode of \eqref{eq1} represented by action $a$ is suboptimal for all $x \in Q_{j}$ with respect to the automaton state $s_{i}$ and the set of available modes, and, in particular, for all $x \in Q^{k}_{j'}$, $k = 0, 1 \ldots, m_{j}$.
\end{proposition}
\begin{proof}
The proof assumes the objective of synthesis to be the maximization of the probability of satisfying $\Psi$. We denote by $\widehat{p}$ an upper bound on the maximum upper bound probability of reaching an accepting BSCC in $\mathcal{B} \otimes \mathcal{A}$ from $\left<Q_{j}, s_{i} \right>$ achievable over all memoryless policies choosing action $a \in A(\left<Q_{j}, s_{i} \right>)$ at state $\left<Q_{j}, s_{i} \right>$. The assumption that $a$ is suboptimal with respect to $A(\left<Q_{j}, s_{i} \right>)$ in $\mathcal{B} \otimes \mathcal{A}$ implies that there exists an action $a' \in A(\left<Q_{j}, s_{i} \right>)$ with a known a lower bound $\widecheck{p}'$ on the maximum lower bound probability of reaching an accepting BSCC in $\mathcal{B} \otimes \mathcal{A}$ from $\left<Q_{j}, s_{i} \right>$ achievable over all memoryless policies choosing action $a' \in A(\left<Q_{j}, s_{i} \right>)$ and such that $\widehat{p} < \widecheck{p}'$. Therefore, by virtue of $\mathcal{B}$ being an abstraction of \eqref{eq1}, $\forall x \in Q_{j}$, it follows that $\widehat{p}_{mode} < \widecheck{p}'_{mode}$, where $\widehat{p}_{mode}$ and $\widecheck{p}'_{mode}$ are a lower bound and an upper bound on the maximum probability that an infinite path of \eqref{eq1} with prefix $\pi = x[0]x[1]x[2] \ ... \ x[k]$, $x[k] =: x$, such that the word $L(x[0])L(x[1])L(x[2]) \ ... \ L(x[k])$ produces a run $s[0]s[1]s[2] \ldots s[k]$, with $s[k] = s_{i}$, satisfies $\Psi$ over all the (memoryless in the product) policies of \eqref{eq1} choosing the modes represented by actions $a$ and $a'$ respectively at path $\pi$. It follows that the mode represented by action $a$ is suboptimal for all $x \in Q_{j}$ with respect to automaton state $s_{i}$ and the set of available modes. In particular, this statement is true for all $x \in Q^{k}_{j'}$, $k = 0, 1 \ldots, m_{j}$, since $Q^{k}_{j'} \subseteq Q_{j}$, proving the proposition. Symmetric arguments prove this proposition \textcolor{black}{in the case of minimization}.
\end{proof}\mbox{}

Furthermore, out of the remaining actions, only a subset of them may be retained for the qualitative problems of constructing the largest and permanent components in $P'$ using Algorithms \ref{AlgPerBSCCW} \textcolor{black}{and \ref{AlgPerWin}}. Indeed, all actions in $A(\left<Q_{j}, s_{i} \right>)$ which were discarded during the graph search for $(WC)_{L}^{G}$ could not, under any policy and adversary, generate a winning component in $\mathcal{B} \otimes \mathcal{A}$. Therefore, based on this fact, we can define the set of actions $A_{qual}(\left<Q^{k}_{j'} , s_{i} \right>) \subseteq A(\left<Q^{k}_{j'} , s_{i} \right>)$ used specifically for the component graph search and containing all actions which, at state $\left<Q_{j}, s_{i} \right>$, allowed for the existence of $(WC)_{L}^{G}$ with respect to the partition $P$.\\

\begin{proposition} 
\label{PropLargest}
Let $\mathcal{B}$ be a BMDP abstraction constructed from a partition $P$ of the domain $D$ of \eqref{eq1}, $\mathcal{A}$ be a DRA corresponding to specification $\Psi$, \textcolor{black}{$\overline{\mathcal{A}}$ be a DRA corresponding to complement specification $\overline{\Psi}$}, and $P'$ be refinement of a partition $P$. If state $\left<Q_{j} , s_{i} \right>$ is not a member of $(WC)_{L}^{G}$ in the product BMDP $\mathcal{B} \otimes \mathcal{A}$ (respectively, \textcolor{black}{$\mathcal{B} \otimes \overline{\mathcal{A}}$}) under any memoryless policy $\mu$ of $\mathcal{B} \otimes \mathcal{A}$ (respectively, \textcolor{black}{$\mathcal{B} \otimes \overline{\mathcal{A}}$}) such that $\mu (\left<Q_{j} , s_{i} \right>) = a \in  A(\left<Q_{j} , s_{i} \right>)$, then, for all $x \in Q_{j}$, the probability that an infinite path with prefix $\pi = x[0]x[1]x[2] \ ... \ x[k]$, $x[k] =: x$, such that the word $L(x[0])L(x[1])L(x[2]) \ ... \ L(x[k])$ produces a run $s[0]s[1]s[2] \ldots s[k]$, with $s[k] = s_{i}$ in automaton $\mathcal{A}$, satisfies $\Psi$ is strictly less than 1 (respectively, strictly greater than 0) for all policies of \eqref{eq1} choosing the mode represented by action $a$ at state $x$. In particular, this statement is true for all $x \in Q^{k}_{j'}$, $k = 0, 1 \ldots, m_{j}$, where $\{ Q^{k}_{j'}\}_{k = 0}^{m_{j}}$, $Q^{k}_{j'} \in P'$, is a partition of state $Q_{j} \in P$.
\end{proposition}
\begin{proof}
The proof assumes the objective of synthesis to be the maximization of the probability of $\Psi$. If state $\left<Q_{j} , s_{i} \right>$ is not a member of $(WC)_{L}^{G}$ under any memoryless policy $\mu$ such that $\mu (\left<Q_{j} , s_{i} \right>) = a$, then it must be true that $\widehat{p} <1$, where $\widehat{p}$ is an upper bound on the probability of $\left<Q_{j} , s_{i} \right>$ to reach an accepting BSCC in $\mathcal{B} \otimes \mathcal{A}$ under all memoryless policies $\mu$ such that $\mu (\left<Q_{j} , s_{i} \right>) = a$. Therefore, by virtue of $\mathcal{B}$ being an abstraction of \eqref{eq1}, it follows that the probability of an infinite path with prefix $\pi = x[0]x[1]x[2] \ ... \ x[k]$, $x[k] =: x$, such that the word $L(x[0])L(x[1])L(x[2]) \ ... \ L(x[k])$ produces a run $s[0]s[1]s[2] \ldots s[k]$, with $s[k] = s_{i}$ in automaton $\mathcal{A}$ to satisfy $\Psi$ is upper bounded by $\widehat{p}$ for all policies of \eqref{eq1} choosing the mode represented by action $a$ for the path $\pi$ and is thus strictly less than 1. In particular, this statement is true for all $x \in Q^{k}_{j'}$, $k = 0, 1 \ldots, m_{j}$, since $Q^{k}_{j'} \subseteq Q_{j}$, proving the proposition. Symmetric arguments prove the proposition with respect to \textcolor{black}{the minimization objective}.  
\end{proof}
\mbox{}

\textcolor{black}{An analogous proposition can be established with respect to the greatest BSCCs $(U)_{L}^{G}$ for Algorithm \ref{AlgPerBSCCW}.}

We also remark that any state $\left<Q_{j}, s_{i} \right>$ belonging to the greatest permanent \textcolor{black}{winning} components $(WC)_{P}^{G}$ of a BMDP abstraction $\mathcal{B} \otimes \mathcal{A}$ constructed from a partition $P$ has to belong \textcolor{black}{to} the greatest permanent components with respect to a refined partition $P'$ if the same control action applied to all $\left<Q_{j}, s_{i} \right> \in (WC)_{P}^{G}$ in the abstraction resulting from $P$ is applied to all their refinement states $\left<Q^{k}_{j'} , s_{i} \right>$. \\

\begin{proposition}
\label{PropCompPer}
Let $\mathcal{B}$ be a BMDP abstraction constructed from a partition $P$ of the domain $D$ of \eqref{eq1}, $\mathcal{A}$ be a DRA corresponding to specification $\Psi$, \textcolor{black}{$\overline{\mathcal{A}}$ be a DRA corresponding to complement specification $\overline{\Psi}$}, and $P'$ be refinement of a partition $P$. A policy $\mu$ of $\mathcal{B}$ induced by a policy in $\mathcal{B} \otimes \mathcal{A}$ (\textcolor{black}{respectively, $\mathcal{B} \otimes \overline{\mathcal{A}}$ in the case of minimization)} generating the greatest permanent winning component $(WC)^{G}_{P}$ of $\mathcal{B} \otimes \mathcal{A}$ (\textcolor{black}{respectively, of $\mathcal{B} \otimes \overline{\mathcal{A}}$)} selects an optimal mode (with the appropriate mode/action correspondence) for all $x \in Q_{j}$ such that $\left<Q_{j}, s_{i} \right> \in (WC)^{G}_{P}$ with respect to the automaton state $s_{i}$ and the set of available modes, and, in particular, for all $x \in Q^{k}_{j'}$, $k = 0, 1 \ldots, m_{j}$, where $\{ Q^{k}_{j'}\}_{k = 0}^{m_{j}}$, $Q^{k}_{j'} \in P'$, is a partition of state $Q_{j} \in P$.
\end{proposition}
\begin{proof}
The proof assumes the objective of synthesis to be the maximization of the probability of $\Psi$. A policy $(\mu)_{\otimes}$ generating $(WC)^{G}_{P}$ in $\mathcal{B} \otimes \mathcal{A}$ ensures that $\widecheck{\mathcal{P}}(\left<Q_{j}, s_{i} \right> \models \Diamond (WC)^{G}_{P}) = 1$ for all $\left<Q_{j}, s_{i} \right> \in (WC)^{G}_{P}$. The policy $\mu$ in $\mathcal{B}$ induced by $(\mu)_{\otimes}$ applied to all $x \in Q_{j}$ such that $\left<Q_{j}, s_{i} \right> \in (WC)^{G}_{P}$ when the automaton state is $s_{i}$ with the appropriate mode/action correspondence guarantees that, for all such $x$, the probability of an infinite path with prefix $\pi = x[0]x[1]x[2] \ ... \ x[k]$, $x[k] =: x$, such that the word $L(x[0])L(x[1])L(x[2]) \ ... \ L(x[k])$ produces a run $s[0]s[1]s[2] \ldots s[k]$, with $s[k] = s_{i}$ in automaton $\mathcal{A}$ to satisfy $\Psi$ is equal to 1, by virtue of $\mathcal{B}$ being an abstraction of \eqref{eq1}. Therefore, $\mu$ selects an optimal mode for all such $x$. In particular, this statement is true for all $x \in Q^{k}_{j'}$, $k = 0, 1 \ldots, m_{j}$, since $Q^{k}_{j'} \subseteq Q_{j}$,  proving the proposition. Symmetric arguments prove the proposition \textcolor{black}{with respect to the minimization case}.
\end{proof}\mbox{}

Therefore, by pruning all states which were a member of $(WC)_{P}^{G}$ in an abstraction constructed $P$, since an action engendering a fixed probability of reaching an accepting BSCC equal to 1 is known for such states, we can reduce the effective set of states for which a controller has to be synthesized in the abstraction arising from a refined partition $P'$ after each refinement step.

Finally, additional crucial information can be exploited to tremendously reduce the number of operations performed in a refined partition. For example, in the numerical examples presented further, all states which were shown to be reachable from a given state $Q_{j}$ under some action in partition $P$ are stored in memory, and only these states or their subsets are inspected for computing the transitions from $Q_{j}$ in the abstraction arising from a refined partition $P'$. This is justified by the fact that, if $\widehat{T}(Q_{1}, Q_{2}) = 0$ for any $Q_{1}$ and $Q_{2}$ in partition $P$, then it follows that $\widehat{T}(Q^{k}_{1}, Q^{k}_{2}) = 0$ for any $Q^{k}_{1} \subseteq Q_{1}$ and $Q^{k}_{2} \subseteq Q_{2}$. Finding other structural properties which are transmitted from one partition to its refined versions will be the focus of future research.

This novel iterative approach \textcolor{black}{that} removes suboptimal actions at each refinement step is promising in terms of scalability compared to \textcolor{black}{existing methods. For instance, prominant tools such as StocHy \cite{cauchi2019stochy} and FAUST$^{2}$ \cite{soudjani2014faust} employ a single gridding approach where a unique (often conservative) partition of the domain guaranteeing a target abstraction error is created and used for computing a switching policy; in this case, all possible actions allowed by the original abstracted system have to be considered on possibly very fine partition grids, causing intractability issues when the action space is large. Here, the action space to be analyzed for the refined states is likely to reduce in size as the partition is progressively rendered finer. Therefore, the number of computations performed to synthesize a switching policy for an equivalent level of abstraction fineness is reduced compared to the aforementioned tools. Furthermore, the continuous domain grid in StocHy and FAUST$^{2}$ depends primarily on the properties of the abstracted system whereas our refinement is specification-guided, i.e, tailored to the specification under consideration only, diminishing the generation of unnecessary discrete states. The iterative refinement technique proposed in \cite{lahijanian2015formal} does not implement an action removal scheme and therefore suffers from the same tractability issues discussed above. In addition, structural properties inherited from coarser partitions are not discussed  and leveraged to lessen the computational burden of synthesis as done in our algorithm. Also, the termination criterion of the algorithm in \cite{lahijanian2015formal} is a low abstraction error under the lower bound maximizing (or upper bound minimizing) policy which, unlike the suboptimality factor introduced in this work, does not directly capture the possible improvement one could achieve by choosing a different policy (which is memoryless in the product construction) on a refined abstraction. Lastly, the selection of states to be refined in \cite{lahijanian2015formal} focuses on one-step transition errors and does not involve the inspection of the overall structure of the abstraction between the two extreme scenarios of the BMDP as in Algorithm \ref{ScorAlgSyn}.}

Our specification-guided, refinement-based synthesis procedure for finite-mode systems is summarized in Algorithm \ref{AlgFiniteMode}. We assume that states selected by the scoring scheme are split in half along their greatest dimension. In this case, the worst-case growth of the \textcolor{black}{BMDP} abstraction throughout this refinement-based synthesis procedure is $\mathcal{O}(|S| \cdot |Act| \cdot 2^{|Q|})$ when every state in the partition is refined. However, the iterative removal of considered actions, coupled with the scoring algorithm targeting only specific regions of the domain, mitigates this exponential growth in practice. \textcolor{black}{The run-time complexity of the sub-components of Algorithm \ref{AlgFiniteMode} is as follows: Algorithm \ref{AlgPerBSCCW} is exponential in $|S| \cdot |Q|$ as the number of SCCs to analyze may grow exponentially in the worst-case; consequently, Algorithm \ref{AlgPerWin}, which calls Algorithm \ref{AlgPerBSCCW}, displays the same run-time complexity; the iterative reachability maximization algorithm on the winning components is polynomial in $|Act| \cdot |S| \cdot |Q|$ \cite{lahijanian2015formal} and Algorithm \ref{ScorAlgSyn}, whose limiting factor is the computation reachability probabilities in MCs, is therefore polynomial in $|S| \cdot |Q|$}.

\begin{algorithm}[H]
\caption{Controller Synthesis for Finite-mode Systems} 
\begin{algorithmic}[1]
\STATE \textbf{Input}: Partition $P_{0}$ of domain $D$ of \eqref{eq1}, $\omega$-regular property $\Psi$ \textcolor{black}{(complement property \textcolor{black}{$\overline{\Psi}$})} and corresponding DRA $\mathcal{A}$ \textcolor{black}{($\overline{\mathcal{A}}$)}, target controller precision $\epsilon_{thr}$ \hspace{-0.6em}
\STATE \textbf{Output}: Maximizing (minimizing) switching policy $\widehat{\mu}_{\Psi}^{low}$ ($\widecheck{\mu}_{\Psi}^{up}$), final partition $P_{fin}$
\STATE \textbf{Initialize}: $\epsilon_{max} := 1$, $i := 0$
\WHILE{$\epsilon_{max} > \epsilon_{thr}$}
\STATE Compute the sets $(WC)_{P}^{G}$ and $(WC)_{L}^{G}$ of the product BMDP $\mathcal{B} \otimes \mathcal{A}$ \textcolor{black}{($\mathcal{B} \otimes \overline{\mathcal{A}}$)} constructed from $P_{i}$ using Algorithms \ref{AlgPerBSCCW} \textcolor{black}{and \ref{AlgPerWin}}
\STATE Compute the policies $\widehat{\mu}_{\Psi}^{low}$ and $\widehat{\mu}_{\Psi}^{up}$ ($\widecheck{\mu}_{\Psi}^{up}$ and $\widecheck{\mu}_{\Psi}^{low}$) of the BMDP $\mathcal{B}$ according to Subsections \ref{SecBMDPSynth} 
\STATE Compute $\epsilon_{max}$ using \eqref{maxsub}
\IF{$\epsilon_{max} > \epsilon_{thr}$}
\STATE Compute \textcolor{black}{a} best-case and worst-case product MC $(\mathcal{M}_{u})_{\otimes}^{\mathcal{A}}$ and $(\mathcal{M}_{l})_{\otimes}^{\mathcal{A}}$ as discussed in Subsection \ref{SubSubRef}
\STATE Apply the scoring procedure in Algorithm \ref{ScorAlgSyn} and refine all states above a user-defined threshold score to produce $P_{i+1}$
\STATE Update the set of actions of all states in $P_{i+1}$ for the component search and reachability problem as discussed in Subsection \ref{SubSubRef}
\STATE $i := i + 1$
\ENDIF
\ENDWHILE
\RETURN $\widehat{\mu}_{\Psi}^{low}$ ($\widecheck{\mu}_{\Psi}^{up}$), $P_{fin} := P_{i}$
\end{algorithmic}
\label{AlgFiniteMode}
\end{algorithm}

\subsubsection{MONOTONICITY AND CONVERGENCE OF SYNTHESIS PROCEDURE}
\label{MonDisSec}

As pointed out in \cite{dutreix2020specification}, it is possible to construct scenarios where, for two states $Q_{i}$ and $Q_{j}$ in a given partition, and two states $Q'_{j}$ and $Q''_{j}$ generated from a refinement of $Q_{j}$, that is, $Q_{j} = Q'_{j} \cup  Q''_{j}$, the inequality $\widehat{T}_{ex}(Q_i, a, Q_j) < \widehat{T}_{ex}(Q_i, a, Q'_j) + \widehat{T}_{ex}(Q_i, a , Q''_j)$ holds for some mode $a$ of system \eqref{eq1}, where $\widehat{T}_{ex}(Q_i, a, Q_j)$ returns the least upper bound on the probability for any continuous state $x \in Q_i$ to transition to a state in $Q_j$ under mode $a$. As a consequence, because the current implementations of the graph search and reachability maximization algorithms view the abstractions created from a partition and its refinements as being independent from one another, our synthesis algorithm may assign a larger amount of probability to the transition from state $Q_{i}$ to the total refined states constituting $Q_{j}$ in the refined abstractions than was allowed in the coarser ones. This phenomenon may cause:
\begin{itemize}
\item The \textcolor{black}{set $(WC)_{L}^{G}$} to increase and \textcolor{black}{the set $(WC)_{P}^{G}$} to decrease upon refinement. Specifically, given a state $\left< Q_j, s_{i} \right>$ of a product BMDP $\mathcal{B} \otimes \mathcal{A}$ constructed from a partition $P$, and a state $\langle Q'_j, s_{i} \rangle$ of a product BMDP $\mathcal{B'} \otimes \mathcal{A}$ constructed from a refinement $P'$ of $P$, where $Q_j' \subset Q_j$, it is possible for $\langle Q'_j, s_{i} \rangle$ to belong to $(WC)_{L}^{G}$ in $\mathcal{B'} \otimes \mathcal{A}$ while $\langle Q_j, s_{i} \rangle$ does not belong to \textcolor{black}{this set} in $\mathcal{B} \otimes \mathcal{A}$, and it is possible for $\langle Q_j, s_{i} \rangle$ to belong to $(WC)_{P}^{G}$ in $\mathcal{B} \otimes \mathcal{A}$ while $\langle Q'_j, s_{i} \rangle$ does not belong to \textcolor{black}{this set} in $\mathcal{B'} \otimes \mathcal{A}$,
\item The lower bound probabilities of reaching $(WC)_{P}^{G}$ to decrease from some states of the product BMDP for a fixed policy, and the upper bound probability of reaching $(WC)_{L}^{G}$ to increase from some states of the product BMDP for a fixed policy.
\end{itemize}

\noindent Therefore, a finer partition could provide ``less certainty" and result in the synthesis of a switching policy yielding a smaller satisfaction lower bound \textcolor{black}{(or greater upper bound in the case of minimization)} for some states of the refined BMDP abstraction. This means that a monotone decrease of the greatest suboptimality factor $\epsilon_{max}$ is not guaranteed under the proposed iterative refinement method. We address the first bullet point by saving the states that belong to the aforementioned components in the coarser abstraction before each refinement step and using the facts enunciated in Propositions \ref{PropLargest} and \ref{PropCompPer}; however, the second bullet point affects the monotonicity of the value iteration algorithm of \cite{lahijanian2015formal} in its current state.

Nonetheless, under a continuity assumption on the dynamics and using adequate BMDP abstraction techniques, it seems that having the size of all discrete states which are not in a permanent component approach zero in the limit is sufficient for guaranteeing convergence of Algorithm \ref{AlgFiniteMode}, as seen in related case studies using iterative refinement \cite{lahijanian2015formal}, \cite{dutreix2020specification} and the case study presented further. We conjecture that the scoring and refinement procedure applied in Algorithm \ref{AlgFiniteMode} satisfies this condition and therefore ensures convergence; however, we leave a thorough investigation and potential formal proof of these facts for future work. Modifying the value iteration algorithm in \cite{lahijanian2015formal} to exploit all information obtained from coarser partitions and enforce monotonicity of the overall procedure is another immediate research direction.\\

In brief, we introduce a quantitative measure of the suboptimality of the devised switching policy in a BMDP abstraction with respect to the original continuous abstracted states. This suboptimality factor defined through \eqref{SubFacMax} and \eqref{maxsub} corresponds to an upper bound on the potential improvement any continuous state of the system could experience in the probability of satisfying the specification \textcolor{black}{using memoryless (in the product) policies} by choosing a different control action from the one prescribed by the computed policy. This factor is established in the BMDP abstraction through a comparison between the worst-case assignment of the probability intervals under the computed policy and the best-case assignment of these probabilities under a policy assuming the most optimistic outcome of the transition intervals. Furthermore, these worst-case and best-case scenarios are used to identify control actions that are certainly suboptimal for a given state as formalized in Proposition \ref{PropSubop}. Lastly, in Algorithm \ref{AlgFiniteMode}, we presented an iterative partition refinement \textcolor{black}{heuristic} which selectively targets certain regions of the state-space by comparing these two extreme scenarios \textcolor{black}{with the objective of achieving} a user-defined precision threshold. Some structural properties transmitted from coarser abstractions to refined ones are identified in Proposition \ref{PropLargest} and \ref{PropCompPer}, allowing to reduce the number of required computations after each refinement step. 

\textcolor{black}{While the techniques derived in this section are applicable to finite mode stochastic systems, they do not straightforwardly extend to the synthesis of control policies for stochastic systems with a continuous set of available inputs as stated in Problem 2. Indeed, the latter systems cannot be abstracted into BMDPs and CIMC abstractions have to be employed instead. The next section first discusses Subproblem 2.1 and the synthesis of control policies for CIMC abstractions with $\omega$-regular objectives in Subsections \ref{SubsecCompCons} and \ref{ReachMaxCont}. Then, an abstraction refinement scheme  directed to Subproblem 2.2 is proposed in Subsection \ref{ContSpaceRef} for this particular framework.}

\section{CONTROLLER SYNTHESIS FOR CONTINUOUS INPUT SYSTEMS}
\label{Continputsec}

In this section, we discuss synthesis for stochastic systems with a continuous set of inputs as defined in Problem 2. Recall that we focus our attention on systems of the form \eqref{eq3} with state update equation $x[k+1] = \mathcal{F}(x[k]) + u[k] + w[k]$.

To synthesize controllers for such systems, we again construct a finite partition $P$ of the continuous domain $D$ of \eqref{eq3} in order to generate a CIMC abstraction $\mathcal{C}$ of the system. Note that the results presented in the lemmas and theorems of Section \ref{SecSynthFin} for BMPDs are not altered if the set of available actions is infinite and consequently apply identically to CIMCs. Therefore, our approach is similar to the synthesis method for BMDPs, that is, a DRA representation $\mathcal{A}$ of the specification of interest $\Psi$ is computed, and the problem is converted to a component search and a reachability maximization step in the product CIMC $\mathcal{C} \otimes \mathcal{A}$.\\

\begin{definition}[Product Controlled Interval-valued Markov Chain]
Let $\mathcal{C} = (Q, U, \widecheck{T}, \widehat{T}, \textcolor{black}{q_{0}}, \Sigma, L)$ be a CIMC and $\mathcal{A} = (S, 2^{\Sigma}, \delta, s_0, Acc)$ be a DRA. The \emph{product} $\mathcal{C} \otimes \mathcal{A} = (Q \times S, U, \widecheck{T'}, \widehat{T'}, \textcolor{black}{q_{0}^{\otimes}},  Acc', L')$ is a CIMC defined similarly to product BMDP with the difference that a continuous set of inputs $U \subset \mathbb{R}^{m}$ replaces the finite set of actions $Act$.
\end{definition}\mbox{}

However, because the number of ``modes'' of \eqref{eq3} corresponding to different choices of input $u$ can be viewed as being uncountably infinite, the techniques established in Section \ref{SecSynthFin}, which rely on exhaustive searches over all possible actions at all states of the abstraction, cannot be applied directly in this context. Instead, we need to consider the underlying continuous dynamics of the abstracted system and exploit their relationship with the bounds of the CIMC abstraction $\mathcal{C}$.

To propose a solution to this problem, we first make the following additional assumptions on \eqref{eq3} which allow to derive closed-form expressions for the lower and upper bound transition maps $\widecheck{T}$ and $\widehat{T}$ as a function of the input parameter $u$.\\

\begin{assumption}
\label{AssumpRec}
The partition $P$ of the domain $D$ of system \eqref{eq3} \textcolor{black}{conforms to the labeling function of \eqref{eq3}} and is rectangular, that is, $\forall Q_{j} \in P$, $Q_{j} = [a_{1}^{j}, b_{1}^{j}] \times [a_{2}^{j}, b_{2}^{j}] \times \ldots \times [a_{n}^{j}, b_{n}^{j}]$.
\end{assumption}\mbox{}

\begin{assumption}
\label{AssumpReach}
For every discrete state $Q_j$ in the partition $P$ of $D$, a rectangular over-approximation of the one-step reachable set from $Q_j$ under $\mathcal{F}$, denoted by $R_{Q_{j}} = [\widecheck{r}^{j}_{1}, \widehat{r}^{j}_{1}] \times [\widecheck{r}^{j}_{2}, \widehat{r}^{j}_{2}] \times \ldots \times [\widecheck{r}^{j}_{n}, \widehat{r}^{j}_{n}]$, is available.
\end{assumption}\mbox{}

\begin{assumption}
\label{assum:uni}
The random disturbance $w[k]$ in \eqref{eq3} is of the form $w[k]= \begin{bmatrix}w_{1}[k]& w_{2}[k]& \ldots& w_{n}[k]\end{bmatrix}^{T}$, where each $w_{i} \in W_{i} \subset \mathbb{R}$ has probability density function $f_{w_i}(x_i)$, $W_i$ is an interval, and the collection $\{w_i\}_{i=1}^n$ is mutually independent. We denote by $F_{w_i}(x)=\int_{-\infty}^x f_{w_i}(\sigma) d\sigma$ the cumulative distribution function for $w_i$. Moreover, the probability density function $f_{w_i}$ for each random variable $w_{i}$ is symmetric and unimodal with mode $c_i$.
\end{assumption}\mbox{}

\textcolor{black}{For systems which cannot satisfy Assumption \ref{AssumpRec}, derivations of probability bounds using over and under-approximations of labeled regions are found in \cite{2019arXiv190101576C} and can be extended to our synthesis framework to allow for a rectangular partition}.  Assumption \ref{AssumpReach} is relevant for wide classes of systems. For example, it was shown that a rectangular over-approximation of the reachable set from any box state could be efficiently computed under mixed-monotone dynamics, which include the well-known class of monotone systems \cite{coogan2015efficient} \cite{Hirsch:2005ek}. Note that, under this assumption, an over-approximation of the reachable set of state $Q_{j}$ under $\mathcal{F}$ with an additive input $u \in U$ is a shifted version of the rectangular set $R_{Q_{j}}$, denoted by $R_{Q_{j}}^{u}$.\\

\begin{remark}
\label{RemReach}
Let $R_{Q_{j}} = [\widecheck{r}^{j}_{1}, \widehat{r}^{j}_{1}] \times [\widecheck{r}^{j}_{2}, \widehat{r}^{j}_{2}] \times \ldots \times [\widecheck{r}^{j}_{n}, \widehat{r}^{j}_{n}] \supseteq \{\mathcal{F}(x) : x\in Q_j \}$ be an over-approximation of the one-step reachable set from discrete state $Q_{j} \in P$ under the state update map $\mathcal{F}(x)$. Then, $R_{Q_{j}}^{u} = [\widecheck{r}^{j}_{1} + u_{1}, \widehat{r}^{j}_{1} + u_{1}] \times [\widecheck{r}^{j}_{2} + u_{2}, \widehat{r}^{j}_{2} + u_{2}] \times \ldots \times [\widecheck{r}^{j}_{n} + u_{n}, \widehat{r}^{j}_{n} + u_{n}] \supseteq\{F(x)+u : x\in Q_j \}$ is an over-approximation of the one-step reachable set from $Q_{j}$ under the state update map $\mathcal{F}(x) +u$.
\end{remark}\mbox{}

In \cite{dutreix2018}, we showed that under Assumptions \ref{AssumpRec} to \ref{assum:uni} and for a fixed $u$, an upper bound on the probability of transition from state $Q_{j}$ to state $Q_{\ell}$ is computed by placing the mode $c$ of disturbance $w$, restricted to the reachable set $R_{Q_{j}}^{u}$, as close as possible to the center of $Q_{\ell}$. A lower bound on this probability is computed by placing the mode of $w$ as far as possible from the center of $Q_{\ell}$.\\

\begin{fact}[\cite{dutreix2018}]
\label{FactBounds}
For system \eqref{eq3} under Assumptions \ref{AssumpRec} to \ref{assum:uni}, an upper and lower bound on the probability of transition from state $Q_{j}$ to state $Q_{\ell}$, $Q_{j}, Q_{\ell} \in P$, under input $u = [u_{1}, u_{2} , \ldots , u_{n}] \in U$, are given by
\begin{align}
\label{eq:15}\widehat{T}_{Q_{j} \xrightarrow{u} Q_{\ell}} & = \prod_{i = 1}^{n} \int_{a^{\ell}_{i}}^{b^{\ell}_{i}} f_{w_{i}}(x_{i}-s_{i,max}^{j \rightarrow \ell})\; dx_{i}, \\
& = \prod_{i = 1}^{n} \bigg(F_{w_{i}}(b^{\ell}_{i}-s_{i,max}^{j \rightarrow \ell}) - F_{w_{i}}(a^{\ell}_{i}-s_{i,max}^{j \rightarrow \ell})\bigg),
\end{align}
\begin{align}
\label{eq:16}\widecheck{T}_{Q_{j} \xrightarrow{u} Q_{\ell}} & = \prod_{i = 1}^{n} \int_{a^{\ell}_{i}}^{b^{\ell}_{i}} f_{w_{i}}(x_{i}-s_{i,min}^{j\rightarrow \ell})\; dx_{i} \\ 
& = \prod_{i = 1}^{n} \bigg(F_{w_{i}}(b^{\ell}_{i}-s_{i,min}^{j \rightarrow \ell}) - F_{w_{i}}(a^{\ell}_{i}-s_{i,min}^{j \rightarrow \ell})\bigg)
\end{align}
where $F_{w_i}$ is the cumulative distribution function for $w_i$ and
\begin{align}
\label{eq:13}
  s_{i,max}^{j \rightarrow \ell}&=\begin{cases}
  s^\ell_{i,max},&\text{if} \;\;   s^\ell_{i,max} \in [\widecheck{r}_{i}^j + u_{i}, \widehat{r}_{i}^j + u_{i}] \\
\widehat{r}_{i}^j + u_{i} , &\text{if} \;\;  s^\ell_{i,max} > \widehat{r}_{i}^j + u_{i} \\
\widecheck{r}_{i}^j + u_{i},& \text{if} \;\;  s^\ell_{i,max} < \widecheck{r}_{i}^j + u_{i},
\end{cases}\\
\label{eq:14}  s_{i,min}^{j \rightarrow \ell}&=\begin{cases}
\widecheck{r}_{i}^j + u_{i} ,&\text{if} \;\; s_{i, max}^{j \rightarrow \ell} > \frac{\widecheck{r}_{i}^j+\widehat{r}_{i}^j}{2} + u_{i} \\
\widehat{r}_{i}^j + u_{i}, &\text{otherwise} \;\; ,
\end{cases}
\end{align}
with $s^\ell_{i,max} = \frac{a^{\ell}_{i} + b^{\ell}_{i}}{2} - c_{i}$.
\end{fact}\mbox{}\\

\begin{figure}
\centering
\includegraphics[scale=0.6]{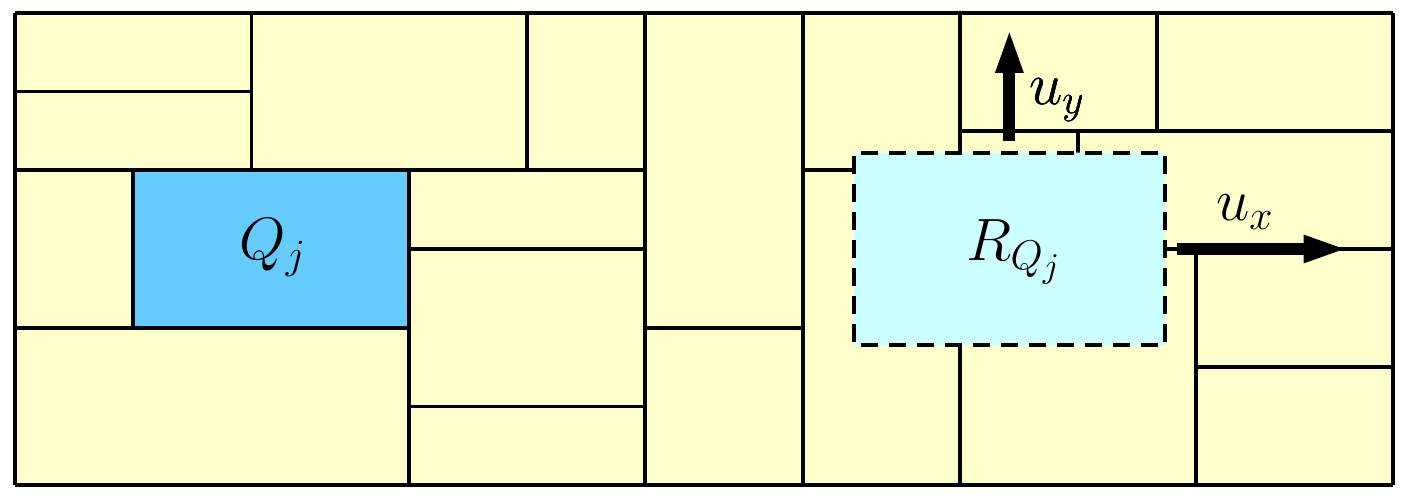}
\caption{2D depiction of the synthesis problem for system \eqref{eq3}. Every state $Q_{j}$ has a reachable set $R_{Q_j}$ under $\mathcal{F}$ which is shifted when an input $u$ is applied. The permanent component construction problem requires positioning $R_{Q_j}$ such that all instances of noise inside $R_{Q_j}$ ensures the satisfiability of the specification. If no input can achieve this, the lower bound reachability maximization problem amounts to finding a position for $R_{Q_j}$ such that the probability of reaching a permanent component is maximized in the worst instance of noise inside $R_{Q_j}$.}
\label{FigSynthProb}
\end{figure}

According to Remark \ref{RemReach}, given a CIMC abstraction $\mathcal{C}$ of  \eqref{eq3}, for every state $\left< Q_{j}, s_{i} \right>$ of the product CIMC $\mathcal{C} \otimes \mathcal{A}$ \textcolor{black}{(or $\mathcal{C} \otimes \overline{\mathcal{A}}$ when the objective is to minimize the probability of satisfying $\Psi$}), the goal is to shift the reachable set $R_{Q_j}$ of $Q_{j}$ via the application of an input $u$ so as to maximize the lower bound probability of reaching a permanent winning component from $\left< Q_{j}, s_{i} \right>$, as illustrated in Figure \ref{FigSynthProb}.

As in the finite-mode case, this is achieved by first solving a qualitative problem, which we call \textit{component construction problem}, where \textcolor{black}{the greatest permanent winning component} of $\mathcal{C} \otimes \mathcal{A}$ \textcolor{black}{($\mathcal{C} \otimes \overline{\mathcal{A}}$ for minimization) is} created; then, a quantitative problem is solved where an input maximizing the lower bound probability of reaching these components is computed for all states of $\mathcal{C} \otimes \mathcal{A}$.

In the following sections, we first provide a solution to Subproblem 2.1 and show that, although the input space $U$ of a CIMC $\mathcal{C}$ is uncountably infinite, the qualitative problem can be converted to a finite-mode component search by carefully selecting a finite number of inputs of $U$, which are identified geometrically under the stated assumptions. Subsequently, we derive an optimization problem for solving the quantitative problem and obtain the desired policies for the CIMC abstraction $\mathcal{C}$ of the system. Finally, the refinement of the partition $P$, from which the CIMC abstraction $\mathcal{C}$ arises, is addressed so as to reach a set level of optimality for the control policies with respect to the abstracted system.

\subsection{COMPONENTS CONSTRUCTION}
\label{SubsecCompCons}

In this subsection, we discuss the problem of generating the greatest permanent component $(WC)_{P}^{G}$ in a product CIMC $\mathcal{C} \otimes \mathcal{A}$ when $\mathcal{C}$ abstracts \eqref{eq3} under Assumptions \ref{AssumpRec} to \ref{assum:uni}, that is, the transition bounds between the states of $\mathcal{C}$ are given as in Fact \ref{FactBounds}.

First, we remark that if all density functions $f_{w_i}$ of the disturbance vector $w[k]$ have infinite support, the probability of making a transition between any two states of $\mathcal{C}$ has a non-zero lower bound for all choices of input. In this case, the IMC abstraction induced by some policy of $\mathcal{C}$ always induces MCs where all possible transitions have a non-zero probability, greatly simplifying the component construction problem. Here, we remove this restriction and alternatively assume that each $w_{i}$ has a probability density function living on a \textcolor{black}{bounded} interval support.\\

\begin{assumption}
\label{AssumpFinDis}
All probability density functions $f_{w_i}$ of the disturbance vector $w[k]= \begin{bmatrix}w_{1}[k]& w_{2}[k]& \ldots& w_{n}[k]\end{bmatrix}^{T}$ of system \eqref{eq3} have a \textcolor{black}{bounded interval} support, that is $W_{i} = [\widecheck{w}_{i},  \widehat{w}_{i}] \subset R$ and  $f_{w_i}(x_{i}) = 0 \;\;  \forall x_{i} \not \in W_{i}$.
\end{assumption}\mbox{}

Recall that, in an IMC, a transition between two states $Q_{j}$ and $Q_{i}$ can be classified into three different categories:
\begin{itemize}
\item An ``off'' transition if $\widehat{T}(Q_{j}, Q_{i}) = 0$,
\item An ``on'' transition if $\widecheck{T}(Q_{j}, Q_{i}) > 0$,
\item A transition which could be either ``on'' or ``off'' depending on the assumed transition values if $\widecheck{T}(Q_{j}, Q_{i}) = 0$ and $\widehat{T}(Q_{j}, Q_{i}) > 0$.\\
\end{itemize}

\noindent The connectivity properties of an IMC $\mathcal{I}$ dictate which states belong to a permanent winning component or a largest winning component in the product between $\mathcal{I}$ and an automaton $\mathcal{A}$. Provided that the partition $P$ of the system's domain is finite, the number of possible connectivity structures of an IMC abstraction arising from this partition is finite as well. Therefore, in the case of a CIMC abstraction, the objective is to find all connectivity structures which are achievable with the set of inputs $U$, choose an input $u \in U$ for all such structures and for all states $Q_{j}$ of $\mathcal{C}$, and feed the resulting finite-input BMDP $\mathcal{B}$ into the component search algorithms introduced in Section \ref{SecSynthFin} in order to compute the \textcolor{black}{permanent winning component} of the product CIMC $\mathcal{C} \otimes \mathcal{A}$, where $\mathcal{C}$ is the CIMC abstraction of \eqref{eq3} with domain partition $P$. The same procedure can be applied to find the greatest winning $(WC)_{L}^{G}$ of $\mathcal{C} \otimes \mathcal{A}$. \\

\begin{fact}
The problem of computing the greatest permanent winning component $(WC)_{P}^{G}$ as well as the greatest winning component $(WC)_{L}^{G}$ of a product CIMC $\mathcal{C} \otimes \mathcal{A}$ can be converted to a component search in a product BMDP.
\end{fact}\mbox{}

Finding the appropriate actions for state $Q_{j}$ is done by partitioning the input space $U$ into regions such that the resulting IMCs upon application of an input in different regions are qualitatively different, as illustrated in Figure \ref{CompConstFig}. We achieve this by first finding the subsets of $U$ where, for each state $Q_{i}$ reachable by $Q_{j}$ under some input, the transition from $Q_{j}$ to $Q_{i}$ behaves differently (``on'', ``off'' or either), formalized below as \textit{trigger regions}.\\

\begin{definition}[Trigger Region]
For any states $Q_{j}$ and $Q_{i}$ of $P$, the \textit{trigger regions} of $Q_{j}$ with respect to $Q_{i}$ are subsets of the input space $U$ defined as follows:
\begin{itemize}
\item The ``off'' trigger region $U^{f}_{Q_{j}}(Q_{i}) \subseteq U$ is the set of inputs such that $\widehat{T}(Q_{j}, u, Q_{i}) = 0$, $\forall u \in U^{f}_{Q_{j}}(Q_{i})$,
\item The ``on'' trigger region $U^{o}_{Q_{j}}(Q_{i}) \subseteq U$ is the set of inputs such that \mbox{} $\widecheck{T}(Q_{j}, u, Q_{i}) > 0$, $\forall u \in U^{o}_{Q_{j}}(Q_{i})$,
\item The ``undecided'' trigger region $U^{\textcolor{black}{n}}_{Q_{j}}(Q_{i}) \subseteq U$ is the set of inputs such that $\widecheck{T}(Q_{j}, u, Q_{i}) = 0$ and $\widehat{T}(Q_{j}, u, Q_{i}) > 0$, $\forall u \in U^{\textcolor{black}{n}}_{Q_{j}}(Q_{i})$.
\end{itemize}
\end{definition}\mbox{}

\begin{figure}
\centering
\includegraphics[scale=0.7]{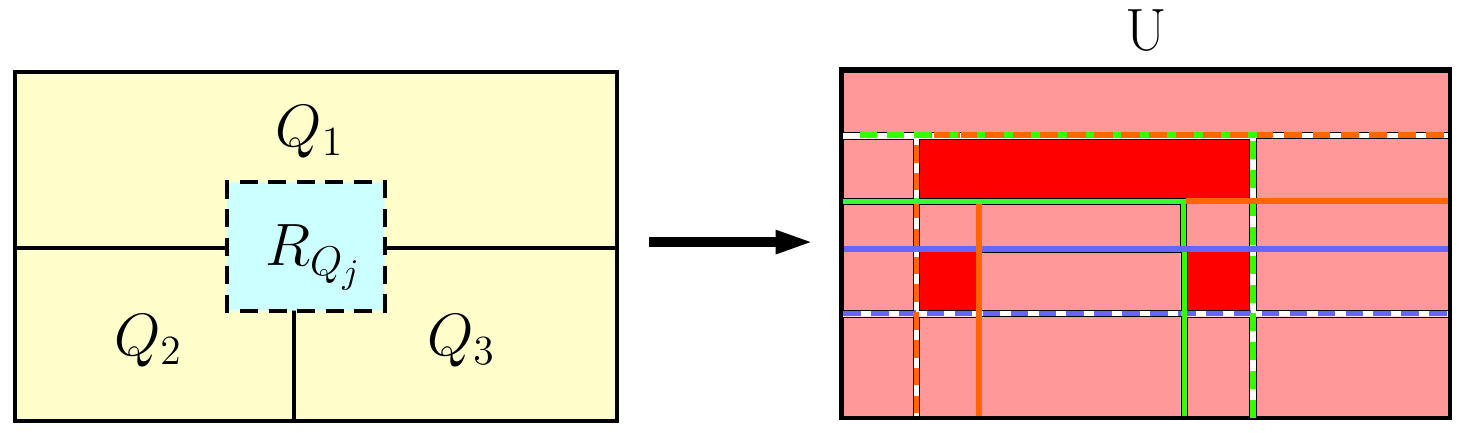}
\caption{Sketch example of the component construction problem. The reachable set $R_{Q_{j}}$ of state $Q_{j}$ induces a partition of the input space $U$ where each region produces a qualitatively different set of transitions. Dashed lines separate regions of $U$ where the transition to some state is turned ``on'' or ``off'', solid lines separate regions where the lower bound probability of transition to some state is zero and non-zero. Blue lines correspond to state $Q_{1}$, green to $Q_{2}$ and orange to $Q_{3}$. Dark red regions highlight inputs causing several transitions to have a zero lower bound and a non-zero upper bound; such regions may need to be further partitioned.}
\label{CompConstFig}
\end{figure}

\noindent Note that some of these triggers regions may evaluate to the empty set for some choices of partition $P$. In addition, the union of all trigger regions of state $Q_{j}$ with respect to state $Q_{i}$ is equal to the input space $U$. For system \eqref{eq3} with Assumptions \ref{AssumpRec} to \ref{AssumpFinDis}, these trigger regions for state $Q_{j}$ are geometrically identifiable due to the structure of both the disturbance and the over-approximation of the one-step reachable state of $Q_{j}$ highlighted in Remark \ref{RemReach}. The ``off" trigger region corresponds to shifted reachable sets of $Q_{j}$ where disturbance $w$ cannot reach $Q_{i}$, the ``on" trigger region corresponds to shifted reachable sets where any position of the disturbance results in an overlap with $Q_{i}$, and the ``undecided" trigger region corresponds to shifted reachable sets where some positions of the disturbance cause an overlap with $Q_{i}$ and some do not.\\

\begin{proposition}
\label{probtrigg}
The trigger regions of state $Q_{j} \in P$ with respect to state $Q_{i} \in P$ and input space $U$ under dynamics \eqref{eq3} with partition $P$ and satisfying Assumptions \ref{AssumpRec} to \ref{AssumpFinDis} are given by
\begin{align}
U^{f}_{Q_{j}}(Q_{i}) & = \{u \in \mathbb{R}^{n} : \exists k \;\; \widehat{r}^{j}_{k} + u_{k} + \widehat{w}_{k} \leq a_{k}^{i} \;  \\ & or \;  \widecheck{r}^{j}_{k} + u_{k} + \widecheck{w}_{k} \geq b_{k}^{i}  \} \cap \; U  , \\
U^{o}_{Q_{j}}(Q_{i}) & = \Big\{u \in \mathbb{R}^{n} : \forall k \; \; \Big( \frac{\widehat{r}^{j}_{k} + \widecheck{r}^{j}_{k}}{2} + u_{k} \geq \frac{a^{i}_{k} + b^{i}_{k}}{2} - c_{i} \\ & and \; \widehat{r}^{j}_{k} + u_{k} + \widecheck{w}_{k} \leq b^{i}_{k} \Big) \; or \; \Big(\frac{\widehat{r}^{j}_{k} + \widecheck{r}^{j}_{k}}{2} + u_{k} \leq \frac{a^{i}_{k} + b^{i}_{k}}{2} - c_{i} \\ & and \; \widecheck{r}^{j}_{k} + u_{k} + \widehat{w}_{k} \geq a^{i}_{k} \Big)   \Big\}  \cap \; U \; , \\
U^{\textcolor{black}{n}}_{Q_{j}}(Q_{i}) & = \left(\mathbb{R}^{n} \setminus (\; U^{o}_{Q_{j}}(Q_{i}) \cup U^{f}_{Q_{j}}(Q_{i}) \; ) \right) \cap \; U \; .
\end{align}
\end{proposition}\mbox{}

\noindent It follows that different overlaps of the trigger regions of state $Q_{j}$ induce qualitatively different profiles for the outgoing transitions of $Q_{j}$.\\

\begin{definition}[Trigger Regions Overlap]
\label{defoverlap}
A \textit{Trigger Regions Overlap} $\textcolor{black}{\mathcal{H}_{Q_{j}}} \subseteq U$ of state $Q_{j} \in P$ is a subset of the input space $U$ such that
\begin{align*}
\textcolor{black}{\mathcal{H}_{Q_{j}}(t_{1}, t_{2}, \ldots, t_{|P|}) }= \bigcap_{i \in \{1, 2, \ldots ,  |P|\}} U^{t_i}_{Q_{j}}(Q_{i}) \;\; ,
\end{align*}
\noindent where $t_{i} \in \{f, o, \textcolor{black}{n} \}, \; \forall i$.
\end{definition}\mbox{}

It should be noticed that an overlap of two or more undecided trigger regions could produce qualitatively different transitions for several subset of its inputs and have to be further examined, as demonstrated in the following example depicted in Figure \ref{FigExample3}. \\

\begin{figure}
\centering
\includegraphics[scale=0.55]{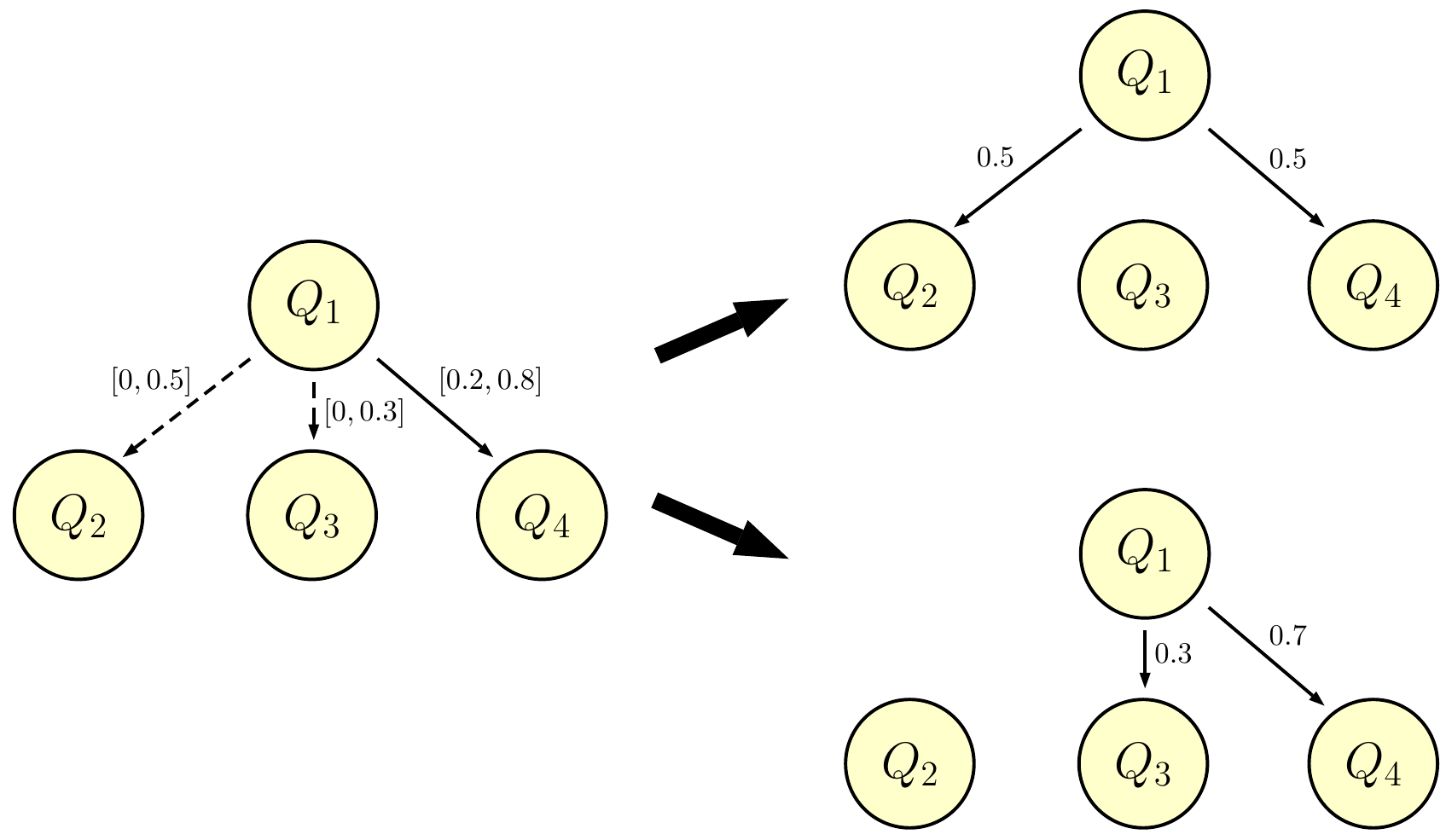}
\includegraphics[scale=0.55]{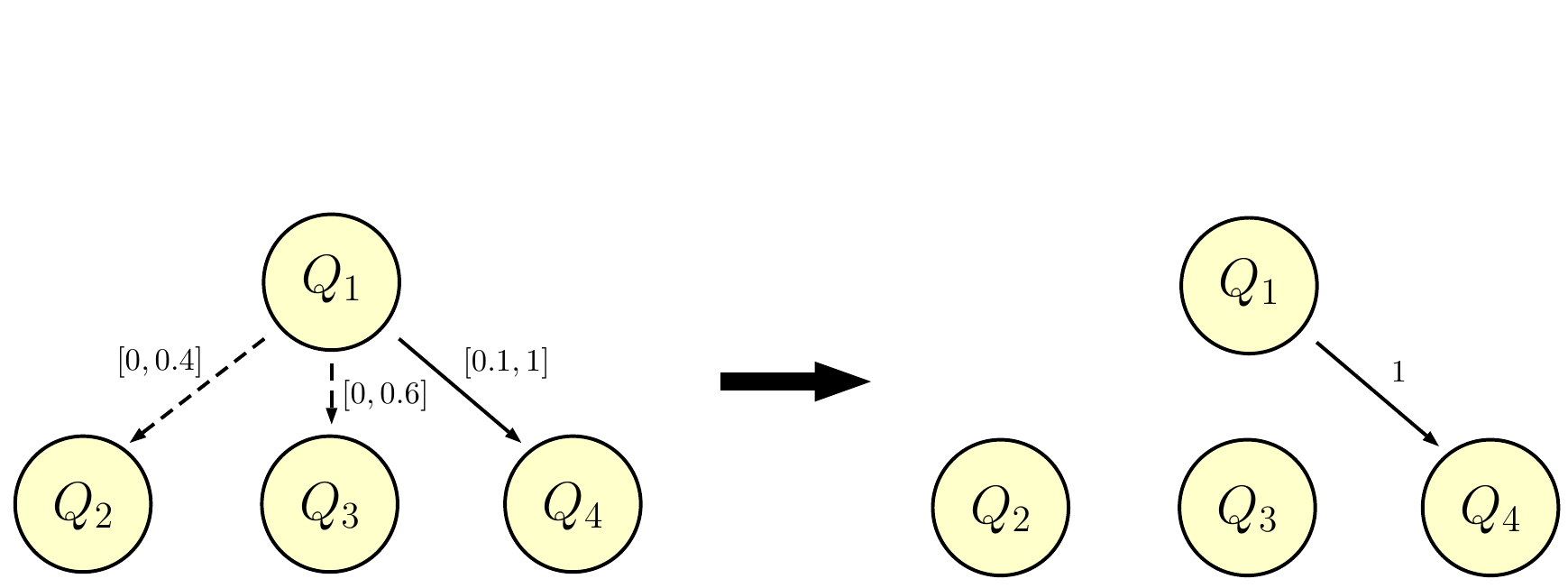}
\caption{\textcolor{black}{Two IMC transition profiles with similar transition types but different qualitative structures as discussed in Example 3. The transitions from $Q_{1}$ to the three other states are of the same type in both cases; however, while the transitions from $Q_{1}$ to $Q_{2}$ and $Q_{3}$ cannot be set to zero simultaneously for any adversary in the top example, this can be achieved in the bottom example.}}
\label{FigExample3}
\end{figure}

\begin{example}
\label{ExTranProf}
Consider the following two transition profiles from state $Q_{1}$ to three states $Q_{2}$, $Q_{3}$ and $Q_{4}$:
\begin{itemize}
\item $T(Q_1, Q_2) = [0, 0.5]$,  $T(Q_1, Q_3) = [0, 0.3]$ and $T(Q_1, Q_4) = [0.2, 0.8]$,
\item $T(Q_1, Q_2) = [0, 0.4]$,  $T(Q_1, Q_3) = [0, 0.6]$ and $T(Q_1, Q_4) = [0.1, 1]$.
\end{itemize}
\textcolor{black}{Although $T(Q_1, Q_2)$ and $T(Q_1, Q_3)$ are undecided in both cases and $T(Q_1, Q_4)$ is ``on" in both cases}, the two profiles are qualitatively different. In the first case, no probability assignment can simultaneously turn off the transitions from $Q_{1}$ to $Q_{2}$ and from $Q_{1}$ to $Q_{3}$; however, in the second case, it is possible to turn off these two transitions at the same time by assigning a probability of 1 to the transition from $Q_{1}$ to $Q_{4}$. 
\end{example}\mbox{}

For all states $Q_{j} \in P$, we denote the set of overlaps with 2 or more undecided trigger regions by $\mathcal{H}^{n}_{Q_j}$, and all other overlaps by $\mathcal{H}^{S}_{Q_j}$. 

In summary, we remark that the components construction problem in a product CIMC $\mathcal{C} \otimes \mathcal{A}$ is solved by converting it to a component search in a finite-action product BMDP $\mathcal{B} \otimes \mathcal{A}$. The construction of $\mathcal{B}$ is achieved by partitioning the input space of all states $Q_{j}$ of $\mathcal{C}$ into trigger region overlaps yielding qualitatively different transition profiles, and by choosing one control action per overlap in $\mathcal{H}^{S}_{Q_j}$, and possibly more than one control actions per overlap in $\mathcal{H}^{n}_{Q_j}$. Indeed, we observed in Example \ref{ExTranProf} that, for every overlap in the set $\mathcal{H}^{n}_{Q_j}$ of a state $Q_{j}$, we have to distinguish the sets of inputs allowing for different combinations of inactive uncertain transitions. We show that the overlaps are geometrically identified for system \eqref{eq3} under Assumption \ref{AssumpRec} to \ref{AssumpFinDis}. 

The input selection procedure is detailed in Algorithm \ref{InpSelecAlg}. This algorithm chooses the minimum energy input in all overlaps in $\mathcal{H}^{S}_{Q_j}$ and performs a search over from the overlaps in $\mathcal{H}^{n}_{Q_j}$ in order to find control inputs allowing for different combinations of inactive uncertain transitions. We emphasize that the optimization problem on line 20 is non-convex under our system assumptions and is in general hard to solve. Note that Algorithm \ref{InpSelecAlg} in its current state may select more actions than needed from the overlaps in $\mathcal{H}^{n}_{Q_j}$. This is due to the fact that our procedure is likely to choose different actions for two distinct combinations of achievable ``off'' uncertain transitions $S$ and $S'$, where none of these combinations is a strict subset of the other, while a single action may be able to accommodate these two combinations at once. A consequence is that the resulting BMDP $\mathcal{B}$ may have a larger action space than necessary. This could be addressed by considering multiple such combinations at once in the constraints on line 20, at the cost of having to potentially solve a greater number of optimization problems.

\begin{algorithm}[H]
\caption{Input Selection for State $Q_{j}$}
\begin{algorithmic}[1]
\STATE \textbf{Input}: Sets of overlaps $\mathcal{H}^{S}_{Q_{j}}$ and $\mathcal{H}^{n}_{Q_{j}}$ of state $Q_{j}$
\STATE \textbf{Output}: Finite set of actions $A(Q_{j})$
\STATE \textbf{Initialize}: $A(Q_{j}) := \emptyset$
\FOR {$\mathcal{H}_{i} \in \mathcal{H}^{S}_{Q_j}$}
\STATE $u^{*} := \textcolor{black}{\argmin_{u \in \mathcal{H}_{i}}} ||u||_{2}^{2}$
\STATE $A(Q_{j}) \leftarrow u^{*}$
\ENDFOR

\FOR {$\mathcal{H}_{i} \in \mathcal{H}^{n}_{Q_j}$}
\STATE $L := \emptyset$, $O := \emptyset$, $Y := \emptyset$
\STATE For all states $Q_{k}$ such that $U^{o}_{Q_{k}} \cap \mathcal{H}_{i} \not = \emptyset $, $O \leftarrow Q_{k}$
\STATE For all states $Q_{k}$ such that $U^{n}_{Q_{k}} \cap \mathcal{H}_{i} \not = \emptyset $, $Y \leftarrow Q_{k}$
\STATE $L \leftarrow Y$

\FOR{$S \in L$}
\FOR{$u \in A(Q_{j})$}
\STATE Check if $\sum_{q \in O} \widehat{T}(Q_{j}, u , q) + \sum_{q \in Y \setminus S} \widehat{T}(Q_{j}, u , q) \geq 1$
\ENDFOR
\IF{Feasible for some $u \in  A(Q_{\textcolor{black}{j}})$}
\STATE Continue for-loop (Line 13)
\ENDIF
\STATE Solve $u^{*} = \textcolor{black}{\argmin_{u \in \mathcal{H}_{i}}} ||u||_{2}^{2}$ such that $\sum_{q \in O} \widehat{T}(Q_{j}, u , q) + \sum_{q \in Y \setminus S} \widehat{T}(Q_{j}, u , q) \geq 1$
\IF{Feasible}
\STATE $A(Q_{j}) \leftarrow u^{*}$ 
\ELSE
\STATE Add the ${|S| \choose  |S| -1}$ combinations of $|S| -1$ states of $S$ (which are not already in $L$ and for which no superset of states previously returned a feasible solution) to $L$
\ENDIF
\ENDFOR
\ENDFOR
\RETURN $A(Q_{j})$
\end{algorithmic}
\label{InpSelecAlg}
\end{algorithm}

\begin{algorithm}[t!]
\caption{Component Construction Method for \eqref{eq3}}
\begin{algorithmic}[2]
\STATE \textbf{Input}: Domain Partition $P$, input Space $U$, DRA $\mathcal{A}$ of specification $\Psi$
\STATE \textbf{Output}: \textcolor{black}{Winning components} $(WC)_{P}^{G}$ and $(WC)_{L}^{G}$ of product CIMC $\mathcal{C} \otimes \mathcal{A}$ constructed from $P$
\STATE Create a BMDP $\mathcal{B}$ with the same states as $P$ and with each action set $A(Q_{j})$ initialized to the empty set
\STATE Compute the overlap sets for all $Q_{j} \in P$ using Proposition \ref{probtrigg} and according to Definition \ref{defoverlap}
\FOR {$Q_{j} \in P$}
\STATE Compute the set of actions $A(Q_{j})$ using Algorithm \ref{InpSelecAlg} as well as their corresponding transition profiles
\ENDFOR
\RETURN $(WC)_{P}^{G}$ and $(WC)_{L}^{G}$ and their corresponding control actions by applying the component search in Algorithm \ref{AlgPerBSCCW} and \ref{AlgPerWin} to $\mathcal{B} \otimes \mathcal{A}$
\end{algorithmic}
\label{CompConsAlg}
\end{algorithm}

Algorithm \ref{CompConsAlg} summarizes the component construction procedure and outputs the greatest permanent winning component $(WC)_{P}^{G}$ of a product CIMC $\mathcal{C} \otimes \mathcal{A}$, as well as its greatest winning component $(WC)_{L}^{G}$, where $\mathcal{C}$ serves as a CIMC abstraction of system \eqref{eq3}.

\subsection{REACHABILITY MAXIMIZATION}
\label{ReachMaxCont}

To devise an optimal control policy for system \eqref{eq3} abstracted by a CIMC $\mathcal{C}$, we now have to find the control inputs in the continuous set $U$ maximizing the lower bound probability of reaching $(WC)_{P}^{G}$ in a product CIMC  according to Theorem \ref{TheoWorstCasePol}. 

Our approach is inspired from the lower bound reachability maximization algorithm for BMDPs in \cite{lahijanian2015formal}. In this algorithm, the procedure for computing a control policy maximizing the lower bound probability of reaching a target set of states $G$ in a finite-action BMDP is based on value iteration and is as follows:
\begin{enumerate}
\item Initialize a probability vector $W^{0} = [p^{0}_{1}, p^{0}_{2}, \ldots, p^{0}_{m}]$ where $p^{0}_{i} = 1$ if $p_{i} \in G$ and $0$ otherwise.
\item At each time step $k$, construct an ascending ordering $\mathcal{O}_{k} = q_{1} q_{2} \ldots q_{m}$, $q_{i} \in Q$, of the states such that $p^{k}_{1} \leq p^{k}_{2} \leq \ldots \leq p^{k}_{m}$.
\item For each state $Q_{j}$ and for each action in $A(Q_{j})$, allocate as much probability mass $z^{j}_{1}$ as possible to state $q_{1}$, then allocate as much probability mass $z^{j}_{2}$ as possible to state $q_{2}$ with the amount of probability left, etc., in order to construct the worst possible assignment of the probabilities allowed by the IMC under each action with respect to the objective of reaching $G$.
\item For each state, pick the action from $A(Q_{j})$ that yields the highest worst-case probability $p^{k+1}_{i} = \sum_{j = 1}^{m} p^{k}_{j} z^{i}_{j}$ of reaching $G$. 
\item Update the probability vector $W^{k+1}$ such that $p^{k+1}_{i} = \sum_{j = 1}^{m} p^{k}_{j} z^{i}_{j}$, with $p^{k+1}_{i}$ being the computed probability under the chosen action at state $Q_{i}$, and construct a new ordering $\mathcal{O}^{k+1}$. Repeat this process until vector $W$ converges \cite{wu2008reachability} and the last selected actions are the lower bound reachability maximizing actions for all states.
\end{enumerate}

We propose to follow the same procedure for computing lower bound maximizing policies in the product CIMC $\mathcal{C} \otimes \mathcal{A}$. However, while finite-mode systems rely on exhaustive search over every possible action to choose the most optimal one at each step $k$ of the above algorithm, systems with a continuous set of inputs $U$ require solving an optimization problem at Step 3 of the above algorithm to find the reachability maximizing input $u$ for all states $\left< Q_{j}, s_{i} \right>$ of the product CIMC $\mathcal{C} \otimes \mathcal{A}$.

We first note that the transition bound functions in $\mathcal{C} \otimes \mathcal{A}$ are determined by the transition bound functions in $\mathcal{C}$, as seen in the definition of a product CIMC. We formulate an optimization problem that outputs the best action $u \in U$ for state $\left< Q_{j}, s_{i} \right>$ at some time step $k$ of the aforementioned algorithm. Consider the set of states $\{q_{\ell}\}_{\ell = 1}^{m}$ which are reachable by $\left< Q_{j}, s_{i} \right>$ under some input, that is $\exists u \in U$ such that $\widehat{T}(\left< Q_{j}, s_{i} \right>, u, q_{\ell}) > 0, \; i = 1, 2, \ldots , m$. We denote the probability of reaching the desired component from state $q_{\ell}$ at the current time step of the algorithm by $p_{\ell}$. Consider an ascending ordering $\mathcal{O} = q_{1} q_{2} q_{3} \ldots q_{m}$ of the states reachable by $\left< Q_{j}, s_{i} \right>$ such that $p_{1} \leq p_{2} \leq \ldots \leq p_{m}$. Step 3 and 4 of the reachability maximization algorithm for the continuous input case are formulated as the optimization program

\begin{align}
\label{optimiz}
& \underset{u \in U}{\text{max}} \;\; \sum_{\ell = 1}^{m} p_{\ell} z_{\ell} \\
& s.t. \; \; \;  z_{\ell} = \text{min} \Bigg\{ \widehat{T} \big( \left< Q_{j}, s_{i} \right>, u, q_{\ell} \big), \;\; 1 - \sum_{k = 1}^{\ell - 1} z_{k} - \sum_{k = \ell+1}^{m} \widecheck{T} \big(\left< Q_{j}, s_{i} \right>, u, q_{k} \big) \Bigg\}, \\
& \qquad \qquad \qquad \qquad  \qquad \qquad \qquad \qquad  \qquad \qquad \qquad \qquad  \ell = 1, 2, 3, \ldots, m \ ,
\end{align}

\noindent where the lower and upper bound terms are given by \eqref{eq:13} and \eqref{eq:14} for the specific case of system \eqref{eq3} under Assumption \ref{AssumpRec} to \ref{assum:uni}, rendering this problem non-convex. The constraints ensure that, for a given input $u$, each state in $\mathcal{O}$ is allocated either its upper bound probability of transition or the maximum probability mass allowed by the lower bound transition probability of the following states in $\mathcal{O}$ and the probability mass distributed to the preceding states in $\mathcal{O}$. In the case study section of this paper, we tackle optimization problem \eqref{optimiz} using numerical heuristics.

Unlike in the finite-mode case, this value iteration procedure for continuous input sets is not guaranteed to converge in a finite number of steps. Therefore, we suggest computing the maximum change in the reachability probability among all states of $\mathcal{C} \otimes \mathcal{A}$ at each step of the algorithm, and terminating the procedure once this change reaches a user-defined convergence threshold.

\subsection{STATE SPACE REFINEMENT}
\label{ContSpaceRef}

Finally, we discuss partition refinement for system \eqref{eq3} to address Subproblem 2.2. 

The \textcolor{black}{quality} of the controller designed in the CIMC abstraction $\mathcal{C}$ with respect to continuous states of \eqref{eq3} can be assessed as in Section \ref{SecSynthFin} for the finite-mode system case. In light of Subsection \ref{SecRefFin}, we need to construct a best-case and a worst-case product MC induced by the product CIMC $\mathcal{C} \otimes \mathcal{A}$ to determine the suboptimality factor of each state of $\mathcal{C} \otimes \mathcal{A}$. In particular, when devising a maximizing control policy, \textcolor{black}{a} best-case MC $(\mathcal{M}_{\otimes}^{\mathcal{A}})_{u}$ is constructed by solving an upper bound reachability maximization problem on the greatest winning component $(WC)_{L}^{G}$ of the product CIMC $\mathcal{C} \otimes \mathcal{A}$, where $\mathcal{C}$ is the CIMC abstraction of \eqref{eq3} under the current partition $P$. When devising a minimizing control policy, \textcolor{black}{a best-case}  MC \textcolor{black}{$(\mathcal{M}_{\otimes}^{\mathcal{A}})_{u}$} is constructed by solving an upper bound reachability maximization problem on the greatest winning component of the product CIMC \textcolor{black}{$\mathcal{C} \otimes \overline{\mathcal{A}}$}, where $\mathcal{C}$ is the CIMC abstraction of \eqref{eq3}. These upper bound reachability maximization problems are addressed using a similar procedure as in Subsection \ref{ReachMaxCont}, with the difference that the ordering $\mathcal{O} = q_{1} q_{2} q_{3} \ldots q_{m}$ in the optimization program \eqref{optimiz} is now descending with respect to the probability of reaching the target set $G$, that is $p_1 \geq p_2 \geq \ldots \geq p_{m}$.

Propositions \ref{PropSubop} to \ref{PropCompPer}, which discuss some properties that are passed from a partition to its refinements for the finite-mode case, are also valid in this continuous input framework. In particular, as in the finite-mode case, subsets of the input space $U$ which can be shown to be certainly suboptimal may be removed. To find such subsets, we suggest building a partition $U(\left< Q_{j}, s_{i} \right>) = \{U_{n} (\left< Q_{j}, s_{i} \right>) \}_{n = 1}^{k}$ of the input space for all states $\left< Q_{j}, s_{i} \right>$ of $\mathcal{C} \otimes \mathcal{A}$. Then, for all subsets $U_{n}$, an upper bound maximization step on $(WC)_{L}^{G}$ is conducted; subsets yielding an upper bound on the maximum upper bound probability of reaching an accepting BSCC from $\left< Q_{j}, s_{i} \right>$  which is lower than the lower bound produced by $(\widehat{\mu}^{low}_{\Psi})_{\otimes}( \left< Q_{j}, s_{i} \right> )$ \textcolor{black}{(respectively, by $(\widecheck{\mu}^{up}_{\Psi})_{\otimes}( \left< Q_{j}, s_{i} \right>)$ for the case of minimization)} are suboptimal with respect to the entire input set of $\left< Q_{j}, s_{i} \right>$ and are removed from $U(\left< Q_{j}, s_{i} \right>)$, as depicted in Figure \ref{fig_update}. \textcolor{black}{Note that a finer discretization of the input space $U(\left< Q_{j}, s_{i} \right>)$ for the update step may result in the removal of a greater volume of suboptimal inputs from $U(\left< Q_{j}, s_{i} \right>)$ at each iteration of the synthesis procedure, allowing to  ``zoom in" on better inputs for state $\left< Q_{j}, s_{i} \right>$ in fewer iterations at the expense of having to solve a larger number of optimization problems per iteration.} 

\begin{figure}[H]
\centering
\includegraphics[scale=0.8]{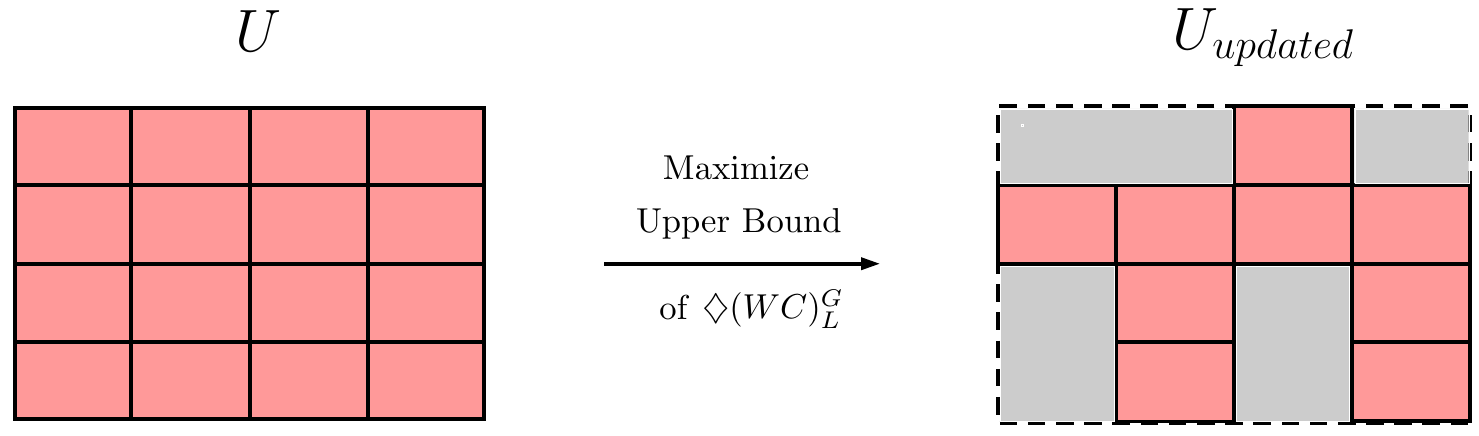}
\caption{Sketch of an input space update before refinement of the domain partition. The original input space $U$ of the considered state is gridded and the upper bound probability of reaching $(WC)_{L}^{G}$ is maximized for all subsets of the grid. The subsets producing suboptimal bounds are shown in gray and are discarded.}
\label{fig_update}
\end{figure}

Finally, once $(\mathcal{M}_{\otimes}^{\mathcal{A}})_{u}$ and $(\mathcal{M}_{\otimes}^{\mathcal{A}})_{l}$ are generated and all input sets are updated, the scoring and refinement procedure are performed identical to the finite-mode case. \textcolor{black}{After refinement, the trigger regions and overlaps of a state $Q_{j}$ calculated in Algorithm \ref{CompConsAlg} have to be re-computed only if there exists a state $\left<Q_{j}, s_{i} \right>$ for some $i$ which belonged to the difference between the greatest winning component and the greatest permanent winning component in the previous abstraction, as only such a state could potentially be a member of a new permanent winning set of states in the refined abstraction as a consequence of Proposition \ref{PropLargest}, and if either $Q_{j}$ has been refined into children states in which case the trigger regions of the children states have to be determined, or a state that was reachable from $Q_{j}$ under some action in the input space $U(Q_{j})$ has been refined into children states with respect to which the trigger regions have to be evaluated.}

The controller synthesis algorithm for continuous input systems is summarized in Algorithm \ref{ContInputSyn}. \textcolor{black}{The run-time complexity of most sub-algorithms of Algorithm \ref{ContInputSyn} has already been presented in Section \ref{SecSynthFin}. Additionally, as previously discussed, the input selection in Algorithm \ref{InpSelecAlg} grows combinatorially in $|Q|$ and the computation of overlaps in Algorithm \ref{CompConsAlg} runs exponentially in $|Q|$. Lastly, the computational complexity of the algorithm strongly depends on the optimization method used for solving the reachability maximization step.}

 \textcolor{black}{Now that the theoretical foundations of our approach have been thoroughly delineated, our synthesis techniques for both finite mode systems and continuous input systems are put into practical use in the following case study section}.

\begin{algorithm}[H]
\caption{Controller Synthesis for Continuous Input Systems} 
\begin{algorithmic}[1]
\STATE \textbf{Input}: Partition $P_{0}$ of domain $D$ of \eqref{eq1}, $\omega$-regular property $\Psi$ \textcolor{black}{(complement property $\overline{\Psi}$)} and corresponding DRA $\mathcal{{A}}$ \textcolor{black}{($\overline{\mathcal{{A}}}$)}, target controller precision $\epsilon_{thr}$ \hspace{-0.6em}
\STATE \textbf{Output}: Maximizing (minimizing) switching policy $\widehat{\mu}_{\Psi}^{low}$ ($\widecheck{\mu}_{\Psi}^{up}$), final partition $P_{fin}$
\STATE \textbf{Initialize}: $\epsilon_{max} := 1$, $i := 0$
\WHILE{$\epsilon_{max} > \epsilon_{thr}$}
\STATE Compute the sets $(WC)_{P}^{G}$ and $(WC)_{L}^{G}$ of the product CIMC $\mathcal{C} \otimes \mathcal{A}$ \textcolor{black}{($\mathcal{C} \otimes \overline{\mathcal{A}}$)} constructed from $P_{i}$ using Algorithm \ref{CompConsAlg}
\STATE Compute the policies $\widehat{\mu}_{\Psi}^{low}$ and $\widehat{\mu}_{\Psi}^{up}$ ($\widecheck{\mu}_{\Psi}^{up}$ and $\widecheck{\mu}_{\Psi}^{low}$) of the CIMC $\mathcal{C}$ according to Subsection \ref{ReachMaxCont}
\STATE Compute $\epsilon_{max}$ using \eqref{maxsub}
\IF{$\epsilon_{max} > \epsilon_{thr}$}
\STATE Compute the best-case and worst-case product MC $(\mathcal{M}_{\otimes}^{\mathcal{A}})_{u}$ and $(\mathcal{M}_{\otimes}^{\mathcal{A}})_{l}$ as discussed in Subsection \ref{ContSpaceRef}.
\STATE Construct a partition $\{U_{n} (\left< Q_{j}, s_{m} \right>) \}_{n = 1}^{k}$ of the input space $U(\left< Q_{j}, s_{m} \right>)$ of all states $\left< Q_{j}, s_{m} \right>$ of the product CIMC $\mathcal{C} \otimes \mathcal{A}$ \textcolor{black}{($\mathcal{C} \otimes \overline{\mathcal{A}}$)}
\FOR{$U_{n}(\left< Q_{j}, s_{m} \right>) \in U(\left< Q_{j}, s_{m} \right>)$}
\STATE Maximize the upper bound probability of $\Diamond (WC)_{L}^{G}$ from $\left< Q_{j}, s_{m} \right>$ with the set of inputs $U_{n}(\left< Q_{j}, s_{m} \right>)$
\ENDFOR
\STATE Apply the scoring procedure in Algorithm \ref{ScorAlgSyn} and refine all states in $P_{i}$ with a score above a user-defined threshold to produce $P_{i+1}$
\STATE Update the set of inputs of all states in the product CIMC $\mathcal{C} \otimes \mathcal{A}$ \textcolor{black}{($\mathcal{C} \otimes \overline{\mathcal{A}}$)} constructed from $P_{i+1}$ as discussed in Subsection \ref{ContSpaceRef}.
\STATE $i := i + 1$
\ENDIF
\ENDWHILE
\RETURN $\widehat{\mu}_{\Psi}^{low}$ ($\widecheck{\mu}_{\Psi}^{up}$), $P_{fin} := P_{i}$
\end{algorithmic}
\label{ContInputSyn}
\end{algorithm}

\section{CASE STUDY}

We now present a numerical example to demonstrate the synthesis procedures derived in previous sections. The code used to generate this example was written in Python 2.7 and is available at \url{https://github.com/gtfactslab/StochasticSynthesis}. All computations were conducted on the Partnership for an Advanced Computing Environment (PACE) Georgia Tech cluster \cite{PACE} which offered 120GB of memory. The examples in Section \ref{finitemodeex} were performed on a single core, while those in Section \ref{continuousinputex} were distributed over 4 cores.

We consider a stochastic model of a bistable switch with dynamics
\begin{equation}
  \begin{aligned}
x_{1}[k+1] & = x_1[k]  + ( \; -a x_{1}[k] + x_{2}[k] \; ) \cdot \Delta T + u_1+ w_1\\
x_{2}[k+1] & = x_2[k]  + \Big(\; \frac{(x_{1}[k])^{2}}{(x_{1}[k])^{2} + 1} - b x_{2}[k] \; \Big) \cdot \Delta T + u_2 + w_2 \;\; ,
 \end{aligned}
\label{eq5}
\end{equation}

\noindent where $w_1$ and $w_2$ are independent truncated Gaussian random variables sampled at each time step. $w_1\sim \mathcal{N}(\mu = -0.3 ; \sigma^2 = 0.1)$ and is truncated on $[-0.4, -0.2]$; $w_2$ is similarly defined. We will consider two sets of inputs in this case study: the continuous set $U = [-0.05, 0.05] \times [-0.05, 0.05]$ and the finite set $U_{fin} = \{[0, 0]^T, [0.05, 0]^T,  [-0.05, 0]^T, [0, 0.05]^T, [0, -0.05]^T\}$ which is a subset of $U$. The domain $D$ of \eqref{eq5} is $[0.0, 4.0] \times [0.0, 4.0]$. To keep the system self-contained in $D$, we assume that any time the disturbance would push the trajectory outside of $D$, it is actually maintained on the boundary of $D$. We choose the parameters $a = 1.3$, $b =0.25$ and $\Delta T =0.05$. Our goal is to synthesize a controller for \eqref{eq5} that maximizes the probability of satisfying the LTL specifications
\begin{align*}
\phi_1 & = \square((\neg A \wedge  \bigcirc A) \rightarrow (\bigcirc \bigcirc A \wedge \bigcirc \bigcirc \bigcirc A)) \ , \\
\phi_2 & = ( \textcolor{black}{\square \lozenge} A \rightarrow \lozenge B) \wedge (\lozenge C \rightarrow \square \neg B) \ ,
\end{align*}
\noindent where $\phi_{1}$ translates to `` always remain in an $A$ state for at least 2 more time steps when entering an $A$ state'' and $\phi_{2}$ translates to  ``reach a $B$ state if the trajectory \textcolor{black}{always eventually returns to an $A$ state}, and never reach a $B$ state if the trajectory reaches a $C$ state'' in natural language. The DRA corresponding to specification $\phi_{1}$ contains 5 states and has 1 Rabin pair, while the DRA representing $\phi_{2}$ contains 7 states and has 3 Rabin pairs. Schematic representations of these DRAs are found in Figure \ref{FigsDRA}. Initial partitions of the domain $D$ along with the labeling of the states are presented in the next subsections. First, we synthesize controllers using the finite set of inputs $U_{fin}$. Second, we devise control policies from the continuous set of inputs $U$. Finally, we compile some observations and concluding remarks in a discussion subsection.

\begin{figure}
\includegraphics[scale=0.7]{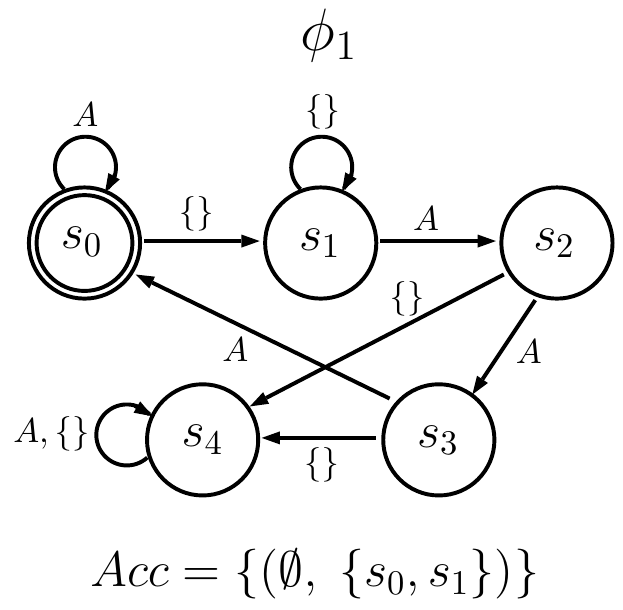}
\hspace{0.7cm}
\includegraphics[scale=0.7]{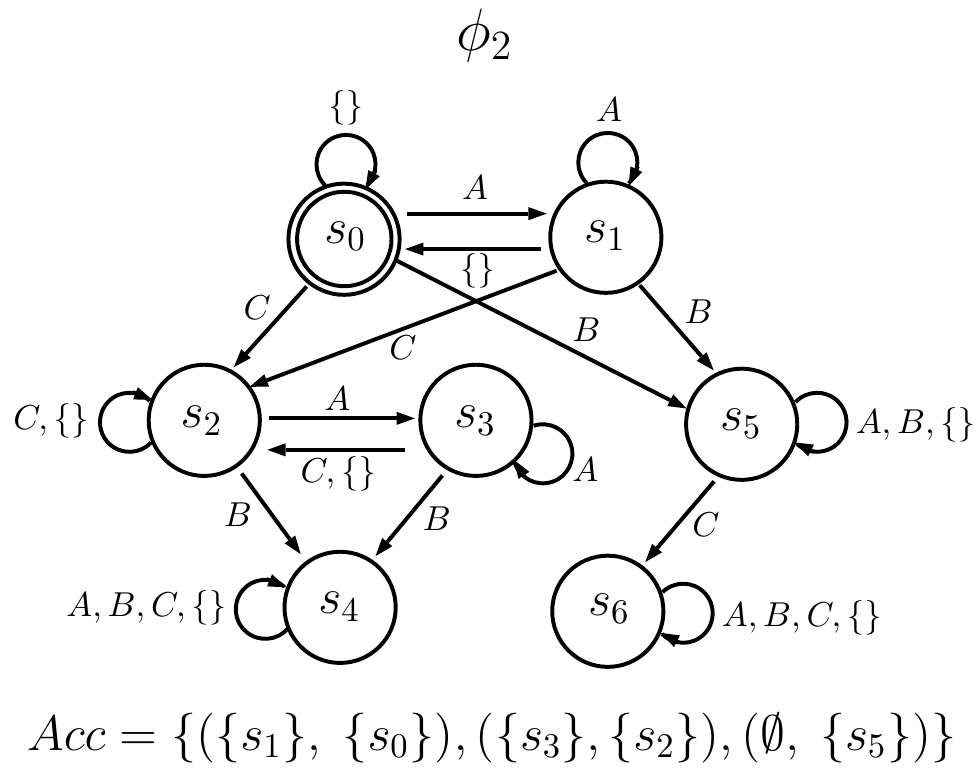}
\caption{\textcolor{black}{Possible DRAs for specification $\phi_{1}$ (Left) and specification $\phi{_2}$ (Right). Note that these DRAs assume the convention that the state initialization ``counts as a transition", i.e., when a state of the BMDP $Q_{j}$ is chosen as initial state, the product BMDP transitions from $\left< Q_{j}, s_{0} \right>$ to $\left< Q_{j}, \delta(s_{0}, L(Q_{j})) \right>$}.}
\label{FigsDRA}
\end{figure}

\subsection{FINITE-MODE SYNTHESIS}
\label{finitemodeex}

First, we synthesize a switching policy for maximizing the probability of satisfying $\phi_{1}$ and $\phi_{2}$ in \eqref{eq5} using the finite set $U_{fin}$, where each input corresponds to one mode, and applying the synthesis Algorithm \ref{AlgFiniteMode} for finite-mode systems with a target precision $\epsilon_{thr} = 0.30$. At each refinement step, states of the current partition with a refinement score that is greater than 5\% of the maximum score are chosen to be refined and split in half along their greatest dimension. The deterministic portion of the dynamics of system \eqref{eq5} are known to be monotone. Therefore, BMDP abstractions of \eqref{eq5} for rectangular partitions of $D$ are efficiently computed using the technique in \cite{dutreix2018} for each mode. The initial partition of the domain $D$ for specification $\phi_{1}$ is given in Figure \ref{FiniteStateSpaces1} (Left), and the initial partition for specification $\phi_{2}$ is in Figure \ref{FiniteStateSpaces2} (Left). At each refinement step, the states selected for refinement are split in half along their greatest dimension.

The component search algorithm is conducted at each iteration of the while loop of Algorithm \ref{AlgFiniteMode} until the set of potential accepting BSCCs \textcolor{black}{$(U)_{\textcolor{black}{pot}}^{G}$} becomes empty, in which case the component construction procedure is skipped and the lower bound maximization problem in Line 6 is performed on the latest known version of the greatest permanent winning component $(WC)_{P}^{G}$. As no new permanent accepting BSCCs can be constructed anywhere else in the state space in this scenario, an under-approximation of $(WC)_{P}^{G}$ containing all possible permanent BSCCs without all permanent sink states is sufficient for the reachability problem. Note that $(WC)_{P}^{G}$ can be updated if permanent sink states with a lower bound of 1 are constructed during the lower bound maximization step.

The controller synthesis procedure for specification $\phi_{1}$ terminated in 13 hours and 27 minutes with a greatest suboptimality factor $\epsilon_{max} = 0.2999$, and created 18418 states in 18 refinement steps, corresponding to 92090 states in the product BMDP constructed from the final partition. The final refined partition is shown in Figure \eqref{FiniteStateSpaces1} (Right). For specification $\phi_{2}$, the procedure terminated in 38 minutes with a greatest suboptimality factor $\epsilon_{max} = 0.2998$ and created 7711 states in 15 refinement steps, corresponding to 53977 states in the product BMDP constructed from the final partition. The final refined partition is shown in Figure \eqref{FiniteStateSpaces2} (Right).

The cumulative execution time against the number of refinement steps is plotted in Figure \ref{RunningTimesFinite} for specification $\phi_{1}$ (Left) and specification $\phi_{2}$ (Right). The average number of actions left at each state of the product BMDP $\mathcal{B} \otimes \mathcal{A}$ after each refinement step is displayed in Figure \ref{AvgNumAct} for specification $\phi_{1}$ (Left) and specification $\phi_{2}$ (Right). Lastly, three possible metrics of precision for the computed controller --- namely, the greatest suboptimality factor, average suboptimality factor of the product BMDP and fractions of states above the target precision $\epsilon_{thr}$ --- as a function of the number of refinement steps are shown in Figure \eqref{OptFacFinite} for specification $\phi_{1}$ (Left) and specification $\phi_{2}$ (Right).

\begin{figure}[H]
\centering
\includegraphics[scale=0.32]{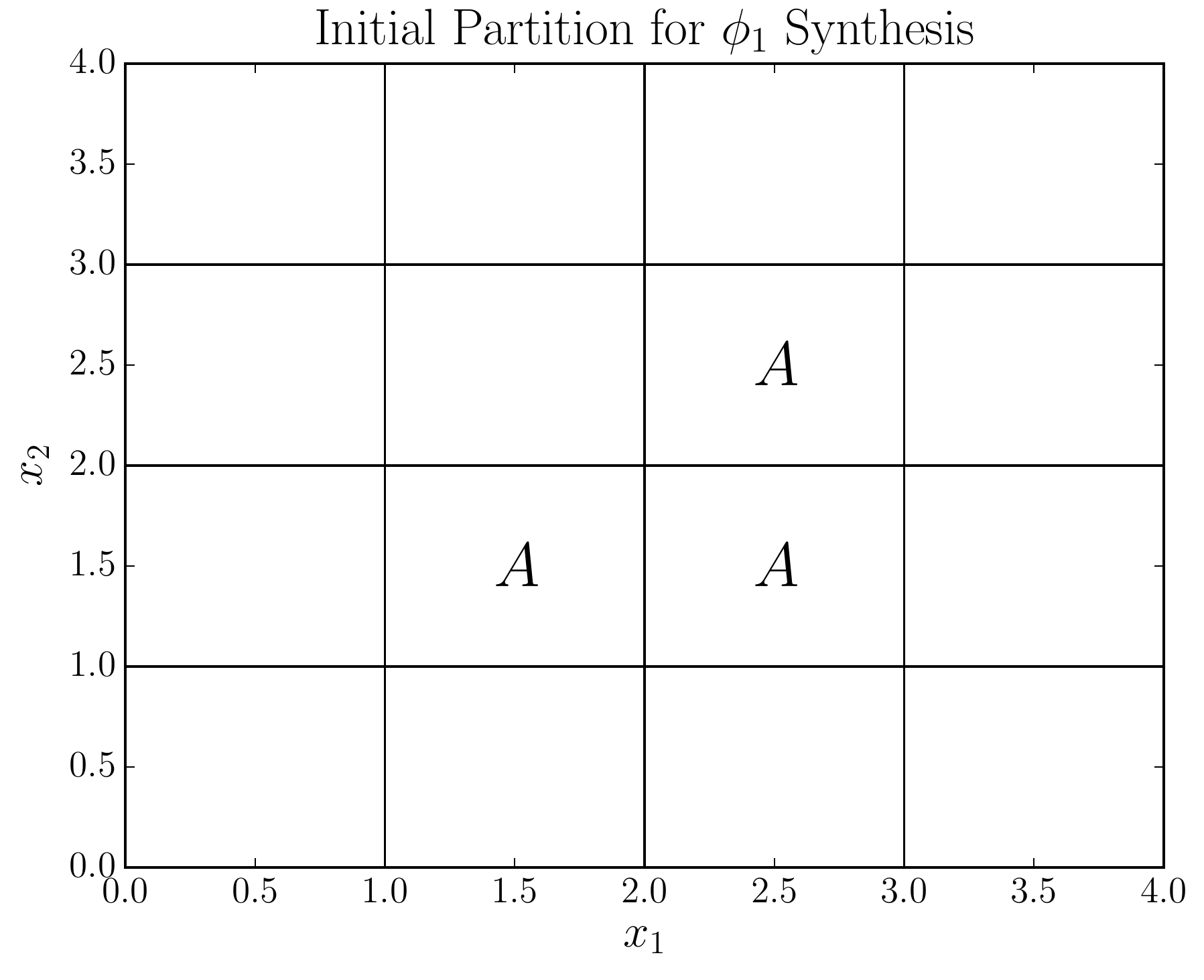}
\includegraphics[scale=0.32]{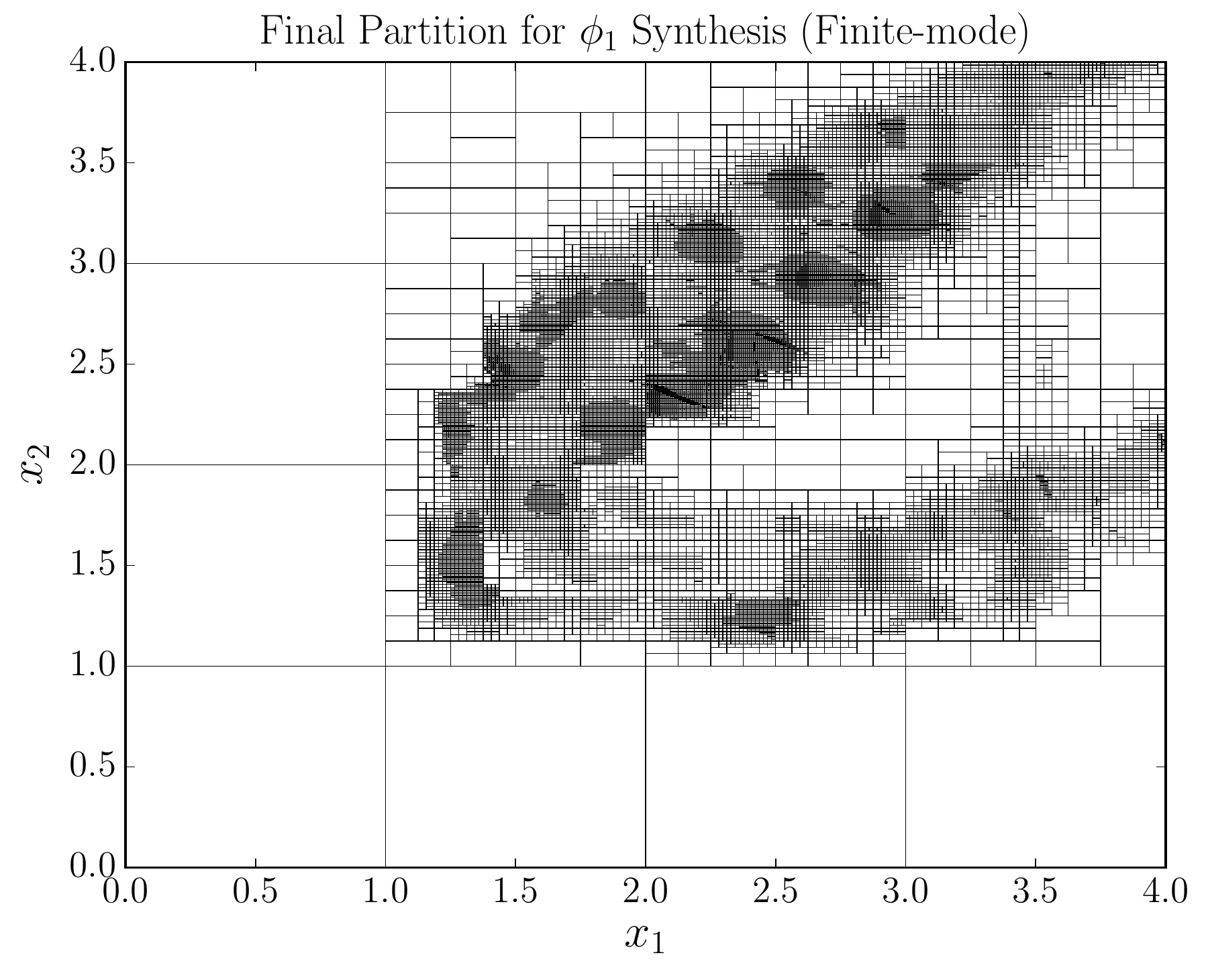}
\caption{Initial domain partition with state labeling (Left) and final domain partition upon synthesis of a controller for maximizing the probability of satisfying $\phi_{1}$ in \eqref{eq5} using the finite set of inputs $U_{fin}$ after 18 refinement steps (Right). The final partition contains 18418 states, corresponding to 92090 states in the resulting product BMDP abstraction.}
\label{FiniteStateSpaces1}
\end{figure}

\begin{figure}[H]
\centering
\includegraphics[scale=0.32]{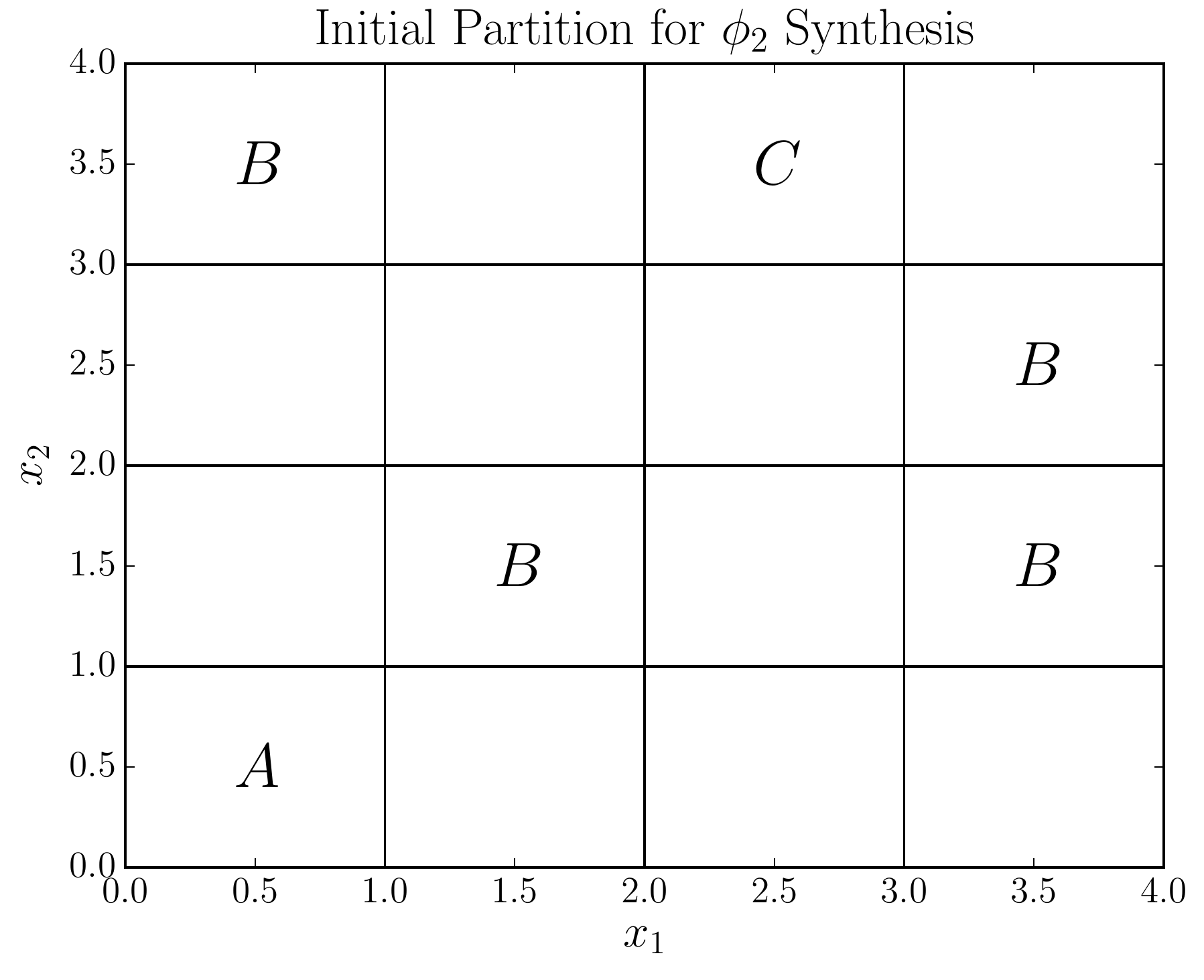}
\includegraphics[scale=0.32]{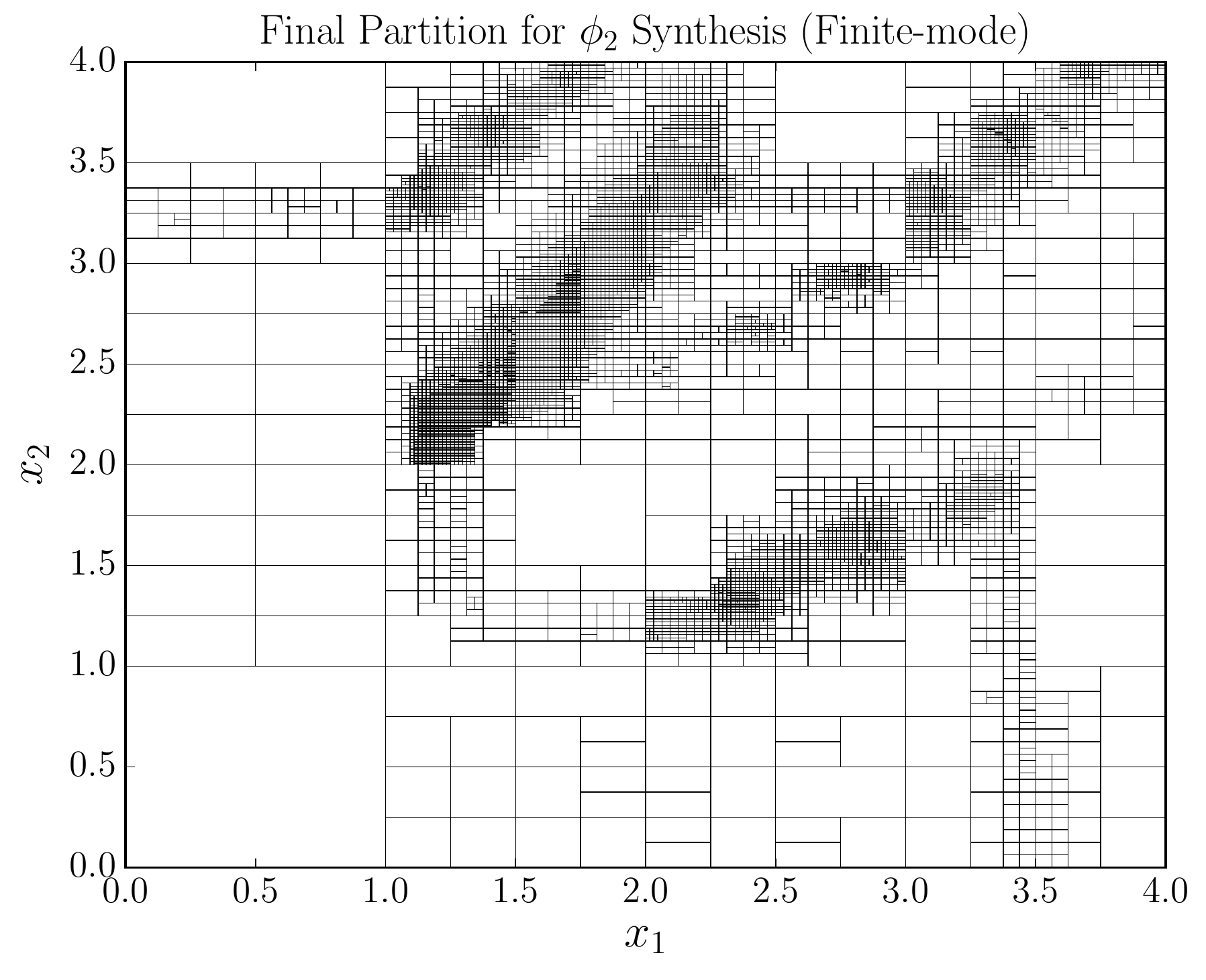}
\caption{Initial domain partition with state labeling (Left) and final domain partition upon synthesis of a controller for maximizing the probability of satisfying $\phi_{2}$ in \eqref{eq5} using the finite set of inputs $U_{fin}$ after 15 refinement steps (Right). The final partition contains 7711 states, corresponding to 53977 states in the resulting product BMDP abstraction.}
\label{FiniteStateSpaces2}
\end{figure}

\begin{figure}[H]
\centering
\includegraphics[scale=0.30]{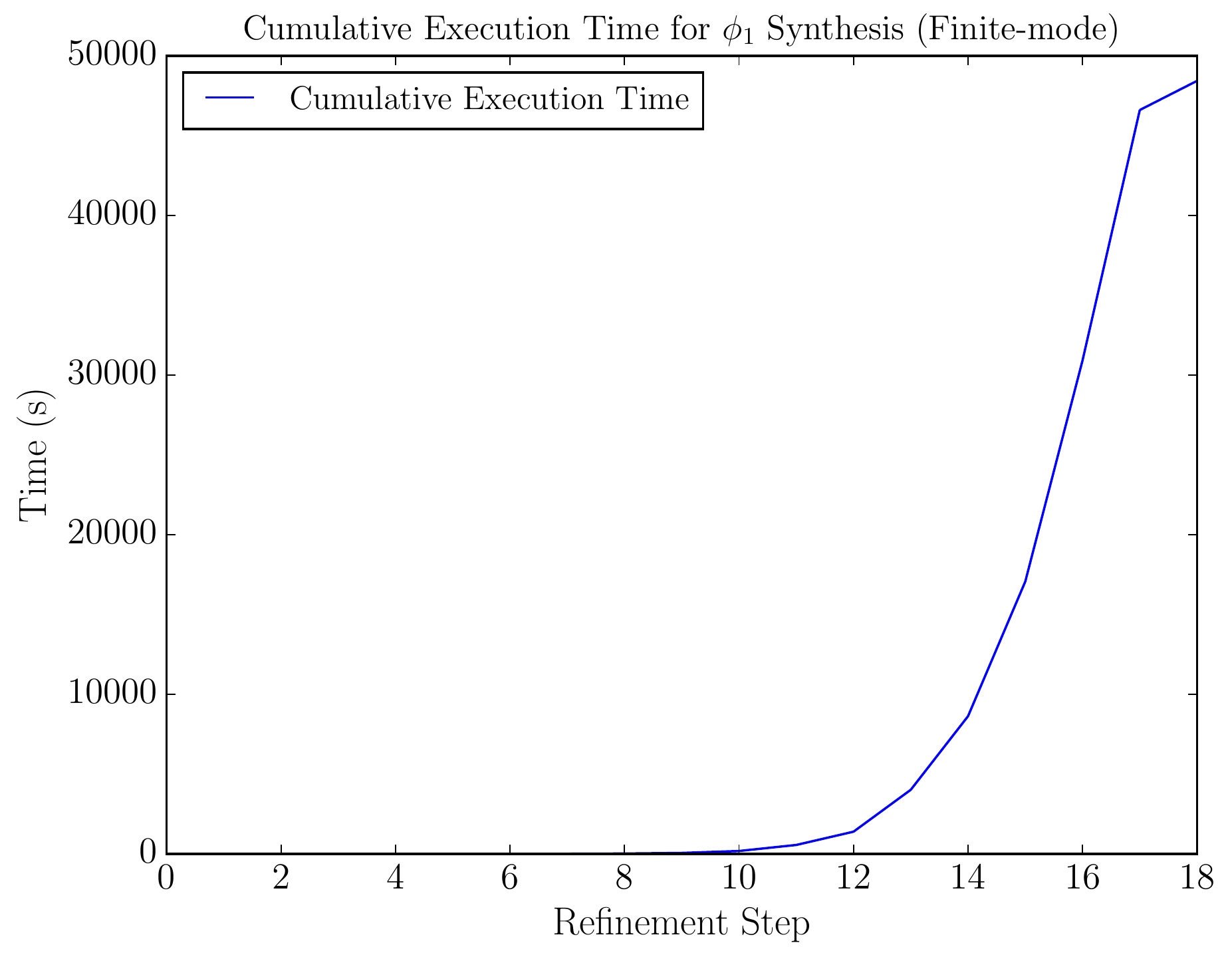}
\includegraphics[scale=0.30]{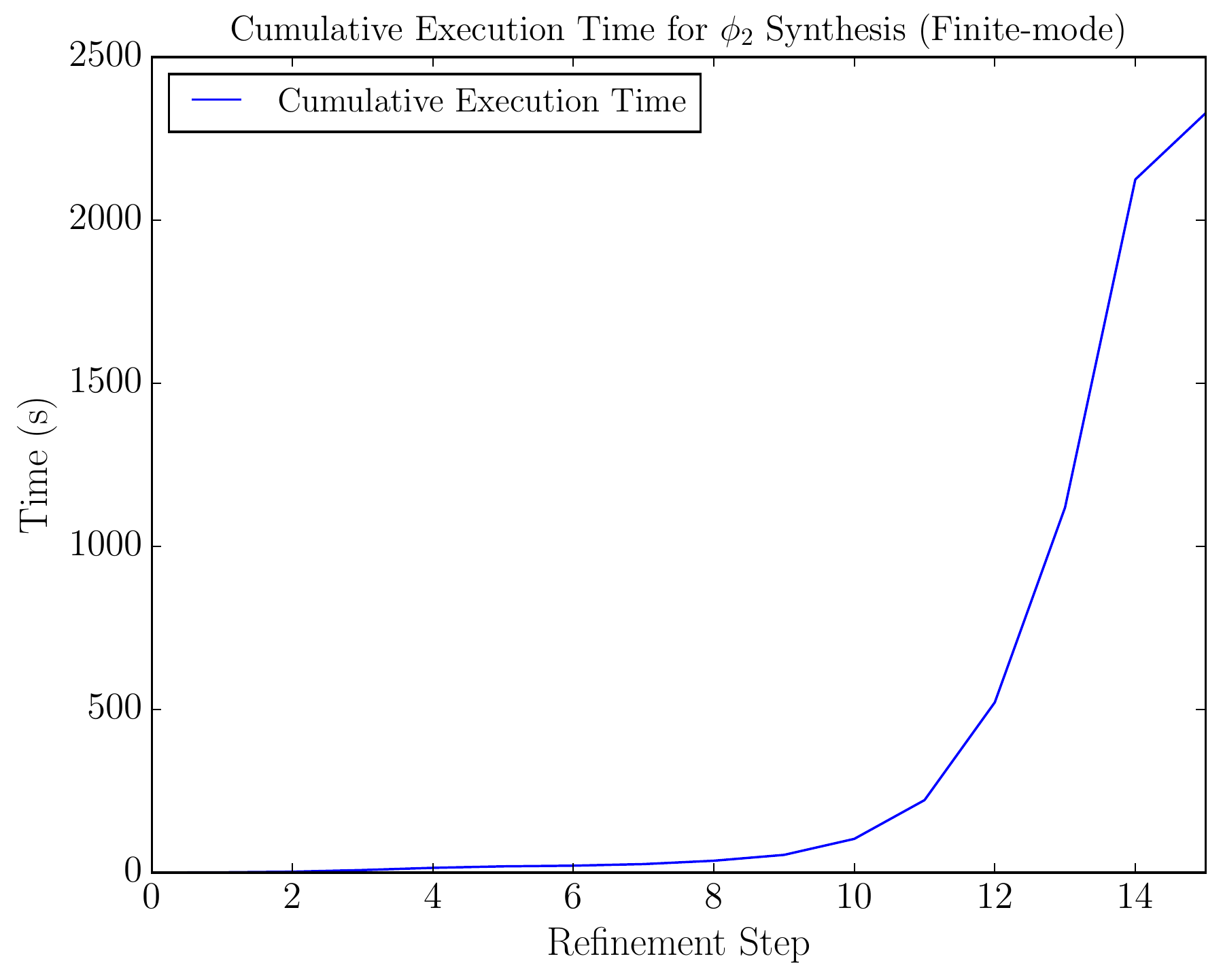}
\caption{Cumulative execution time of the synthesis procedure with the finite input set $U_{fin}$ as a function of the number of refinement steps for specification $\phi_{1}$ (Left) and specification $\phi_{2}$ (Right). The synthesis procedure for $\phi_{1}$ terminated in 13 hours and 27 minutes; the synthesis procedure for $\phi_{2}$ terminated in 38 minutes}
\label{RunningTimesFinite}
\end{figure}

\begin{figure}[H]
\centering
\includegraphics[scale=0.30]{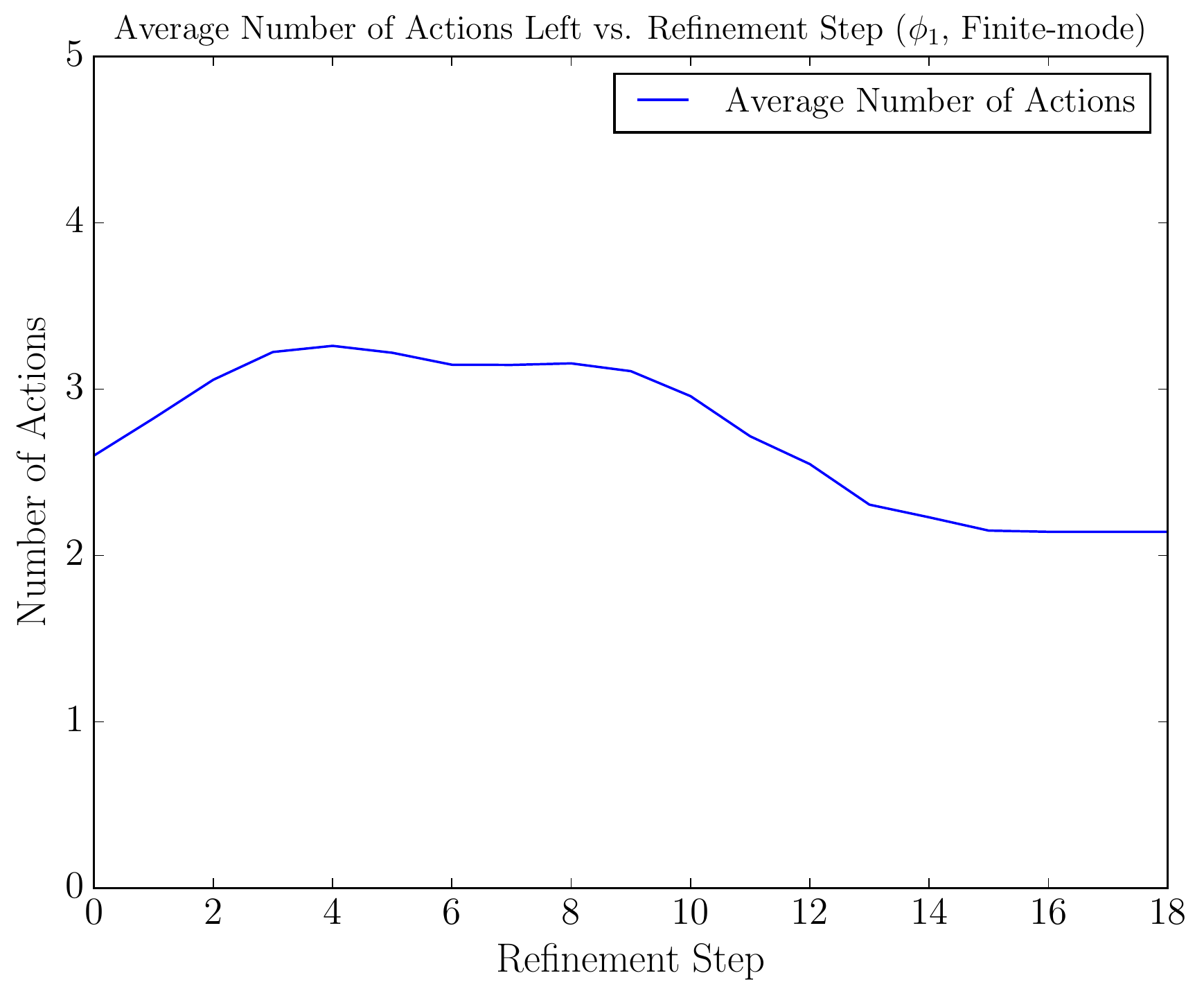}
\includegraphics[scale=0.30]{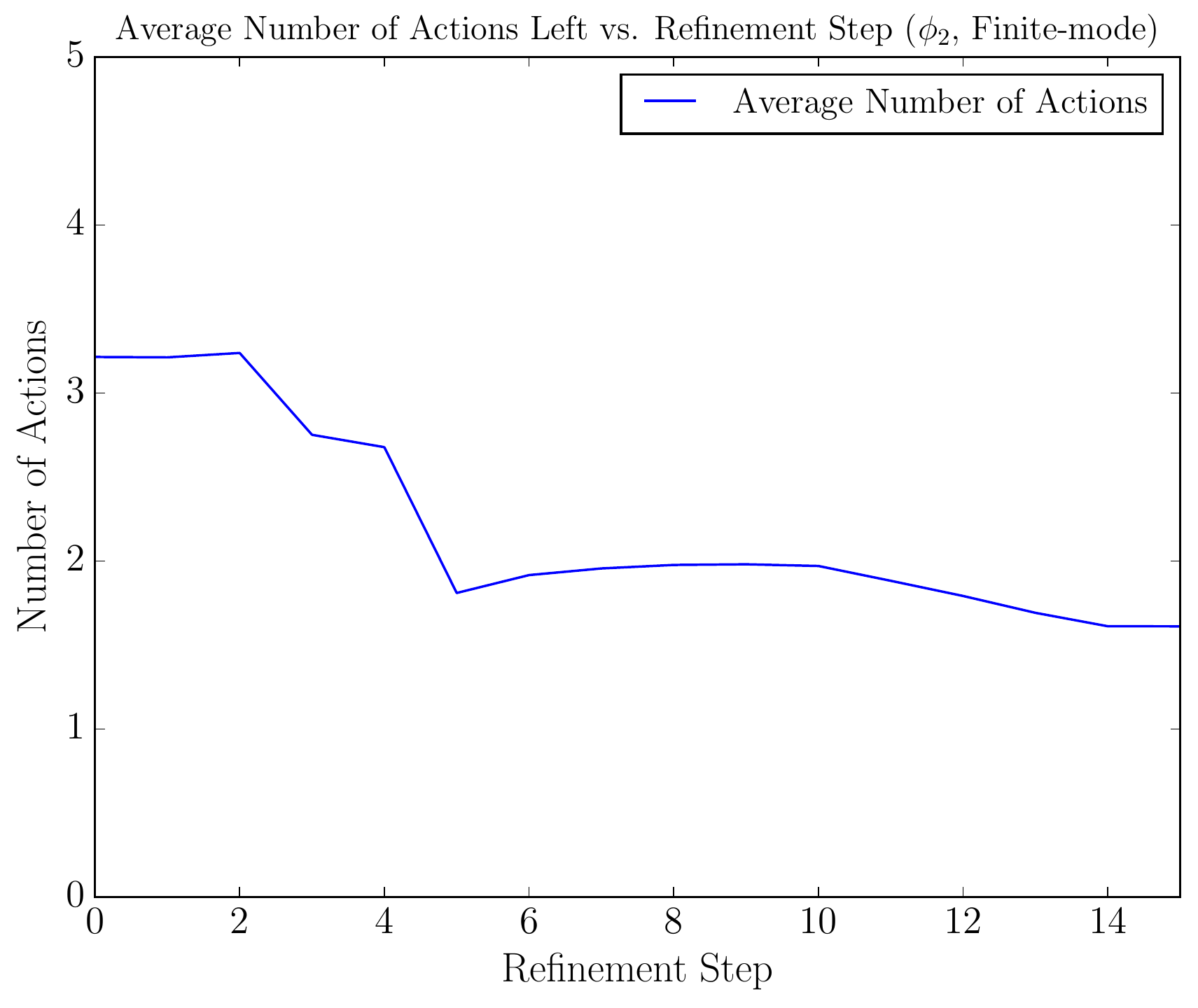}
\caption{Average number of actions left at each state of the product BMDP as a function of the number of refinement steps for specification $\phi_{1}$ (Left) and specification $\phi_{2}$ (Right).}
\label{AvgNumAct}
\end{figure}

\begin{figure}[H]
\centering
\includegraphics[scale=0.30]{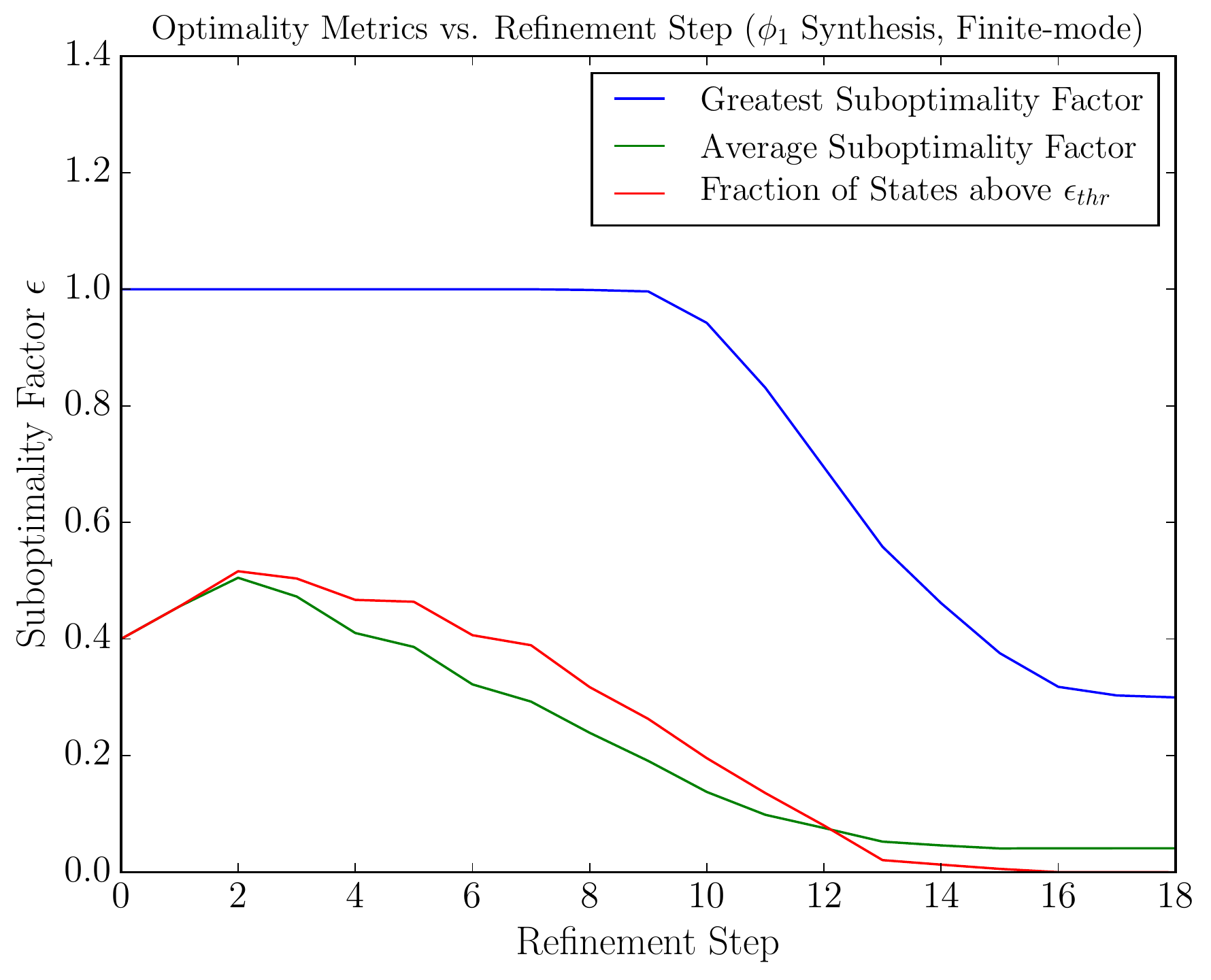}
\includegraphics[scale=0.30]{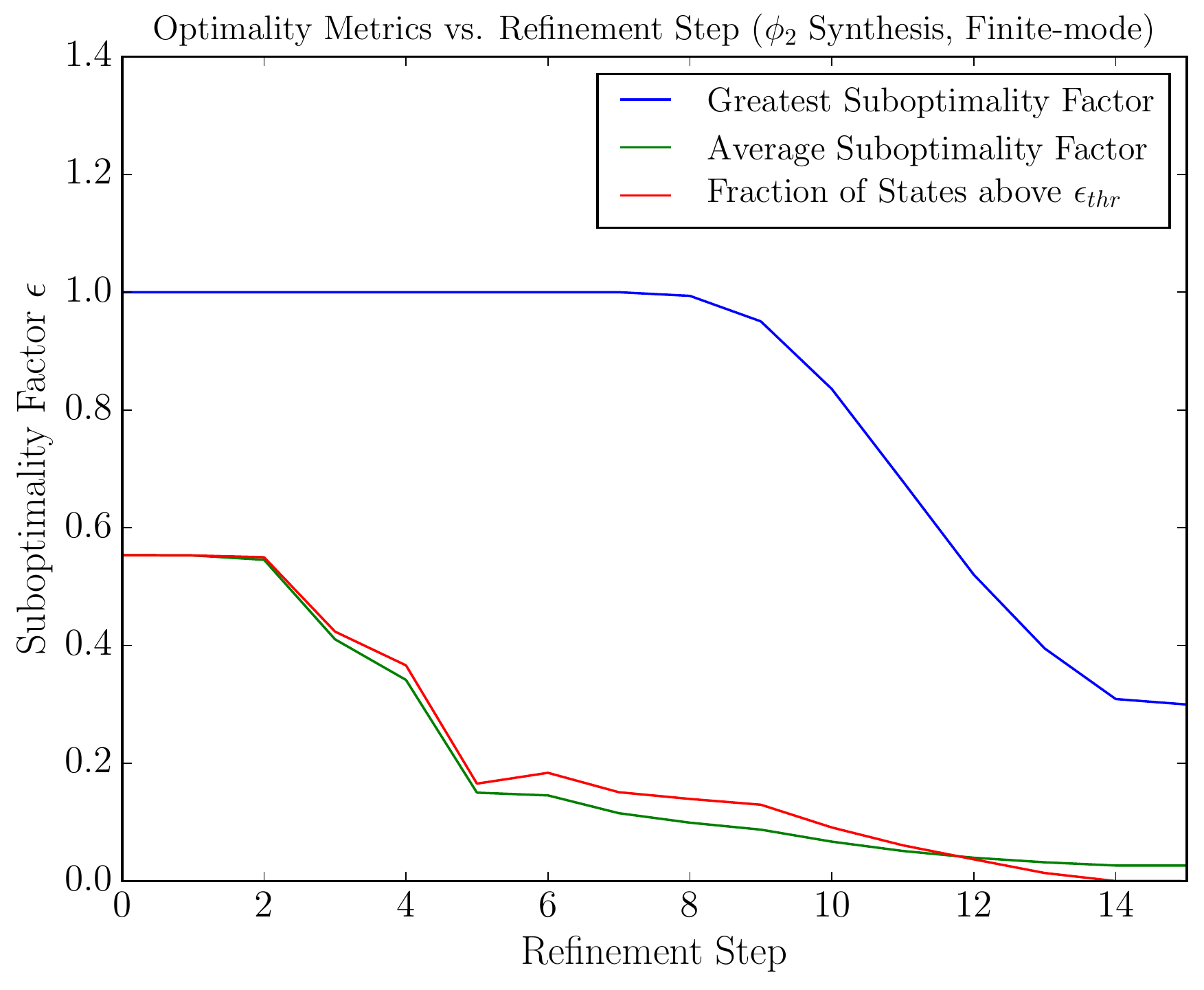}
\caption{Different metrics of precision for the controller computed from the finite input set $U_{fin}$ as a function of the number of refinement steps for specification $\phi_{1}$ (Left) and specification $\phi_{2}$ (Right). The synthesis algorithm reaches the target $\epsilon_{thr} = 0.30$ for both specifications. This means that the probability of satisfying the specifications can only increase by a maximum of 0.30 from all possible states of the abstracted system by choosing another switching policy.}
\label{OptFacFinite}
\end{figure}

\subsection{CONTINUOUS INPUT SET SYNTHESIS}
\label{continuousinputex}

Next, we generate a control policy from the set of continuous inputs $U$ by applying Algorithm \ref{ContInputSyn}. 

The desired threshold precision is chosen to be $\epsilon_{thr} = 0.30$. At each refinement step, states of the current partition with a refinement score that is greater than 1\% of the maximum score are chosen to be refined and split in half along their greatest dimension. Tight rectangular over-approximation of the deterministic reachable set of \eqref{eq5} are obtained efficiently from the results in \cite{coogan2015efficient} thanks to the monotone property of the state update map. The input space of all states in the product CIMC is stored as a union of rectangles. When evaluating the optimality of the synthesized controller before every refinement step, we partition each rectangle of the input space of all states into 4 rectangles of equal area. This allows the input spaces to always remain a union of rectangles in case some sub-regions of the input space were removed, as in Figure \ref{fig_update}, which facilitates the computation of the overlaps in Algorithm \ref{CompConsAlg}.

The non-convex optimization problem in Algorithm \ref{InpSelecAlg}, line 14, and the non-convex optimization problem \eqref{optimiz} are solved by gridding each rectangle $U_{i}$ of the input space of interest with an $N$-by-$N$ meshgrid, where $N = \text{max}  \lbrace N_{min}, $ $\lceil{N_{init} \cdot \frac{Area(U_{i})}{Area(U)} \rceil} \rbrace$ with $N_{min} = 3$ and $N_{init} = 12$, and using a convex solver from all points of the grid. The component construction algorithm is conducted at each iteration of the while loop of Algorithm \ref{ContInputSyn} until the set of potential accepting BSCCs $(U)_{\textcolor{black}{pot}}^{G}$ becomes empty, as in the finite-mode examples. The threshold of convergence for the reachability value iteration scheme is set to 0.01.

The controller synthesis procedure for specification $\phi_{1}$ was manually terminated after 12 refinement steps which lasted 22 hours and 32 minutes with a greatest suboptimality factor $\epsilon_{max} = 0.8705$, and created 16079 states, corresponding to 80395 states in the product BMDP constructed from the final partition. The final refined partition is displayed in Figure \ref{StateSpaceCont1} (Right). The procedure for specification $\phi_{2}$ was manually terminated after 14 refinement steps which lasted 73 hours with a greatest suboptimality factor $\epsilon_{max} = 0.7754$, and created 24607 states in 14 refinement steps, corresponding to 172249 states in the product BMDP constructed from the final partition. The final refined partition is displayed in Figure \ref{StateSpaceCont2} (Right). 

The cumulative execution time against the number of refinement steps is plotted in Figure \ref{RunningTimesCont} for specification $\phi_{1}$ (Left) and specification $\phi_{2}$ (Right). The original input space for all states of the system is shown in Figure \ref{InputSpaceCont}, along with the reduced input space with respect to specification $\phi_{1}$ and $\phi_{2}$ upon refinement for 2 states of the system. Finally, the greatest suboptimality factor, average suboptimality factor of the product CIMC and fractions of states above the target precision $\epsilon_{thr}$ as a function of the number of refinement steps are shown in Figure \eqref{OptFacCont} for specification $\phi_{1}$ (Left) and specification $\phi_{2}$ (Right).

\begin{figure}[H]
\centering
\includegraphics[scale=0.32]{InitialPartition_phi1}
\includegraphics[scale=0.32]{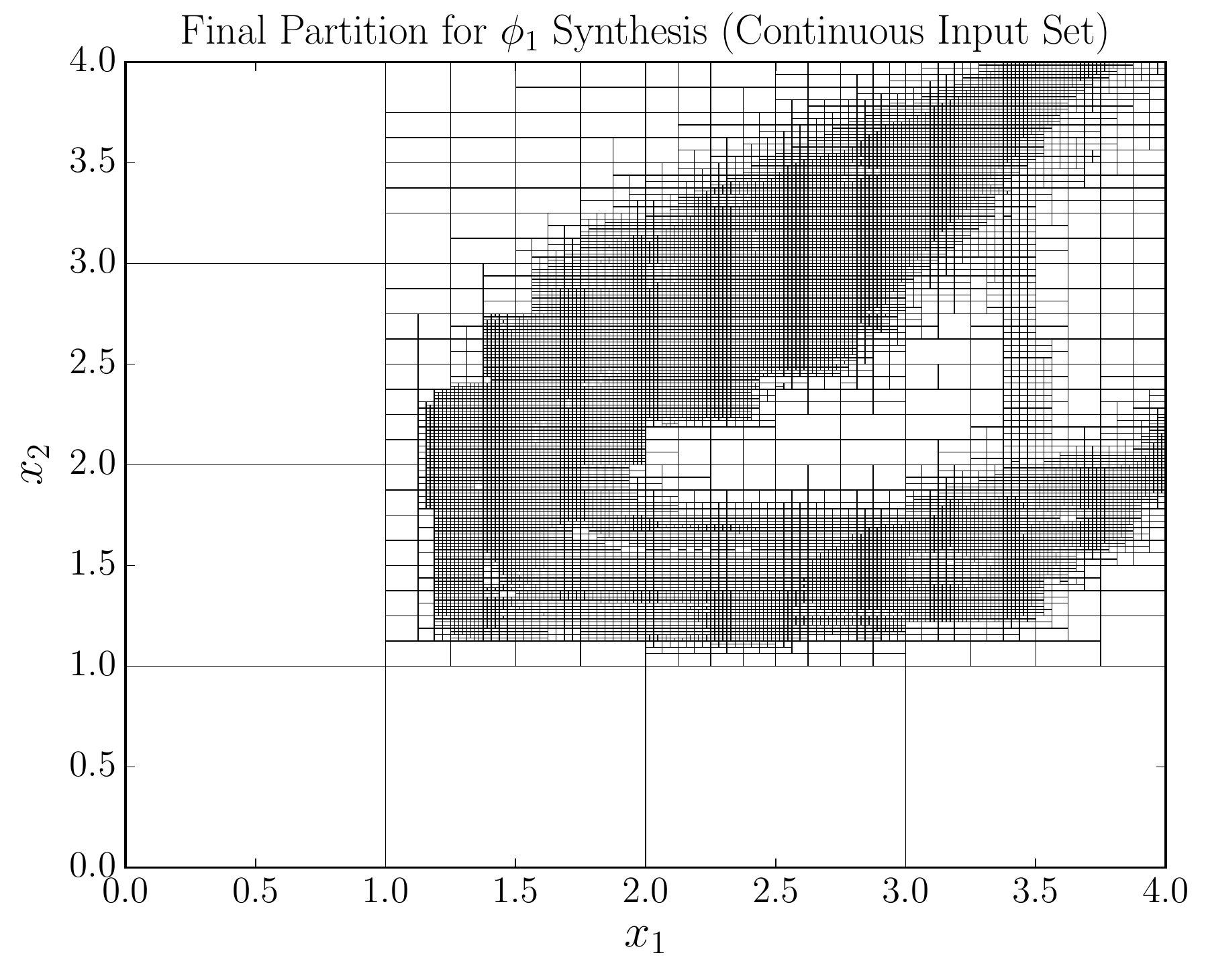}
\caption{Initial domain partition with state labeling (Left) and final domain partition upon synthesis of a controller for maximizing the probability of satisfying $\phi_{1}$ using the continuous set of inputs $U$ after 12 refinement steps (Right). The final partition contains 16079 states, corresponding to 80395 states in the resulting product CIMC abstraction.}
\label{StateSpaceCont1}
\end{figure}

\begin{figure}[H]
\centering
\includegraphics[scale=0.32]{InitialPartition_phi2}
\includegraphics[scale=0.32]{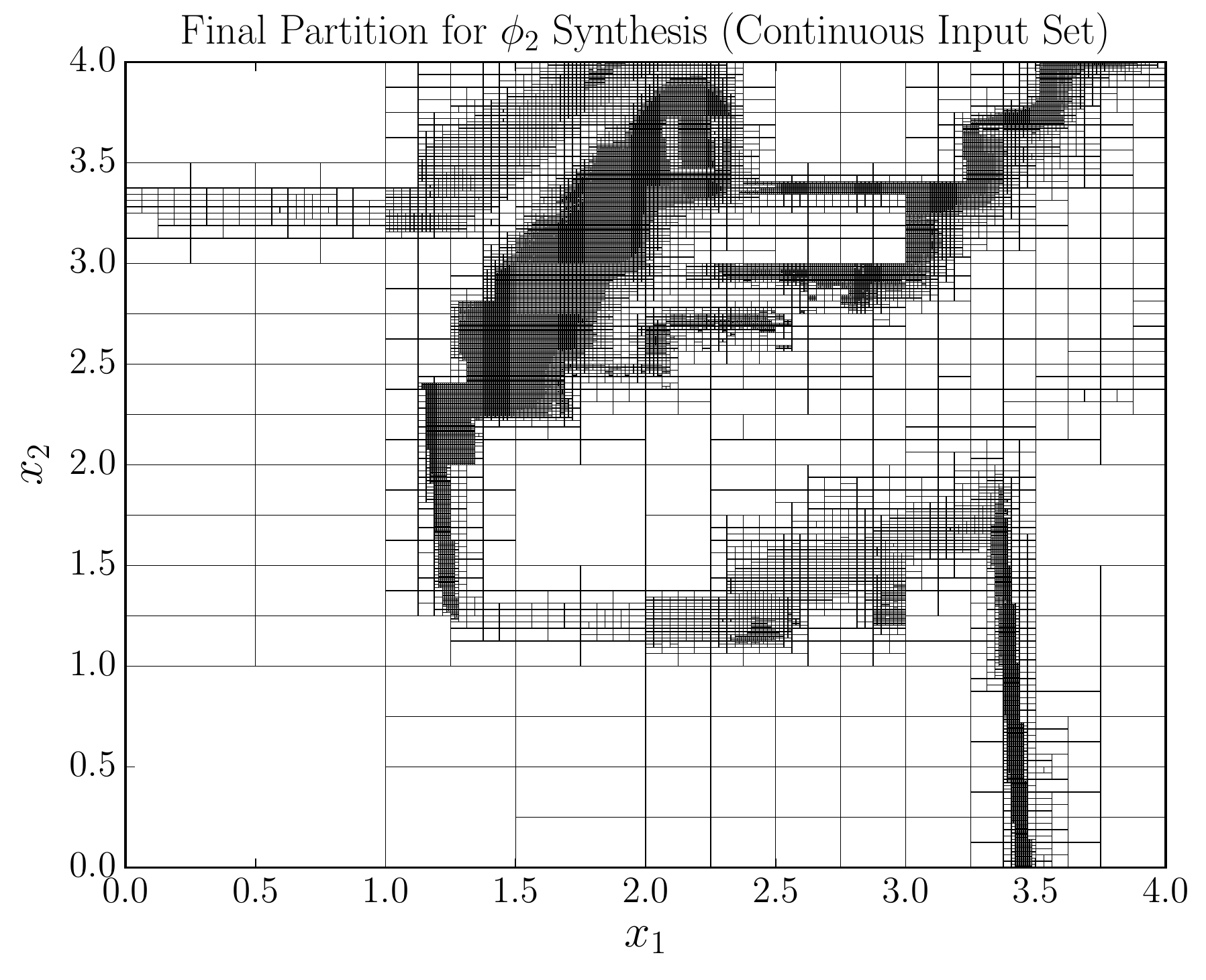}
\caption{Initial domain partition with state labeling (Left) and final domain partition upon synthesis of a controller for maximizing the probability of satisfying $\phi_{2}$ using the continuous set of inputs $U$ after 14 refinement steps (Right). The final partition contains 24607 states, corresponding to 172249 states in the resulting product CIMC abstraction.}
\label{StateSpaceCont2}
\end{figure}

\begin{figure}[H]
\centering
\includegraphics[scale=0.31]{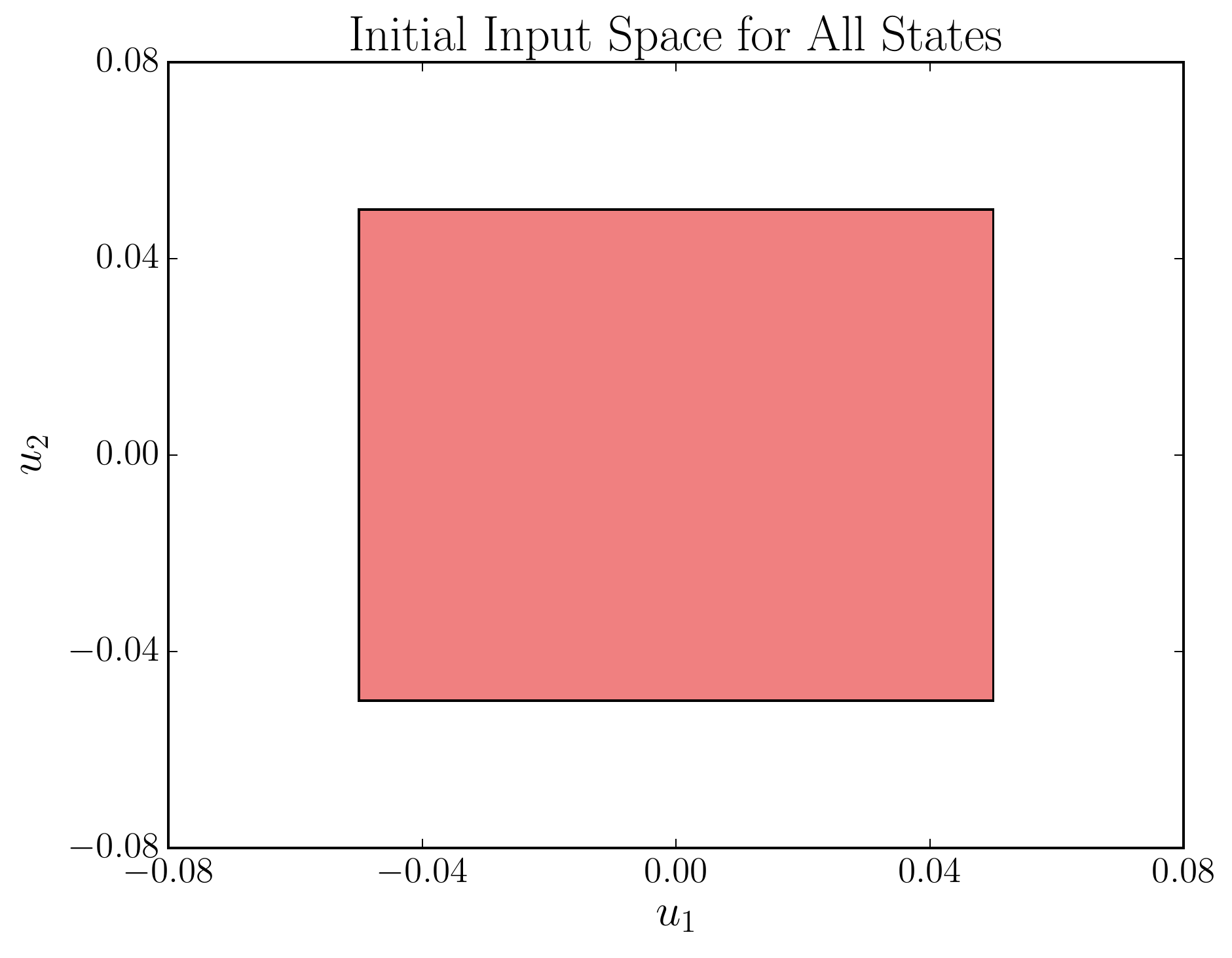}\\
\includegraphics[scale=0.31]{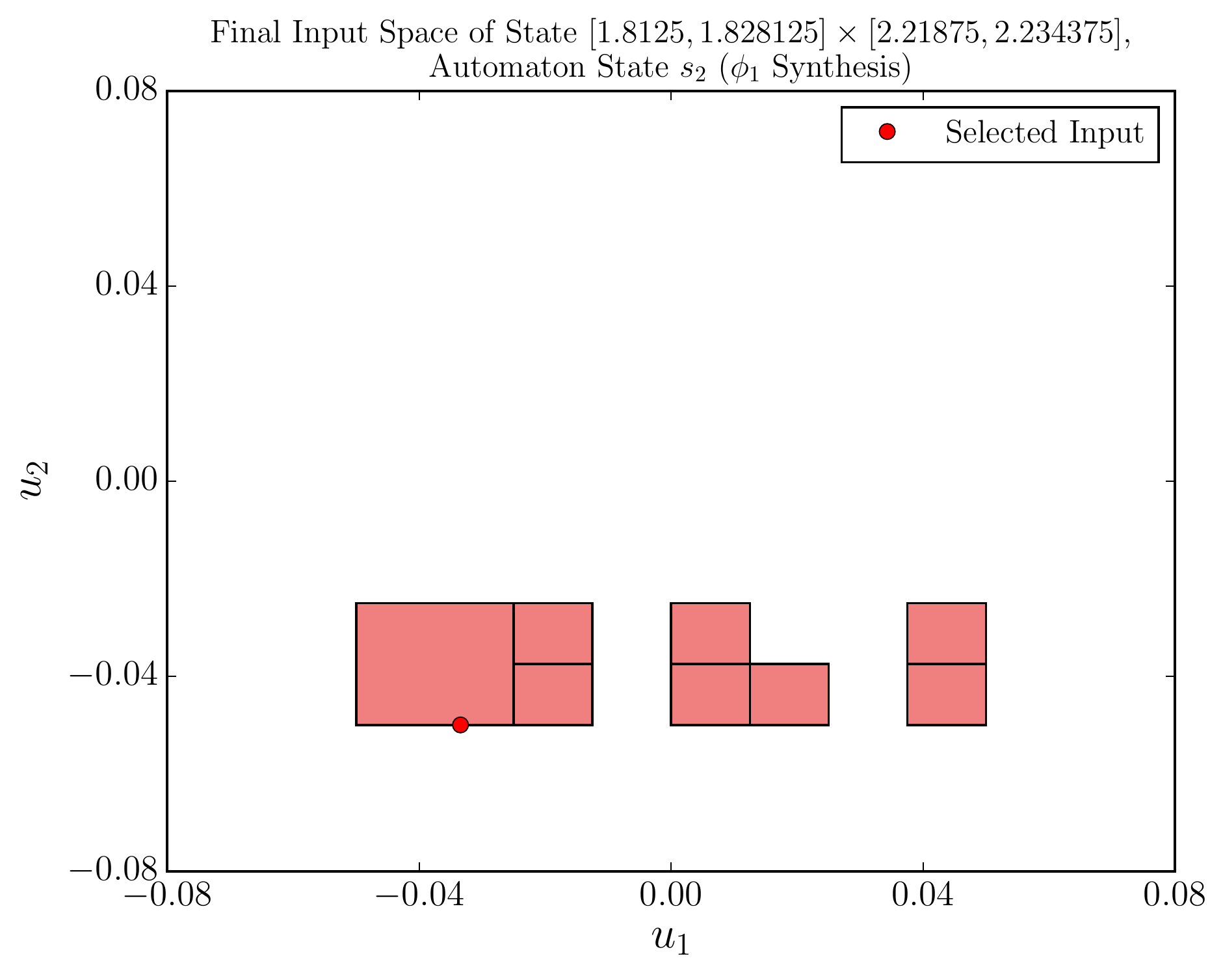}
\includegraphics[scale=0.31]{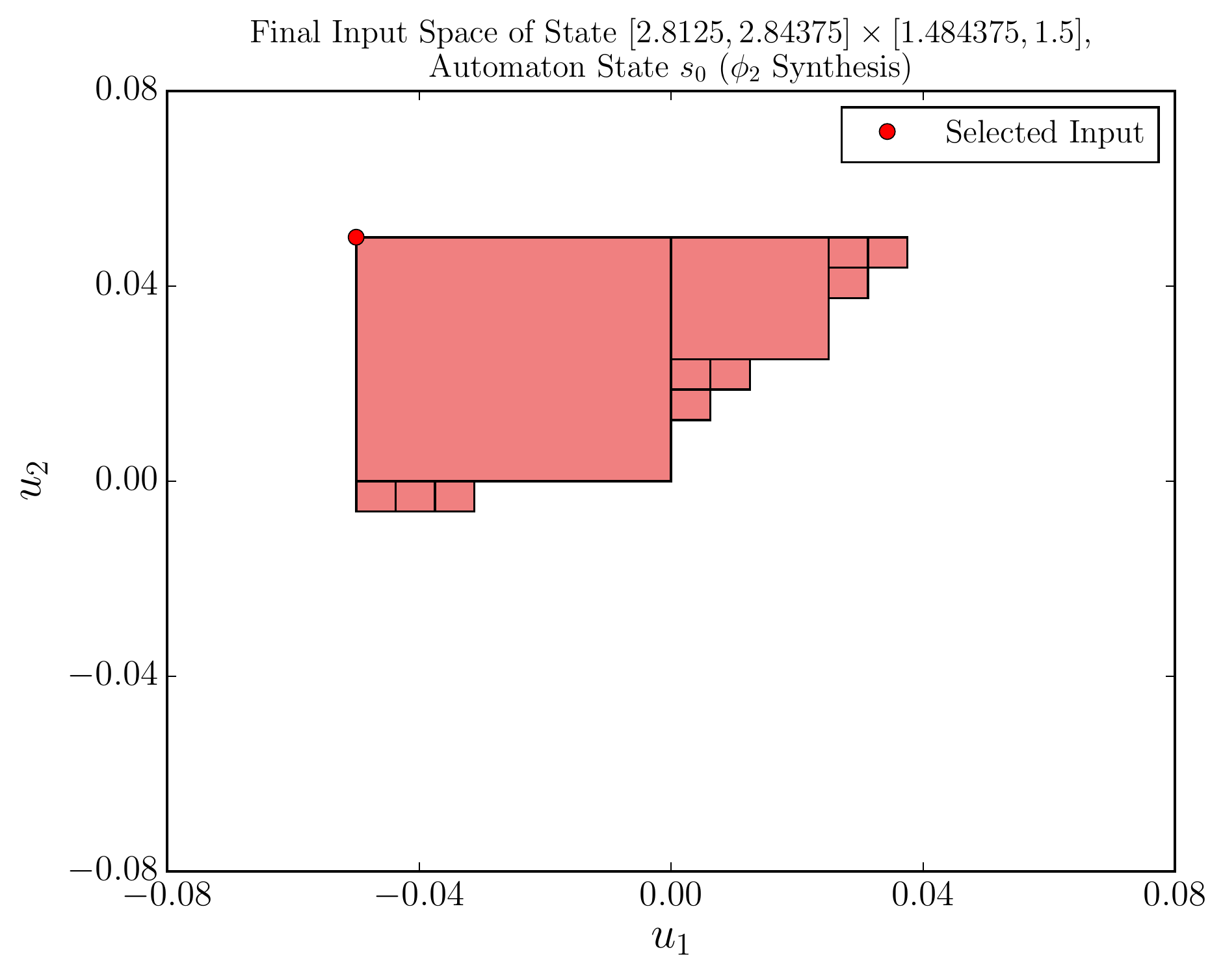}
\caption{Plot of the initial input space $U$ (Top) for all states of the state space. The reduced input space of state $[1.8125, 1.828125] \times [2.21875, 2.234375]$ with automaton state $s_2$ with respect to specification $\phi_1$ upon refinement is shown in the bottom left plot. The reduced input space of state $[2.8125, 2.84375] \times [1.484375, 1.5]$ with automaton state $s_{0}$ with respect to specification $\phi_2$ upon refinement is shown in the bottom right plot.}
\label{InputSpaceCont}
\end{figure}

\begin{figure}[H]
\centering
\includegraphics[scale=0.30]{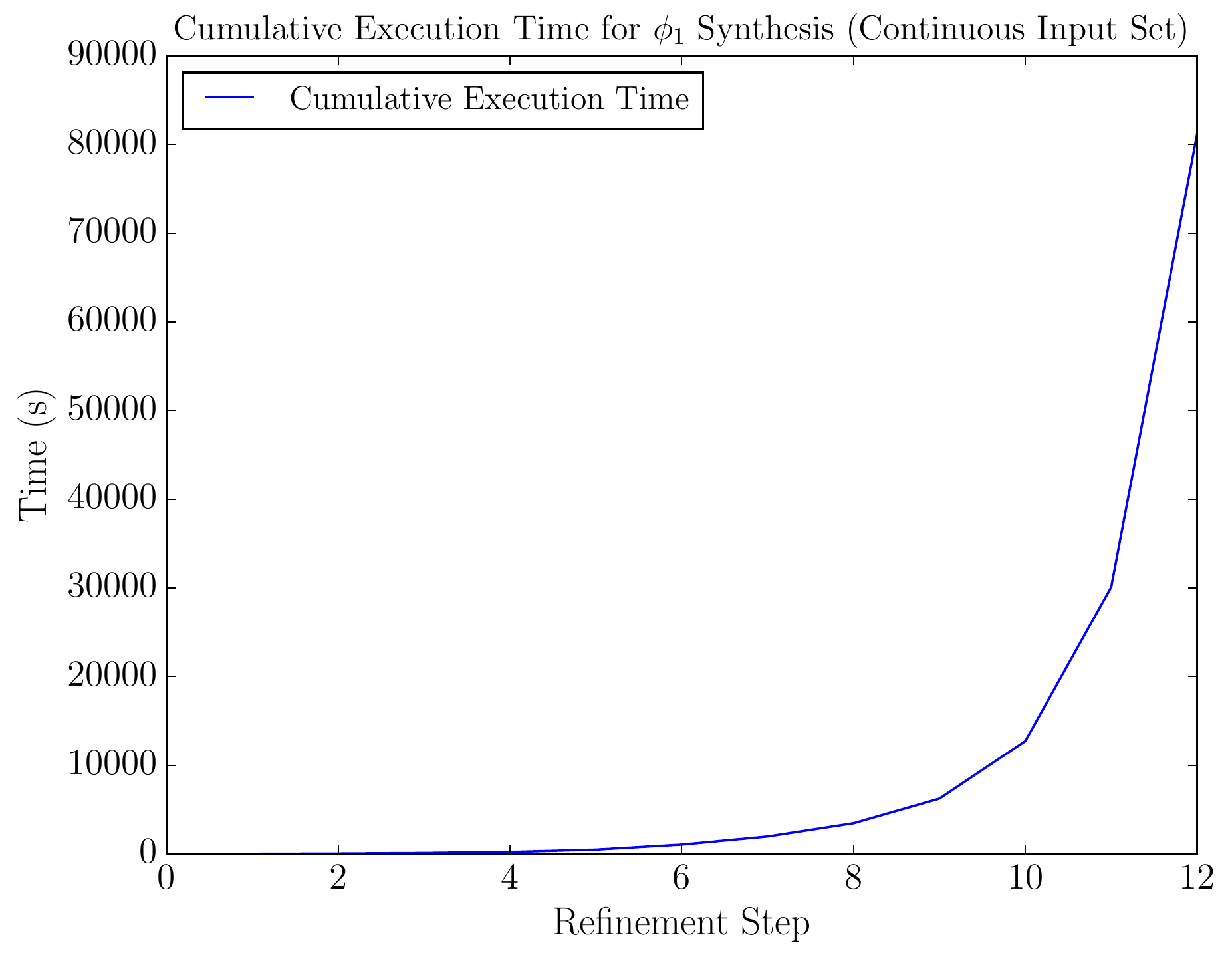}
\includegraphics[scale=0.30]{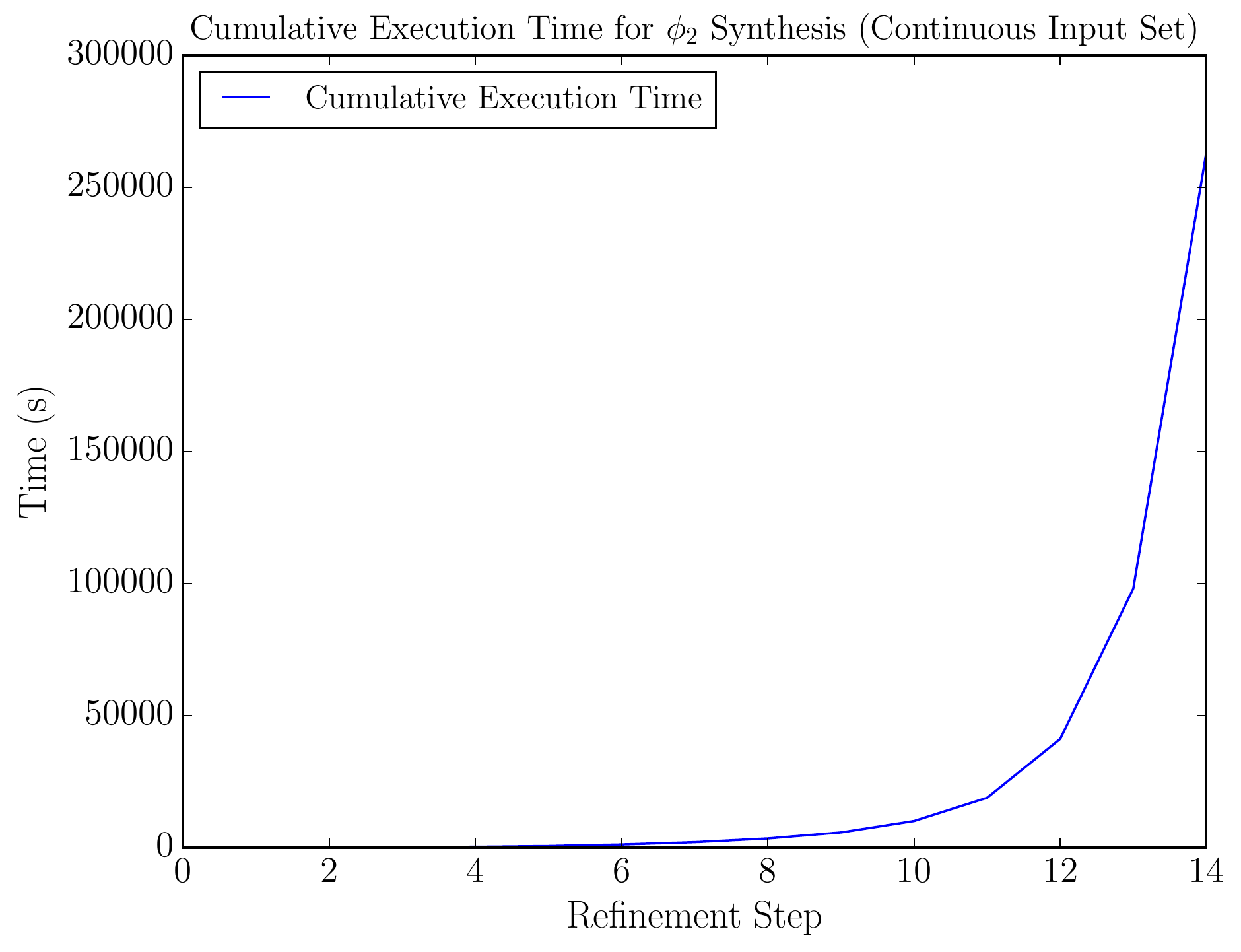}
\caption{Cumulative execution time of the synthesis procedure with the continuous input set $U$ as a function of the number of refinement steps for specification $\phi_{1}$ (Left) and specification $\phi_{2}$ (Right).}
\label{RunningTimesCont}
\end{figure}

\begin{figure}[H]
\centering
\includegraphics[scale=0.30]{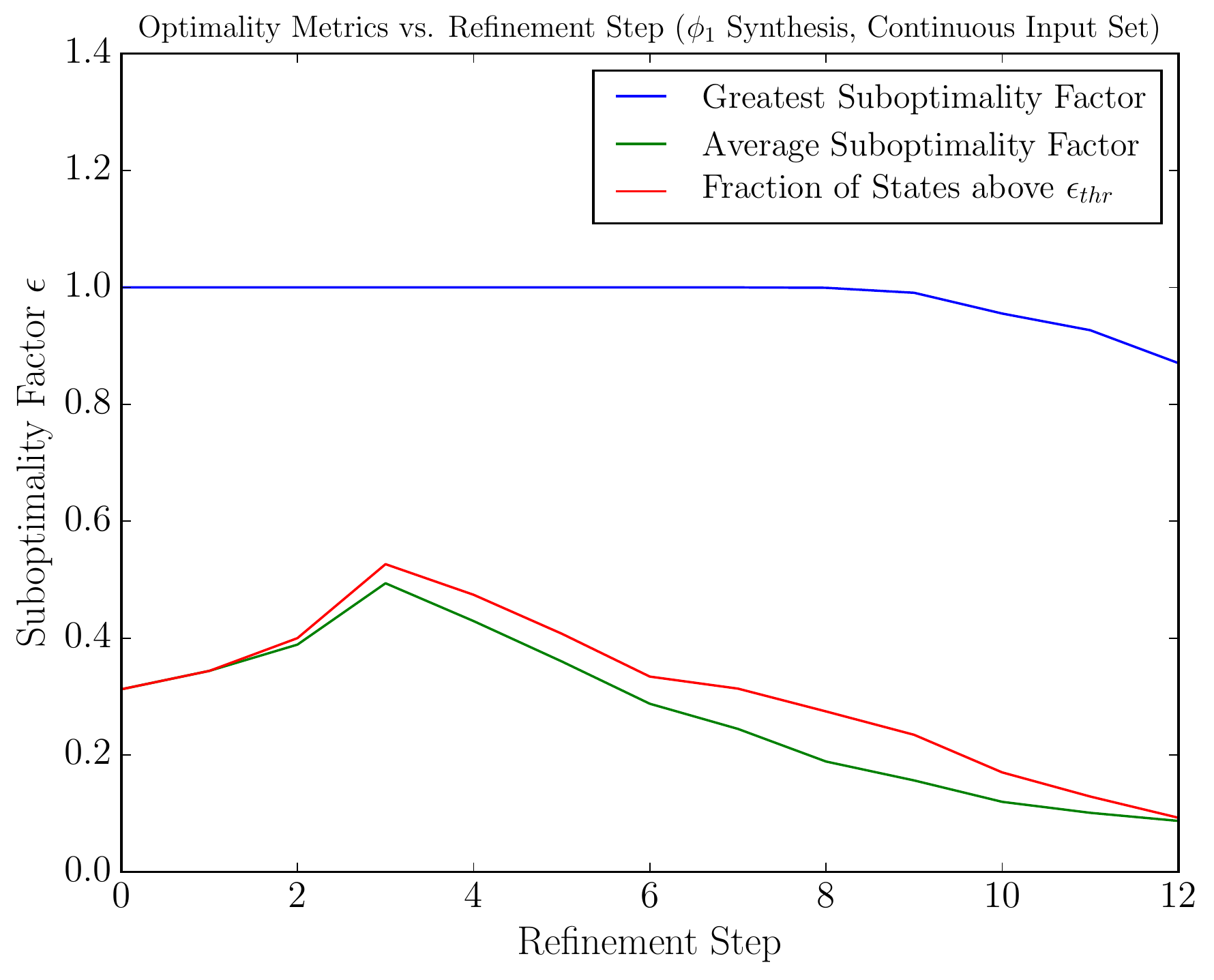}
\includegraphics[scale=0.30]{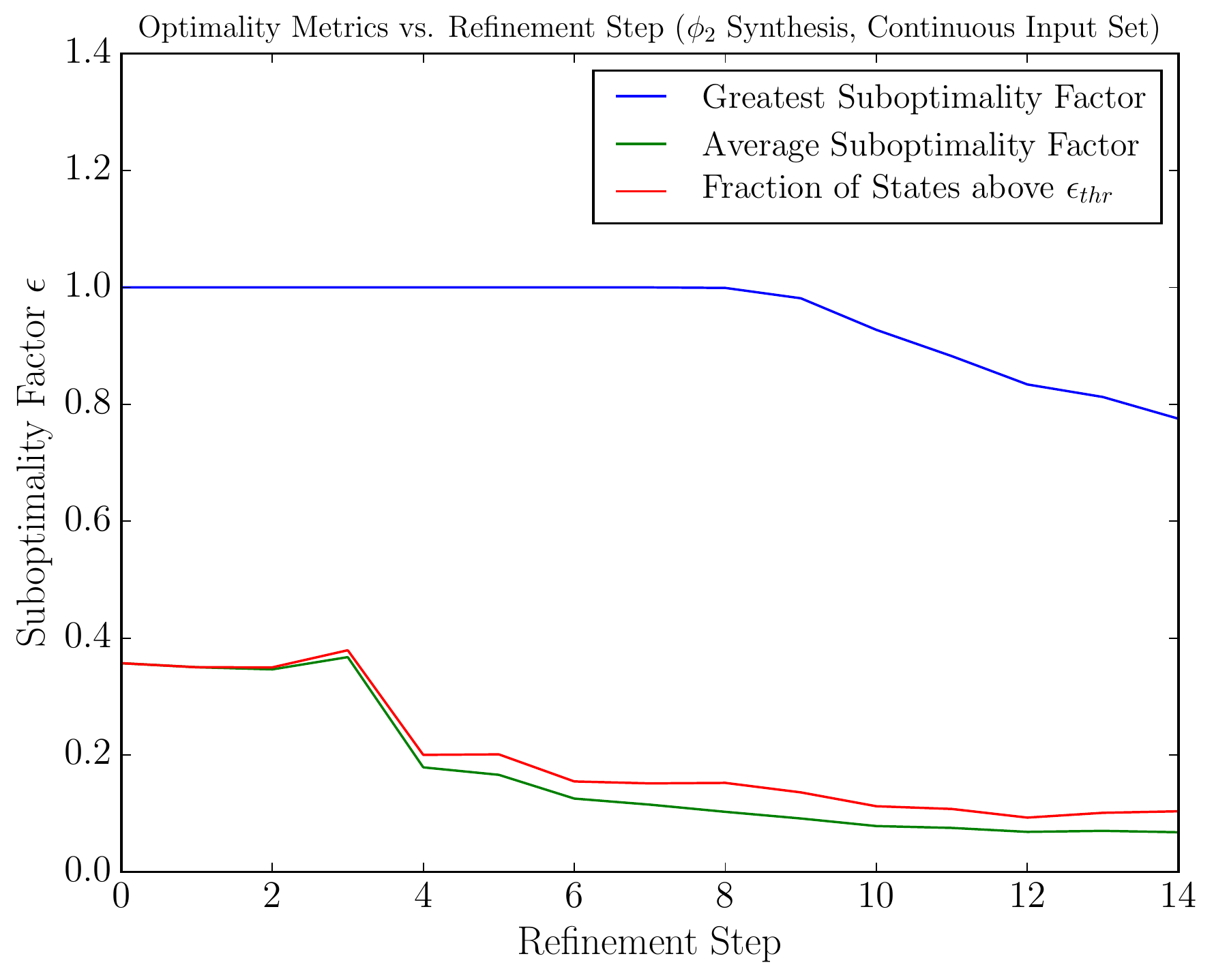}
\caption{Different metrics of precision for the computed controller with the continuous input set as a function of the number of refinement steps for specification $\phi_{1}$ (Left) and specification $\phi_{2}$ (Right). The synthesis algorithm is manually terminated before reaching the target $\epsilon_{thr} = 0.30$ for both specifications.}
\label{OptFacCont}
\end{figure}

\subsection{DISCUSSION}

The synthesis algorithms presented in the previous sections successfully designed controllers from both the finite set of inputs $U_{fin}$ and the continuous set of inputs $U$. Moreover, the algorithms conducted synthesis for two different complex specifications that existing tools could not accommodate, and automatically produced a targeted domain refinement for the two cases so as to achieve a higher level of optimality for the computed controllers. We also consider our approach to be an improvement over related synthesis works in terms of scalability; for instance, our finite-mode algorithm is orders of magnitude faster than the technique used for the synthesis case study in \cite{lahijanian2015formal}, which designed a switching policy for a 3-mode 2D linear system with a simple reachability specification over the course of several days.

To further demonstrate the synthesis procedure, in Figure \ref{Verifcomp} (Top), we display the verification of system \eqref{eq5} against $\phi_{1}$ without any available input with respect to a satisfaction threshold of 0.8 from the work in \cite{dutreix2020specification}, where the initial states in green have a probability of satisfying the specification which is greater than 0.8, the states in red have a probability which is below 0.8, and the states in yellow are undecided at the level of precision of the available partition. In the bottom left, we display the verification of system \eqref{eq5} under the computed switching policy in the finite-mode section, and in the bottom right, we show the verification of system \eqref{eq5} under the computed control policy from the continuous set of inputs. As expected, moving counter-clockwise through the plots, we observe that some red regions of the state-space are converted to green regions.

\begin{figure}[H]
\centering
\includegraphics[scale=0.33]{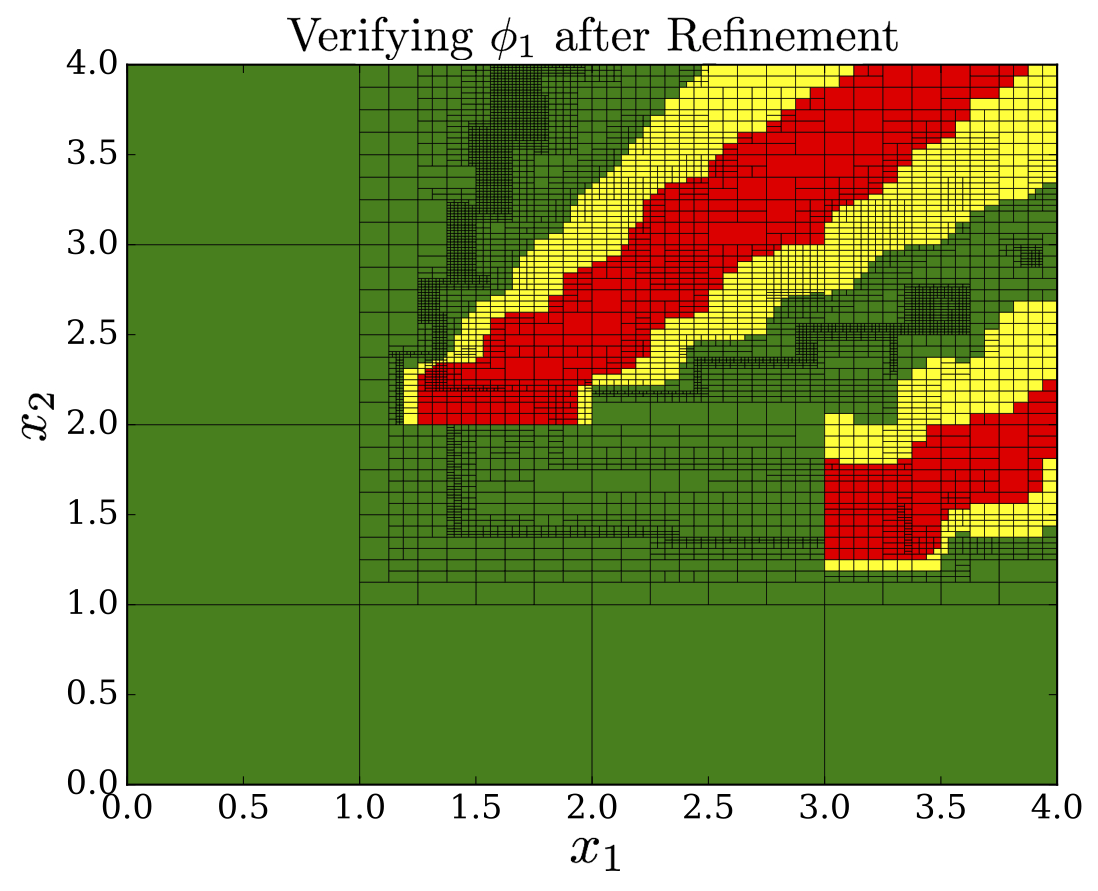}\\
\includegraphics[scale=0.325]{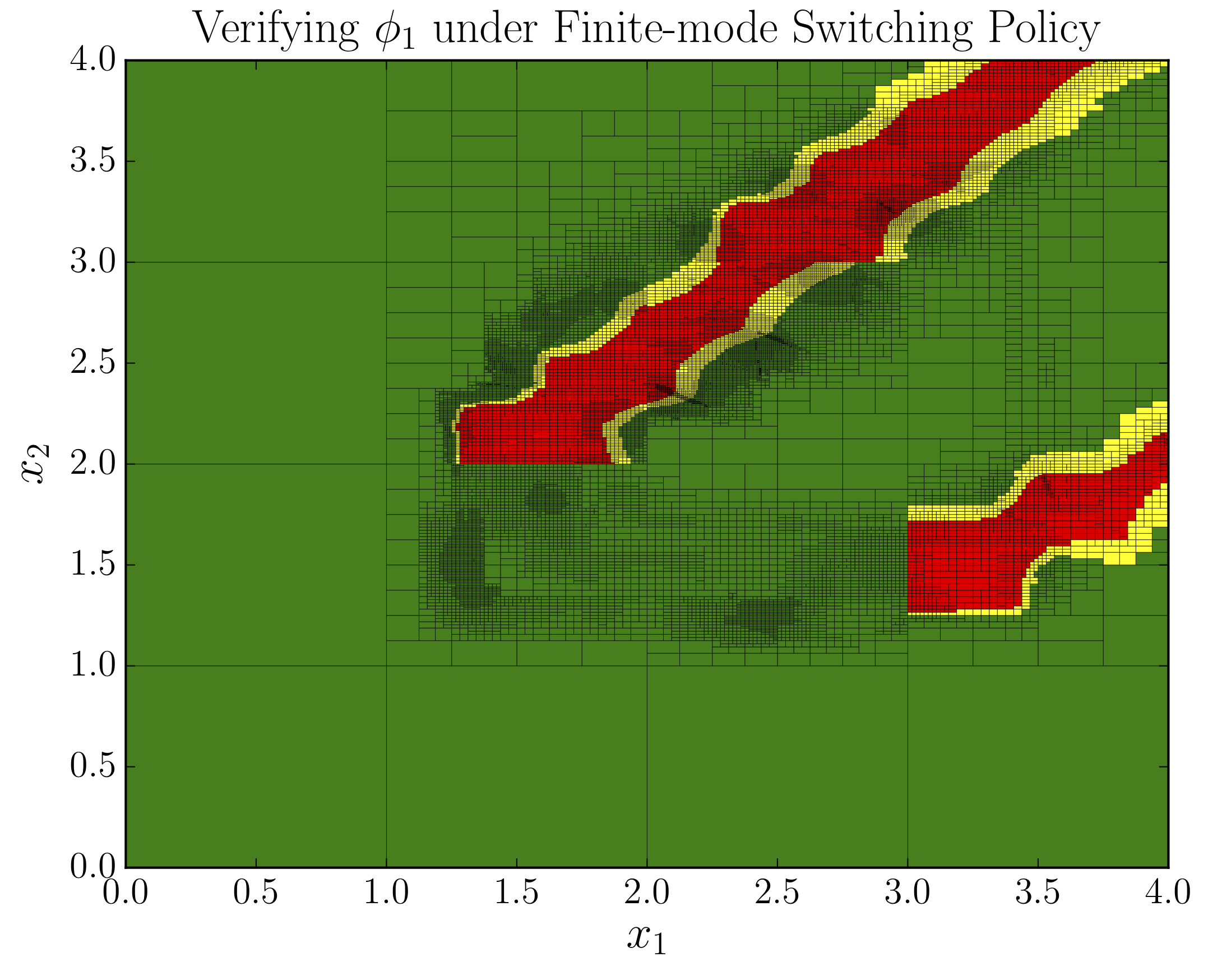}
\includegraphics[scale=0.325]{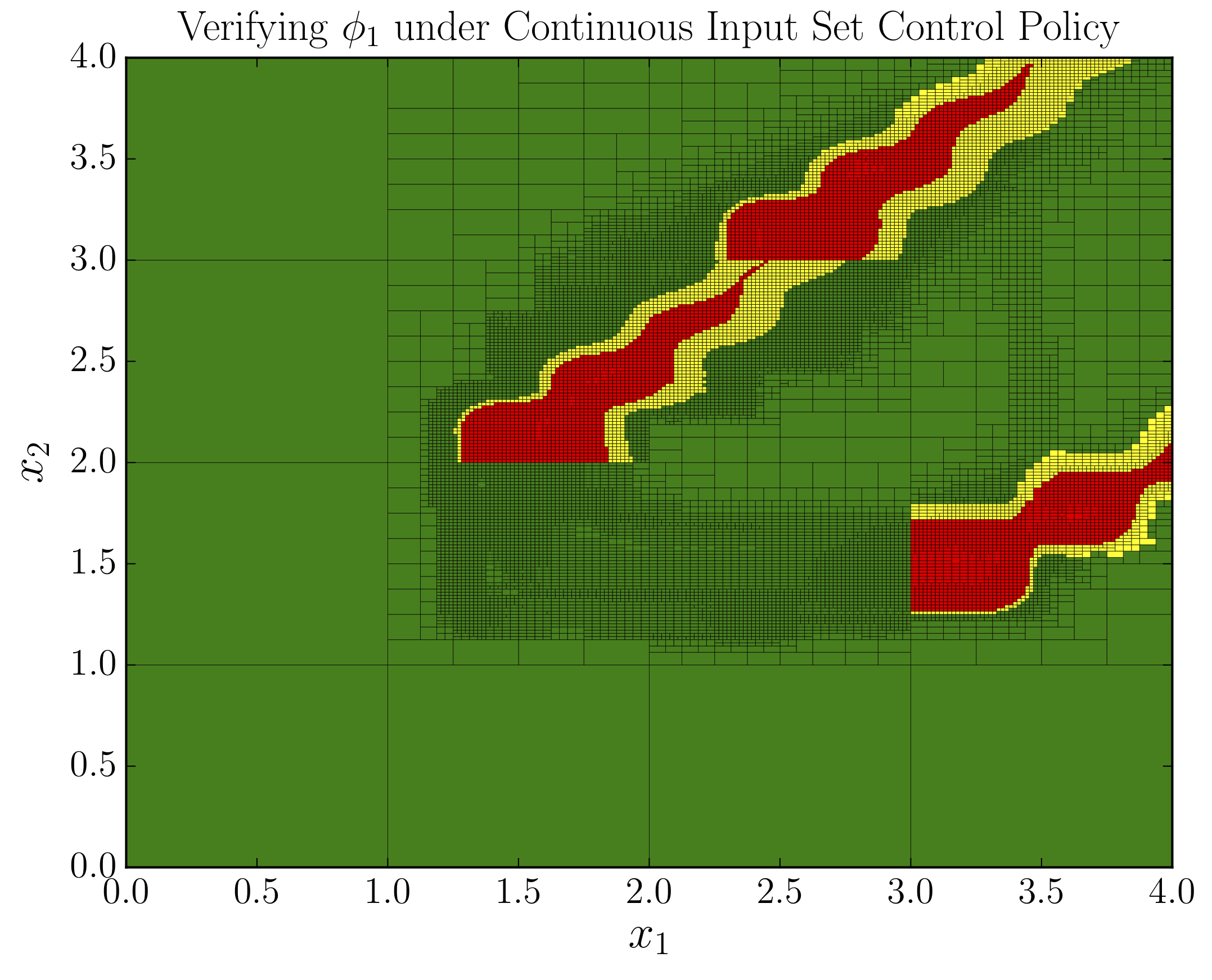}
\caption{Verification of system \eqref{eq5} against $\phi_{1}$ with respect to a satisfaction threshold of 0.8 without any input (Top), and under both the switching policy computed from the finite input set $U_{fin}$ (Bottom Left) and the control policy computed from the continuous input space $U$ (Bottom Right). The initial states in green have a probability of satisfying the specification which is greater than 0.8, the states in red have a probability which is below 0.8, and the states in yellow are undecided. The controlled versions of \eqref{eq5} convert some red regions of the state-space in the uncontrolled case to green regions.}
\label{Verifcomp}
\end{figure}

It is evident that computing controllers from a continuous set of inputs requires a more significant amount of computational effort compared to the finite input case. The largest portion of the continuous-input synthesis algorithm is expended solving the optimization problems for the value iteration step of the procedure, which is the clear scalability bottleneck of our current implementation. Moreover, we notice that the greatest suboptimality factor decreases at a slower rate as a function of refinement steps in the continuous input case than in the finite mode case, which causes a much finer partition of the domain and is the reason for the manual termination in the former example. We explain this phenomenon by observing that the suboptimality factor is more dependent on the abstraction error when using the continuous set of inputs. To see this, consider an optimal input $u^{*}$ computed for a state of the product CIMC $\mathcal{C} \otimes \mathcal{A}$, yielding an interval of satisfaction $[a, b]$ for this state. Now, consider another input $u^{*} +\epsilon$ for a small disturbance $\epsilon$. Assuming the dynamics of interest are continuous, it follows that the interval of satisfaction under the disturbed input is $[a + \epsilon_{a}, b + \epsilon_{b}]$. Therefore, the suboptimality factor for this state will be at least $b + \epsilon_{b} - a \approx b - a$, which is the size of the satisfaction interval of the considered state under the computed optimal input. Nonetheless, the algorithm still results in overall progress towards the goal optimality across all metrics as it performs more refinement steps.

\section{CONCLUSION}

In this paper, we developed abstraction-based controller synthesis techniques for stochastic systems with $\omega$-regular objectives. First, we showed a method to compute \textcolor{black}{switching policies} in stochastic systems with a finite number of modes by performing a permanent component search and a reachability maximization task \textcolor{black}{in an abstraction of the dynamics}. We proposed a specification-guided domain partition refinement scheme which targets states causing the most uncertainty in the abstraction and discards the system modes that are guaranteed to be suboptimal. We extended these results to stochastic systems with a continuous set of inputs and designed a synthesis method for the specific class of affine-in-input and affine-in-disturbance systems. Finally, we presented a numerical example where controller synthesis is conducted for both finite and continuous input sets on a nonlinear system with complex temporal logic tasks.

Future works will further explore the relationship between original partitions and their refined versions in order to reduce the number of operations performed in the components search and reachability algorithms after each refinement step and consequently improve scalability of our technique. An adaptation of these algorithms to guarantee a monotone decrease of the suboptimality factor throughout the synthesis procedure will also be investigated. \textcolor{black}{Obtaining formal convergence guarantees of the refinement heuristic is another important issue.}. Other immediate research directions include the study of wider classes of systems with continuous sets of inputs to which our abstraction-based technique can be extended.

\bibliography{NAHS1}

\newpage
\appendix

\subsection*{Proof of Lemma 1}

\mbox{}\\ We provide a constructive proof for this lemma. Consider a product BMDP $\mathcal{B} \otimes \mathcal{A}$ with set of states $Q \times S$, \textcolor{black}{set of policies $\mathcal{U}_{\otimes}^{\mathcal{A}}$ and set of memoryless policies $(\mathcal{U}_{\otimes}^{\mathcal{A}})_{mem}$}. We define the greatest permanent accepting BSCC $(U)_{P}^{G} \subseteq Q \times S$ as the set of all states of $\mathcal{B} \otimes \mathcal{A}$ such that, if \textcolor{black}{$q \in (U)_{P}^{G}$}, then there exists a policy in $\mathcal{U}_{\otimes}^{\mathcal{A}}$ such that $q$ belongs to a permanent accepting BSCC in $\mathcal{B} \otimes \mathcal{A}$.

The first part of the proof consists in showing that there exists a set of \textcolor{black}{memoryless} policies $\mathcal{U}_{(U)_{P}^{G}} \subseteq (\mathcal{U}_{\otimes}^{\mathcal{A}})_{mem}$ such that, under all product IMCs induced by a policy in $\mathcal{U}_{(U)_{P}^{G}}$, all states in $(U)_{P}^{G}$ belong to a permanent winning component simultaneously and, therefore, $(U)_{P}^{G} \subseteq (WC)_{P}^{G}$.

The second part of the proof shows that, for any other states of $\mathcal{B} \otimes \mathcal{A}$ which can be made a permanent winning component under some policy \textcolor{black}{in $\mathcal{U}_{\otimes}^{\mathcal{A}}$}, there exists a set of \textcolor{black}{memoryless} policies in $ (\mathcal{U}_{\otimes}^{\mathcal{A}})_{mem}$ (which is a subset of $\mathcal{U}_{(U)_{P}^{G}}$), such that all these states are a permanent winning component simultaneously. \\

\noindent \textbf{\textit{I] Proof of existence of \textcolor{black}{memoryless} policies generating the greatest permanent accepting BSCC as a permanent winning component}} \\

First, we constructively show that, if there exists a policy $\mu_{1} \in \mathcal{U}_{\otimes}^{\mathcal{A}}$ generating a permanent accepting BSCC $B_{1} \subseteq Q \times S$ in $(\mathcal{B} \otimes \mathcal{A})[\mu_{1}]$, and if there exists another policy $\mu_{2} \in \mathcal{U}_{\otimes}^{\mathcal{A}}$ generating a permanent accepting BSCC $B_{2} \subseteq Q \times S$ in $(\mathcal{B} \otimes \mathcal{A})[\mu_{2}]$, then there has to exist a set of \textcolor{black}{memoryless} policies in $(\mathcal{U}_{\otimes}^{\mathcal{A}})_{mem}$  causing the set $B_{1} \cup B_{2}$ to be a permanent winning component in $\mathcal{B} \otimes \mathcal{A}$. Consider a policy $\mu_{3} \in (\mathcal{U}_{\otimes}^{\mathcal{A}})_{mem}$ \textcolor{black}{constructed as follows}:\\

\noindent \textcolor{black}{1) For the states in $B_{1}$, consider the following reasoning: by virtue of $B_{1}$ being a permanent accepting BSCC for some policy, it has to hold that, for some state $q_{acc} \in B_{1}$, $ F_{i} \in L'(q_{acc})$ and $ E_{i} \not \in L'(q) \; \forall q \in B_{1}$, for some $i$. Moreover, as $B_{1}$ is a permanent BSCC under $\mu_{1}$, for any state $q \in B_{1}$, there exists a sequence of inputs chosen by $\mu_{1}$ such that the lower bound probability of reaching $q_{acc}$ is 1, that is, $\widecheck{\mathcal{P}}_{(\mathcal{B} \otimes \mathcal{A})[\mu_{1}] }(q \models \Diamond q_{acc}) = 1$. Since reachability problems in BMDPs have memoryless optimal policies \cite{haddad2018interval}, it must be true that a memoryless policy $\mu_{1}^{mem}$ choosing no other actions than the ones prescribed by $\mu_{1}$ at all states $q\in B_{1}$ and guaranteeing $\widecheck{\mathcal{P}}_{(\mathcal{B} \otimes \mathcal{A})[\mu_{1}^{mem}] }(q \models \Diamond q_{acc}) = 1$ for all $q \in B_{1}$ exists as well. For all $q \in B_{1}$, set $\mu_{3}(q) = \mu_{1}^{mem}(q)$.\\}

\noindent \textcolor{black}{ 2) For all states $ q \in B_{2} \setminus (B_{1} \cap B_{2})$, apply the same reasoning with respect to the problem of reaching $(B_{1} \cap B_{2})$ instead of $q_{acc}$, that is, there exists a memoryless policy $\mu_{2}^{mem}$ choosing no other actions than the ones prescribed by $\mu_{2}$ such that $\widecheck{\mathcal{P}}_{(\mathcal{B} \otimes \mathcal{A})[\mu_{2}^{mem}] }(q \models \Diamond (B_{1} \cap B_{2})) = 1$ for all $ q \in B_{2} \setminus (B_{1} \cap B_{2})$. For all $ q \in B_{2} \setminus (B_{1} \cap B_{2})$, set $\mu_{3}(q) = \mu_{2}^{mem}(q)$.}  \\

\noindent \textcolor{black}{3) For all states $q \in (Q \times S) \setminus (B_{1} \cup B_{2})$, choose any action in $Act(q)$ as $\mu_{3}(q)$.}\\

\noindent \textcolor{black}{As $B_{1}$ is a permanent BSCC under $\mu_{1}$, no state of $B_{1}$ can transition outside of $B_{1}$ under $\mu_{3}$, that is, it holds that $\widehat{\mathcal{P}}_{(\mathcal{B} \otimes \mathcal{A})[\mu_{3}] }(q \models \Diamond (Q \times S) \setminus B_{1}) = 0$. Moreover, since $\widecheck{\mathcal{P}}_{(\mathcal{B} \otimes \mathcal{A})[\mu_{3}] }(q \models \Diamond q_{acc}) = 1$ for all $q \in B_{1}$, it follows that any trajectory starting in $B_{1}$ will always return to $q_{acc}$, that is, $\widecheck{\mathcal{P}}_{(\mathcal{B} \otimes \mathcal{A})[\mu_{3}] }(q \models \square \Diamond q_{acc}) = 1$ for all $q \in B_{1}$, and will additionally never reach a state $q_{n-acc} \in Q \times S$ satisfying $E_{i} \in L'(q_{n-acc})$. Therefore, any trajectory starting in $B_{1}$ satisfies the Rabin acceptance condition with lower bound probability 1, and $B_{1}$ is a member of the permanent winning component of $(\mathcal{B} \otimes \mathcal{A})[\mu_{3}]$. Furthermore, for all $q \in B_{2} \setminus (B_{1} \cap B_{2})$, we have $ \widecheck{\mathcal{P}}_{(\mathcal{B} \otimes \mathcal{A})[\mu_{3}] }(q \models \Diamond B_{1} )  =  \widecheck{\mathcal{P}}_{(\mathcal{B} \otimes \mathcal{A})[\mu_{3}] }(q \models \Diamond (B_{1} \cap B_{2})) = 1$ and thus, $B_{2} \setminus (B_{1} \cap B_{2})$ is a member of the permanent winning component of $(\mathcal{B} \otimes \mathcal{A})[\mu_{3}]$. Therefore, $B_{1} \cup B_{2}$ is a member of the permanent winning component of $(\mathcal{B} \otimes \mathcal{A})[\mu_{3}]$}.

Iteratively applying this logic with $B_{1} \cup B_{2}$ and any other member of $(U)_{P}^{G}$ shows that there exists a set of policies in $\mathcal{U}_{(U)_{P}^{G}} \subseteq (\mathcal{U}_{\otimes}^{\mathcal{A}})_{mem}$ such that all states in $(U)_{P}^{G}$ belong to a permanent winning component simultaneously.\\

\noindent \textbf{\textit{II] Proof of existence of greatest permanent winning component \textcolor{black}{and of memoryless policies generating this component}}}\\

Now, we consider the set $R = (Q \times S) \setminus \textcolor{black}{(U)_{P}^{G}}$ of all states of $\mathcal{B} \otimes \mathcal{A}$ which do not belong to \textcolor{black}{$(U)_{P}^{G}$}. 

\textcolor{black}{We define the set $\mathcal{U}^{out}_{(U)_{P}^{G}}$ of all policies which are history-dependent outside of $(U)_{P}^{G}$ and generate $(U)_{P}^{G}$ with an (arbitrary) memoryless policy on the states in $(U)_{P}^{G}$.}

For a policy $\mu \in \textcolor{black}{\mathcal{U}^{out}_{(U)_{P}^{G}}}$, the set of all states $C \subseteq R$ that belong to the permanent winning component $(WC)_{P}$ of $(\mathcal{B} \otimes \mathcal{A})[\mu]$ without being a member of $(U)_{P}^{G}$ --- that is, $C \cup (U)_{P}^{G} = (WC)_{P}$ and $C \cap (U)_{P}^{G} = \emptyset$ --- has to satisfy two conditions:\\

\noindent a) $C$ does not allow a transition outside of $C \cup (U)_{P}^{G}$ under any adversary of $(\mathcal{B} \otimes \mathcal{A})[\mu]$, that is, $\widehat{\mathcal{P}}_{(\mathcal{B}\otimes \mathcal{A})[\mu] }\bigg(q \models \Diamond \Big((Q \times S) \setminus \big(C \cup (U)_{P}^{G}\big)\Big) \bigg) = 0$ for all $q \in C$,\\

\noindent b) No subset of $C$ can form a losing component under any adversary of $(\mathcal{B} \otimes \mathcal{A})[\mu]$, that is, no state in $C$ is a member of the largest losing component $(LC)_{L}$ of the product IMC $(\mathcal{B} \otimes \mathcal{A})[\mu]$, or $C \cap (LC)_{L} = \emptyset$. \\

\noindent With these two conditions fulfilled, all states in $C$ either transition to $(U)_{P}^{G}$ or reach an accepting BSCC formed within $C$ under all adversaries of $(\mathcal{B} \otimes \mathcal{A})[\mu]$, and therefore reach an accepting BSCC with lower bound probability 1.\\

Now, we constructively show that, if there exists a policy $\mu_{1} \in \textcolor{black}{\mathcal{U}^{out}_{(U)_{P}^{G}}}$ inducing a product IMC $(\mathcal{B} \otimes \mathcal{A})[\mu_{1}]$ with permanent winning component $(WC^{1})_{P}$ and with a set of states $C_1 \in R$ satisfying conditions a) and b) such that $C_{1} \cup (U)_{P}^{G} = (WC^{1})_{P}$ and $C_{1} \cap (U)_{P}^{G} = \emptyset$, and if there exists a policy $\mu_{2} \in \mathcal{U}^{out}_{(U)_{P}^{G}}$ inducing a product IMC $(\mathcal{B} \otimes \mathcal{A})[\mu_{2}]$ with permanent winning component $(WC^{2})_{P}$ and with a set of states $C_2 \in R$ satisfying conditions a) and b) such that $C_{2} \cup (U)_{P}^{G} = (WC^{2})_{P}$ and $C_{2} \cap (U)_{P}^{G} = \emptyset$, then there has to exist a \textcolor{black}{memoryless} policy $\mu_{3} \in \mathcal{U}_{(U)_{P}^{G}}$ inducing a product IMC $(\mathcal{B} \otimes \mathcal{A})[\mu_{3}]$ with permanent winning component $(WC^{3})_{P}$ and with the set of states $(C_1 \cup C_2) \in R$ satisfying conditions a) and b) such that $(C_1 \cup C_2) \cap (U)_{P}^{G} = \emptyset$. Consider a policy $\mu_{3} \in \mathcal{U}_{(U)_{P}^{G}}$ \textcolor{black}{constructed as follows}:\\

\noindent \textcolor{black}{
 1) For all state $q \in C_{1}$, consider the following reasoning inspired by the arguments in the proof of \cite[Theorem 8]{chatterjee2004trading} on the optimality of memoryless policies in MDPs for Rabin objectives: for any state in $q \in C_{1}$, it must be true that any trajectory initiated at $q$ under policy $\mu_{1}$ reaches with lower bound probability 1 a set $K \subseteq C_{1}$ such that the continuation of any trajectory that reaches $K$ is confined to $K \cup (U)_{P}^{G}$ and either reaches $(U)_{P}^{G}$ or visits an unmatched accepting Rabin state in $K$ infinitely often.  Consider the arbitrarily ordered set $(K_{1}, K_{2}, \ldots, K_{m})$ of all such sets which can by reached by some initial state $q \in C_{1}$ under $\mu_{1}$. Due to the optimality of memoryless policies for reachability problems in BMDPs and the properties of the $K$ sets, there must exist a memoryless policy $\mu^{K_{1}}_{1}$ such that, for all $q \in K_{1}$, $\widecheck{\mathcal{P}}_{(\mathcal{B}\otimes \mathcal{A})[\mu^{K_{1}}_{1}] }\Big(q \models \Diamond  ((U)_{P}^{G}) \cup A \Big) = 1$, where $A$ is the set of all unmatched Rabin accepting states in $K_{1}$, and $\widehat{\mathcal{P}}_{(\mathcal{B}\otimes \mathcal{A})[\mu^{K_{1}}_{1}] }\Big(q \models \Diamond (Q \times S) \setminus ((U)_{P}^{G} \cup K_{1}) \Big) = 0$. Set $\mu_{3}(q) = \mu^{K_{1}}_{1}(q)$ for all $q \in K_{1}$. Apply the same procedure recursively to $K_{2} \setminus K_{1}$ and replacing $(U)_{P}^{G}$ with $(U)_{P}^{G} \cup K_{1}$, then to $K_{3} \setminus (K_{2} \cup K_{1})$ etc. For the states $q \in C_{1}$ outside the $K$ sets, design $\mu_{3}$ such that  $\widecheck{\mathcal{P}}_{(\mathcal{B}\otimes \mathcal{A})[\mu_{3}] }\Big(q \models \Diamond  (U)_{P}^{G} \cup K_{1} \cup \ldots \cup K_{m} \Big) = 1$, which again can be achieved with a memoryless choice of actions due to the optimality of memoryless policies for reachability and the fact that $\mu_{1}$ satisfies this condition.}  \\

\noindent \textcolor{black}{
 2) For all state $q \in C_{2} \setminus (C_{1} \cap C_{2})$, choose the actions in $\mu_{3}$ by following the same reasoning as in 1) after replacing $C_{1}$ with $q \in C_{2} \setminus (C_{1} \cap C_{2})$ and $(U)_{P}^{G}$ with $(U)_{P}^{G} \cup C_{1}$.}\\

\noindent \textcolor{black}{
 3) For all state $q \in (Q \times S) \setminus (C_{1} \cup C_{2})$ (not in $(U)_{P}^{G}$, since actions are already fixed in this set), choose any action in $Act(q)$ as $\mu_{3}(q)$.\\}

\noindent \textcolor{black}{By construction, the set $C_1 \cup C_2$ satisfy condition a) and b), as no subset of $C_1 \cup C_2$ can form a losing component under the actions prescribed by $\mu_{3}$ and no trajectory can leave $C_1 \cup C_2 \cup (U)_{P}^{G}$.} Therefore, $C_1 \cup C_2$ is a subset of the permanent winning component $(WC^{3})_{P}$ of $(\mathcal{B} \otimes \mathcal{A})[\mu_{3}]$. 

\textcolor{black}{Replacing the set $(U)^{G}_{P}$ from the beginning of section II] with $C_{1} \cup C_{2} \cup (U)^{G}_{P}$ and applying the same process iteratively proves the existence of a set $(WC)_{P}^{G}$ satisfying the properties enunciated in the lemma and of a set of \textcolor{black}{memoryless} policies $\mathcal{U}_{(WC)^{G}_{P}}$ generating $(WC)_{P}^{G}$.}\\

\end{document}